%% file: acfk.tex
\documentclass[a4paper,journal, 10pt, romanappendices]{IEEEtran}
\usepackage{balance}
\usepackage{graphicx}
\usepackage{bm}
\usepackage{amsmath}
\usepackage{amsfonts}
\usepackage{amssymb}
\usepackage{color}
\usepackage{mathrsfs}
\usepackage{upgreek}
\usepackage[utf8]{inputenc}
\usepackage{tikz}
\usetikzlibrary{arrows, calc, decorations.markings, positioning}
\usepackage{cite}
\usepackage{booktabs}
\usepackage{multirow}
\usepackage{algorithm}
\usepackage{algorithmic}
\usepackage{makecell}
\usepackage[flushleft]{threeparttable}
\usepackage{tabularx}
\usepackage{stfloats}
 \usepackage[caption=false,font=footnotesize]{subfig}
\input{SupportDocuments/YFPreamble.tex}

\usepackage{hyperref}					
\usepackage{winsnotation}
\usepackage{flushend}
\hypersetup{colorlinks,
	linkcolor=purple,%
	anchorcolor=blue,
    citecolor=purple}

\newtheorem{definition}{Definition}    
\newtheorem{lemma}{Lemma}     
\newtheorem{proposition}{Proposition} 
\newtheorem{corollary}{Corollary} 
\newtheorem{remark}{Remark}    

\newcounter{MYtempeqncnt}

\title{Auto-correlation Function Keying}
\author{Weijiang Zhao, Yifeng Xiong, \IEEEmembership{Member, IEEE}, Fan Liu, \IEEEmembership{Senior Member, IEEE}, Shuangyang Li, \IEEEmembership{Member, IEEE}, Shi Jin, \IEEEmembership{Fellow, IEEE}, and Jianhua Zhang, \IEEEmembership{Fellow, IEEE}
\thanks{This work was supported in part by the National Natural Science Foundation of China (NSFC) under Grant U25B2009, Grant 62522107, and Grant 62331023. (\textit{Corresponding authors: Yifeng Xiong; Fan Liu.})

W. Zhao and Y. Xiong are with the Department of Communication Engineering, Beijing University of Posts and Telecommunications, Beijing 100876, China (e-mail: zhaoweijiang@bupt.edu.cn, yifengxiong@bupt.edu.cn).

F. Liu and S. Jin are with the National Mobile Communications Research Laboratory, Southeast University, Nanjing 210096, China (email: fan.liu@seu.edu.cn, jinshi@seu.edu.cn).

Shuangyang Li is with the Faculty of Electrical Engineering and Computer Science, Technical University of Berlin, 10587 Berlin, Germany (e-mail: shuangyang.li@tu-berlin.de).

Jianhua Zhang is with the State Key Laboratory of Networking and Switching Technology, Beijing University of Posts and Telecommunications, Beijing 100876, China (e-mail: jhzhang@bupt.edu.cn).
}}

\begin{document}

\markboth{}{}

\maketitle

\begin{abstract}
We propose ACFK: Auto-correlation Function Keying, a new \ac{isac} waveform that carries random communication data while directly controlling the \ac{psl} of the \ac{pacf}. In contrast to existing works aiming at controlling the \ac{esl}, which fails to characterize realization-specific sidelobe behaviors, we formulate a mutual information maximization problem under PSL and power constraints, and show that a continuous ACF-domain uniform distribution is asymptotically optimal at high \ac{snr} over quasi-static frequency-flat channels.

Motivated by this principle, ACFK maps finite-constellation symbols onto \ac{acf}-domain sidelobes and uses independent phase symbols to exploit the remaining degrees of freedom. The resulting waveform enables exact control of the nominal \ac{pacf}, which coincides with the actual \ac{pacf} when the power spectral non-negativity condition is satisfied. We further analyze the non-negativity violation probability and bound the corresponding \ac{pslr} degradation. A reference \ac{isac} transceiver and its high-\ac{snr} approximate \ac{ber} analysis are also provided. Numerical results show that \acs{acfk} achieves stronger \ac{pslr} control, and improved weak-target detection performance, than a generalized \ac{pas} baseline at similar data rate and \ac{ber}.
\end{abstract}


%
\IEEEpeerreviewmaketitle

\section{Introduction}
\IEEEPARstart{E}{nvisioned} as a transformative paradigm, \acf{isac} has been formally designated by the International Telecommunication Union as one of the six core usage scenarios for \ac{6g} wireless networks \cite{ITU2023}. By coordinating wireless resources such as time, frequency, beam, and power across sensing and communication functionalities over a unified hardware platform, \ac{isac} effectively empowers numerous emerging applications, including autonomous driving, digital twins and low-altitude economy \cite{Chafii2023CST, saad2019vision, 9606831}. One of the fundamental challenges in \ac{isac} lies in designing high-performance dual-functional waveforms that are capable of supporting both accurate target sensing and reliable information delivery tasks in a resource-efficient manner \cite{9737357,9540344,10770016}. Due to low implementation costs and seamless compatibility with legacy cellular infrastructure, the communication-centric designs are viewed as a promising paradigm for future \ac{6g} \ac{isac} networks \cite{10012421}.

In current communication-centric \ac{isac} signaling frameworks, sensing functionality largely relies on reference signals \cite{9921271,10561589,9746355}, such as \ac{csirs} and \ac{dmrs}, which are well-designed deterministic sequences with desirable correlation properties. However, such reference signals typically account for only around $10\%$ of the total time-frequency resources \cite{3gpp_ts38211_2022}, which is insufficient to fulfill the increasing requirements for high-precision sensing. To improve the sensing performance and resource efficiency, a prospective solution is to reuse the communication data payloads, which occupy the remaining $90\%$ of time-frequency resources, for sensing tasks. Against this background, one of the fundamental performance-limiting mechanisms is identified as the \ac{drt} \cite{10147248,10471902,liu2023deterministic}. Specifically, \ac{drt} reveals that deterministic signals are beneficial to achieve reliable sensing performance, whereas the intrinsic randomness of data payloads causes degradation in sensing performance. Therefore, it is crucial to evaluate the performance of sensing with data payload signals and develop techniques for adjusting the \ac{drt}.

Motivated by the \ac{drt}, sensing with random communication payloads has recently attracted considerable attention \cite{11206742, 11087656, iceberg, du2024reshaping, zhao2026input, 11417988, 11204821, 11322556, zhang2025discrete}. From a theoretical perspective, the central problem is to characterize the joint achievable region of communication rate and a suitably defined sensing performance metric \cite{9785593, 9731514, 11203006, 10935713}. For the specific task of target detection and ranging, the sidelobes of the \acf{acf} provide a direct measure of range-domain inter-target interference: high sidelobes generated by a strong target may mask nearby weak targets and introduce ambiguity in delay estimation \cite{richards2005fundamentals}. Since payload-bearing waveforms are random, their \acfp{pacf} are also random. This has motivated most existing studies to adopt the \acf{esl}, an average measure of the random \ac{pacf} sidelobe energy, as the sensing metric. In this context, a closed-form expression of the average squared \ac{pacf} of random communication signals is derived in \cite{11087656}, which reveals that the \ac{esl} is related to the orthonormal modulation bases and the kurtosis of constellations. It was furthen proven that, among all orthogonal linear modulation schemes, \ac{ofdm} achieves the lowest average ranging sidelobe level for \ac{iid} \ac{qam} or \ac{psk} symbols. In terms of \ac{drt}-adjusting techniques, \cite{du2024reshaping} proposed a \ac{pcs} approach for \ac{ofdm} systems over \ac{awgn} channels to strike a scalable tradeoff between the \ac{esl} and achievable rate by adjusting the value of kurtosis of constellations. Overall, existing studies have mainly established theoretical analyses \cite{10147248, 11087656, zhang2025discrete} and \ac{drt}-adjusting techniques around \ac{esl} control \cite{11206742, 11087656, iceberg, du2024reshaping, zhao2026input}.

Although the \ac{esl} characterizes the average behavior of the random \ac{pacf}, it cannot fully reflect variability in the sidelobe level \cite{11087656}. In particular, even under identical \ac{esl} values different single \ac{pacf} realizations may exhibit severe sidelobe level fluctuations, which may markedly affect target detection stability. To overcome this limitation, one could alternatively control the \acf{psl} of random \acp{pacf}. To elaborate, the \ac{psl} is a vital metric for ensuring stable and robust weak-target detection performance, since it describes the worst-case sidelobe level and reflects the strongest sidelobe interference imposed on the weaker targets \cite{richards2005fundamentals}.

The precise control of the \ac{psl} of deterministic or pseudo-random signals has been extensively investigated in the sequence design literature \cite{1055219, 1056116, fan1999class, tang2000lower, 4609004, 6484164, 7362231, 7967829}. Sequences with desirable periodic auto-correlation and cross-correlation properties play an important role in synchronization \cite{11367063}, channel estimation \cite{11320243}, interference mitigation \cite{6503922} and radar sensing \cite{6237581}. An ideal sequence is expected to feature both impulse-like auto-correlation functions and zero cross-correlation functions in a sequence set. However, given the sequence period and the size of sequence sets, the sequence design involves a fundamental tradeoff between the maximal auto-correlation and cross-correlation magnitudes, as given by the Welch bound \cite{1055219} and the Sarwate bound \cite{1056116}. In light of this, \cite{fan1999class} proposed the concepts of \ac{zcz} and \ac{lcz}. \ac{zcz}$/$\ac{lcz} sequences achieve ``locally''  ideal or nearly ideal correlation properties only at correlation lags within a specified window. The cardinality of \ac{zcz}$/$\ac{lcz} sequence sets is strictly limited by the Tang-Fan-Matsufuji bound \cite{tang2000lower}.

Although a large body of sequence construction methods have been proposed in \cite{4609004, 6484164, 7362231, 7967829, 11367063, 11320243}, the quantity of pseudo-random sequences with favorable correlation properties is considerably limited compared with the communication codewords generated by conventional modulation schemes. Thus it is infeasible to directly use these sequence sets as communication codebooks in practical communication systems. This further highlights the compelling need for precise control of the \ac{psl} of the \ac{pacf} of random communication signals in \ac{isac} systems. In particular, this gives rise to a constrained information-bearing waveform design problem: the transmitted payload should preserve a high communication rate while enforcing stringent constraints on the random \ac{pacf} sidelobes.

Therefore, in this paper, rather than minimizing \ac{psl} for a deterministic sequence, we seek input distributions and modulation structures that maximize mutual information under \ac{psl} and power constraints. Motivated by the structure of high-\ac{snr} capacity-achieving input distribution, we propose a payload-bearing modulation architecture, termed as \emph{\acf{acfk}}. Specifically, by directly embedding data symbols onto the \ac{acf}-domain sidelobes, \ac{acfk} enables exact control of the nominal \ac{pacf} and high-probability control of the actual \ac{pacf}. For clarity, our main contributions are summarized as follows:
\begin{itemize}
\item We first formulate our signal design problem as a mutual information maximization problem under the power budget constraint as well as the \ac{psl} constraints. We then characterize the high-\ac{snr} asymptotically capacity-achieving input distributions under \ac{psl} constraints over quasi-static channels. In particular, the continuous \ac{acf}-domain uniform construction provides an asymptotically optimal high-\ac{snr} design principle under \ac{psl} constraints over quasi-static frequency-flat channels.
\item We propose \ac{acfk}, which can be viewed as a structured implementation of this asymptotically optimal high-\ac{snr} design principle in practical finite-constellation systems. From the sensing perspective, although unable to provide precise control of the \ac{pacf} due to potential violation of a power spectral non-negativity constraint, \ac{acfk} is capable of achieving a high-probability control of the \ac{pacf}. Regarding the impact of non-negativity violation, we derive bounds for the non-negativity violation probability and \acf{pslr}. In addition, we further extend our scheme to coherent processing over multiple channel uses for a controllable \ac{psl} reduction.
\item We provide a reference design for \ac{isac} transceivers employing \ac{acfk}. From a practical perspective, \ac{acfk} can be implemented by modifying and amalgamating the existing \ac{ofdm} and \ac{sc} modulation schemes. Under this transceiver design, we derive analytical expressions for the approximate \ac{ber} in the high-\ac{snr} regime over quasi-static channels.
\end{itemize}

The remainder of this paper is organized as follows. Sec.~\ref{sec:model} presents the system model and sensing metrics. Sec.~\ref{sec:highsnr_input} derives the high-\ac{snr} capacity-achieving input distributions under \ac{psl} and power budget constraints. Sec.~\ref{sec:acfk} introduces \ac{acfk} and analyzes its nominal and actual \ac{pacf}, with an \ac{afk} extension for multiple channel uses. Sec.~\ref{sec:architecture} presents the reference \ac{acfk} transceiver design. Sec.~\ref{sec:performance} analyzes the \ac{ber} of the reference receiver. Sec.~\ref{sec:simulations} provides numerical results, and Sec.~\ref{sec:conclusion} concludes the paper.

\textit{Notations:} Throughout this paper, $\rv{a}$, $\RV{a}$, $\RM{A}$ and $\RS{A}$ represent random variables (scalars), random vectors, random matrices and random sets, respectively; their realizations, or the corresponding deterministic quantities, are denoted by $a$, $\V{a}$, $\M{A}$, and $\Set{A}$, respectively. The $m$-dimensional vector of zeros (resp. ones) is denoted by $\V{0}_{m}$ (resp. $\V{1}_{m}$). The $m$-by-$m$ identity matrix is denoted by $\M{I}_{m}$. These subscripts are omitted if they are clear from the context. The notation $[\M{A}]_{i,j}$ denotes the $(i,j)$-th entry of matrix $\M{A}$. $[\V{a}]_{i}$ denotes the $i$-th entry of vector $\V{a}$. $[\V{x}; \V{y}]$ and $[\V{x}, \V{y}]$ denote vertical and horizontal concatenation of vectors, respectively. The Hadamard product between matrices (or vectors) $\M{A}$ and $\M{B}$ is denoted by $\M{A}\odot\M{B}$. We denote by $\V{a} \circledast \V{b}$  the circular convolution between vectors $\V{a}$ and $\V{b}$. The notation ${\rm circ}(\V{a})$ denotes the circulant matrix generated by $\V{a}$. The notation $|\V{x}|$ denotes the vector containing the entrywise magnitude of $\V{x}$, while $|\V{x}|^2$ and $\sqrt{\V{x}}$ denote the entrywise squared magnitudes and square roots of magnitudes of $\V{x}$, respectively. ${\rm{diag}}(\V{x})$ denotes the diagonal matrix constructed by placing $\V{x}$ on its main diagonal. $|\RS{A}|$ denotes the cardinality of the set $\RS{A}$. $\langle \cdot \rangle_N$ denotes the modulo $N$ operation. $\Set{C}(\Set{S})$ denotes the space of continuous functions over $\Set{S}$, and $\mathbb{I}_{\Set{S}}(\cdot)$ denotes the indicator function of set $\Set{S}$, which takes the value of $1$ over $\Set{S}$ while taking the value of zero outside $\Set{S}$. $\mathbb{T}^n$ denotes the $n$-dimensional torus.

\section{System Model} \label{sec:model}
In this paper, we consider target range estimation under a monostatic \ac{isac} system. The \ac{isac} \ac{tx} emits an  \ac{isac} signal modulated with random communication symbols, which is captured at a communication \ac{rx}, and is simultaneously reflected back to the sensing \ac{rx} by various targets at different ranges. The sensing \ac{rx} collocated with the \ac{isac} \ac{tx} performs \ac{mf} to detect these targets and estimate their delay parameters, based on the fully known random \ac{isac} signal.

\subsection{Channel Models}
We model both the communication and sensing channels as quasi-static multipath channels containing $I$ paths and $Q$ targets, respectively. Their time-domain impulse responses can be expressed as follows
\begin{equation}
h_{\rm c}(\tau) = \sum_{i=1}^I\alpha_i\delta(\tau-\tau_{c,i}),\quad h_{\rm s}(\tau) = \sum_{q=1}^Q\gamma_q\delta(\tau-\tau_{s,q}),
\end{equation}
where $\delta(\tau)$ is the Dirac delta function, $\alpha_i$ and $\gamma_q$ denote the complex amplitudes of the $i$-th communication path and the $q$-th sensing target, respectively, with $\tau_{c,i}$ and $\tau_{s,q}$ representing the corresponding time delay.

Then, we investigate an \ac{isac} system with $N$ subcarriers, and focus on the transmission of a single data stream. A \ac{cp} is added to the \ac{isac} signal for eliminating \ac{isi} caused by multipath effect. Without loss of generality, we assume that the length of \ac{cp} is larger than the maximum delay of the communication paths and sensing targets in this paper. After the \ac{cp} removal at the receiver, the frequency-domain representations of the communication and sensing channels can be expressed as
\begin{subequations}
\begin{align}
\RV{y}_{\rm{c}} &= \V{h}_{\rm{c}} \odot \RV{x} + \RV{n}_{\rm{c}}, \label{com} \\
\RV{y}_{\rm{s}} &= \V{h}_{\rm{s}} \odot \RV{x} + \RV{n}_{\rm{s}}, \label{sen}
\end{align}
\end{subequations}
where $\RV{x}\in\mathbb{C}^N $ denotes the transmitted dual-functional waveform emitted from the \ac{isac} \ac{tx} for performing both communication and sensing tasks. $\RV{y}_{\rm{c}}\in\mathbb{C}^N$ and $\RV{y}_{\rm{s}}\in\mathbb{C}^N$ denote the received communication and sensing signals, respectively. $\V{h}_{\rm{c}}\in\mathbb{C}^N$ and $\V{h}_{\rm{s}}\in\mathbb{C}^N$ denote the frequency-domain representations of the communication channel and the target response vector, which correspond to the \ac{dft} of the time-domain impulse responses $h_{\rm{c}}(\tau)$ and $h_{\rm{s}}(\tau)$, respectively. We consider a block fading model for both the communication channel and the target response vector. Specifically, it is assumed that $\V{h}_{\rm{c}}$ and $\V{h}_{\rm{s}}$ remain constant within a transmission block and vary in an \ac{iid} manner across different transmission blocks. We assume that the communication channel $\V{h}_{\rm{c}}$ is perfectly known due to channel estimation conducted prior to the transmission signal design. $\RV{n}_{\rm{c}}\in\mathbb{C}^N $ denotes the received communication noise, modelled as \ac{iid} circularly symmetric complex Gaussian distributed random variables with zero mean and variance $\sigma_c^2$, namely $\RV{n}_{\rm c}\sim\mathcal{CN}(\V{0},\sigma_c^2\M{I}_{N})$. Similarly, we also model the sensing noise $\RV{n}_{\rm s}\in\mathbb{C}^N$ as $\RV{n}_{\rm s}\sim\mathcal{CN}$($\V{0}, \sigma_s^2 \M{I}_{N}$).

\subsection{Sensing Performance Metrics}
As for the sensing performance metric, it is well-known that frequency-domain wireless resources are related to ranging accuracy \cite{levanon2004radar}. In light of this, we employ the normalized \ac{pacf} of the time-domain transmitted signal, which is also referred to as the zero-Doppler slice of \ac{af}, to characterize the sensing performance, in particular for the matched filtering process at the sensing \ac{rx}. By leveraging the Wiener-Khinchin theorem, we can express the \ac{pacf} as
\begin{equation}\label{pacf_def}
\RV{r}_{\RV{x}} = \frac{1}{\sqrt{N}}\M{F}_N^{\rm{H}}|\RV{x}|^2,
\end{equation}
where $\M{F}_N=\left[\V{f}_1,\V{f}_2,\cdots,\V{f}_N\right]$ denotes the $N$-point unitary \ac{dft} matrix.

Due to the intrinsic randomness of communication symbols, the \ac{pacf} is random, and hence it is insufficient to evaluate sensing performance only based on a single realization of the data sequence. In this context, \ac{esl} has been investigated in the literature \cite{11206742, 11087656,iceberg}. A lower \ac{esl} implies improved average resolvability of weak targets in the presence of interference from neighboring stronger targets \cite{11087656}. However, \ac{esl} itself cannot fully capture realization-specific sidelobe behaviors. Even under identical \ac{esl} values, different \ac{pacf} realizations may exhibit significantly distinct sidelobe level patterns, which may ultimately affect target detection stability.

To address this limitation, we propose to use \ac{psl} to evaluate sensing performance. \ac{psl} characterizes the worst-case sidelobe level, which is particularly informative for the widely applied class of \ac{cfar} detectors. To elaborate, lowering the \ac{psl} of the \ac{pacf} would reduce the occurrence probability of spurious peaks, thereby reducing the false alarm probability under the same detection probability.

Upon denoting the set of all the sidelobe bins of \ac{pacf} as $\Set{S}_{\rm ACF}=\{2,3,\dotsc,N\}$, the aforementioned sensing metrics are expressed as follows
\begin{subequations}\label{define_psl}
\begin{align}
[\ac{esl}]_i &= \mathbb{E}\left\{|[\RV{r}_{\RV{x}}]_i|^2\right\},\quad \forall i \in \Set{S}_{\rm ACF}, \\
\ac{psl} &= \max_{i \in \Set{S}_{\rm ACF}}\left\{|[\RV{r}_{\RV{x}}]_i|^2\right\}. \label{def_psl}
\end{align}
\end{subequations}
Since the mainlobe level is in general not constant, the \ac{pslr} given by
\begin{equation}
\ac{pslr} = \frac{\max_{i \in \Set{S}_{\rm ACF}}\left\{|[\RV{r}_{\RV{x}}]_i|^2\right\}}{|[\RV{r}_{\RV{x}}]_1|^2}
\end{equation}
could serve as a normalized version of \ac{psl}, and hence will also be applied in the following sections.

Before delving into detailed analyses, we first provide the intuition behind the proposed modulation scheme. According to the definition of \ac{psl} in \eqref{def_psl}, controlling the \ac{psl} essentially requires controlling the support of the \ac{pacf}. This observation motivates us to directly modulate communication symbols onto the \ac{pacf}. In principle, such a modulation scheme can be realized through the ``inverse mapping'' associated with \eqref{pacf_def}, given by
\begin{equation}
\RV{x}=\sqrt{\sqrt{N}\M{F}_N\RV{r}_{\RV{x}}}.
\end{equation}
However, this modulation strategy leaves several issues to be addressed. First, strictly speaking, \eqref{pacf_def} does not define an invertible mapping, since the magnitude operation discards the phase degrees of freedom. As a result, the above ``inverse mapping'' only exploits the magnitude of the signal representation in the frequency domain, while leaving the phase degrees of freedom unused. How this \ac{pacf}-domain modulation should be combined with modulation over the phase degrees of freedom remains unclear. Second, this strategy can only ensure that the \ac{psl} is controllable, but does not necessarily guarantee an optimal sensing-communication performance tradeoff. Finally, it is also unclear how a communication bit stream should be explicitly mapped onto the \ac{pacf} vector.

To address the above issues, in the next section we first investigate the capacity-achieving input distribution under the \ac{psl} constraint. We will show that, under certain conditions, when amalgamated with phase modulation, embedding communication information into the \ac{pacf} can indeed asymptotically achieve the optimal \ac{psl}-rate tradeoff in the high-\ac{snr} regime.


\section{High-\ac{snr} Capacity-Achieving Input Distributions}\label{sec:highsnr_input}
In this section, we aim to analyze and derive input distributions that asymptotically achieve the capacity in the high-\ac{snr} regime, under the \ac{psl} constraints as well as the mainlobe level constraint, formulated as
\begin{subequations}\label{acfkcapacity}
\begin{align}
\max_{p_{\RV{x}}(\V{x})}&~~{I}(\RV{x};\RV{y}_{\rm{c}})\\
{\rm s.t.}&~~{\left|\frac{1}{\sqrt{N}}\V{f}_i^{\rm H}|\RV{x}|^2\right|^2} \leq \frac{1}{N\zeta_{\rm ACF}}, \quad \forall i \in \Set{S}_{\rm ACF}. \\
&~~{\left|\frac{1}{\sqrt{N}}\V{f}_1^{\rm H}|\RV{x}|^2\right|^2} = 1, \label{mainlobe_constraint}
\end{align}
\end{subequations}
where the mainlobe level constraint \eqref{mainlobe_constraint} can also be viewed as an energy budget constraint for all legitimate codewords, and $\zeta_{\rm ACF}$ denotes a positive sidelobe attenuation factor.

In general, the optimal input distribution of problem \eqref{acfkcapacity} can be solved numerically using the \ac{ba} algorithm \cite{cover1999elements}. However, the computational complexity of the \ac{ba} algorithm can be prohibitive as the number of subcarriers $N$ grows. To this end, and in seeking of a more systematic construction of the input distribution, in this treatise we confine ourselves to the high-\ac{snr} regime, namely in the limit of $\sigma_{\rm c}^2\rightarrow 0$. Let us denote the power spectrum and the phases over the subcarriers as
\begin{equation}
\RV{p}=\begin{bmatrix}
\rv{p}_1,\dotsc,\rv{p}_N
\end{bmatrix}^{\rm T},~\RV{\theta}=\begin{bmatrix}
\rv{\theta}_1,\dotsc,\rv{\theta}_N
\end{bmatrix}^{\rm T},
\end{equation}
respectively, where
\begin{equation}\label{transform_coordinate}
\rv{p}_i=|\rv{x}_i|^2,~\rv{x}_i=\sqrt{\rv{p}_i} e^{j\rv{\theta}_i},
\end{equation}
we may rewrite the problem \eqref{acfkcapacity} under the transformed coordinate $\RV{u}=\begin{bmatrix}
(\RV{p}^\prime)^{\rm T},~\RV{\theta}^{\rm T}
\end{bmatrix}^{\rm T}\in\mathbb{R}^{2N-1}$, where $\RV{p}^\prime =[\rv{p}_1,\dotsc,\rv{p}_{N-1}]^{\rm T}$, through the transformation
\begin{equation}
\RV{x}=\Psi(\RV{u})
\end{equation}
characterized by \eqref{transform_coordinate}. The reason that we remove $\rv{p}_N$ from the transformed coordinate is that it is determined by other entries in $\RV{p}$ as follows
\begin{equation}
\rv{p}_N = {N} - \sum_{i=1}^{N-1} \rv{p}_i,
\end{equation}
according to the energy constraint \eqref{mainlobe_constraint}. The feasible region of $\RV{p}^\prime$ is given by
\begin{align}
\Set{P} &= \bigl\{\RV{p}^\prime\in\mathbb{R}^{N-1}:~\rv{p}_i\geq 0,~i=1\dotsc N, \nonumber \\
&\hspace{3mm}|\V{f}_k^{\rm H}\RV{p}|^2\leq \zeta_{\rm ACF}^{-1},~k\in\Set{S}_{\rm ACF} \bigr\}.
\end{align}
The global feasible region of $\V{u}$ is thus $\Set{M}=\Set{P}\times \mathbb{T}^N$. In the high-\ac{snr} regime we thus have the following result.
\begin{lemma}[Principal terms of mutual information]\label{lem:highsnr}
Assume that $[\V{h}_{\rm c}]_i\neq 0,~\forall i$. For continuous input density $p_{\RV{u}}(\V{u})$ having full support over $\Set{M}$ in the sense that
\begin{align}
p_{\RV{u}}(\V{u}) &= g(\V{u})\mathbb{I}_{\Set{M}}(\V{u}),~g_{\min}\leq g(\V{u})\leq g_{\max}, \nonumber \\
0&<g_{\min}\leq g_{\max}<\infty,~g(\V{u})\in C(\Set{M}), \label{conditions_input}
\end{align}
the high-\ac{snr} (i.e., $\sigma_{\rm c}\rightarrow 0$) mutual information can be expressed as
\begin{align}
I(\RV{x};\RV{y}_{\rm c}) &= h(\RV{u}) + \frac{1}{2}\mathbb{E}\{\log \det \M{G}(\RV{u})\} \nonumber \\
&\hspace{3mm}- \frac{2N-1}{2}\log(\pi e\sigma_{\rm c}^2)+o(1), \label{lem1_highsnr}
\end{align}
where $\M{G}(\RV{u})=\M{J}(\RV{u})^{\rm T}\M{J}(\RV{u})$ is the metric matrix of the constant-power manifold, with $\M{J}(\RV{u})=\frac{\partial (\V{h}_{\rm c}\odot \Psi(\RV{u}))}{\partial \RV{u}}$ being the Jacobian matrix from $\RV{u}$ to $\phi(\RV{u}):=[\mathrm{Re}(\V{h}_{\rm c}\odot \Psi(\RV{u}))^{\rm T},~\mathrm{Im}(\V{h}_{\rm c}\odot \Psi(\RV{u}))^{\rm T}]^{\rm T}$.
\begin{IEEEproof}
Please refer to Appendix \ref{sec:proof_highsnr}.
\end{IEEEproof}
\end{lemma}

With the aid of Lemma~\ref{lem:highsnr}, we are now capable of characterizing the asymptotically capacity-achieving input distributions in the high-\ac{snr} regime.
\begin{proposition}[Generic high-\ac{snr} capacity-achieving input distribution]\label{prop:opt_input}
In the high-\ac{snr} regime $\sigma_{\rm c}\rightarrow 0$, the asymptotically capacity-achieving input distribution $p_{\RV{u}}(\V{u})$ maximizing $h(\RV{u}) + \frac{1}{2}\mathbb{E}\{\log \det \M{G}(\RV{u})\}$ is characterized by
\begin{subequations}\label{high_snr_opt_input}
\begin{align}
p_{\RV{u}}(\V{u})&=p_{\RV{p}^\prime}(\V{p}^\prime)p_{\RV{\theta}}(\V{\theta}),\\
p_{\RV{p}^\prime}(\V{p}^\prime)&\propto\sqrt{\sum_{i=1}^N p_i|[\V{h}_{\rm c}]_i|^{-2}},~\V{p}^\prime\in\Set{P},\\
p_{\RV{\theta}}(\V{\theta})&=\frac{1}{(2\pi)^N},~\V{\theta}\in\mathbb{T}^N.
\end{align}
\end{subequations}
\begin{IEEEproof}
Please refer to Appendix \ref{sec:proof_opt_input}.
\end{IEEEproof}
\end{proposition}

Proposition~\ref{prop:opt_input} suggests that the asymptotic capacity-achieving strategy in the high-\ac{snr} regime is to let $\V{h}_{\rm c}\odot \RV{x}$ be uniformly distributed over $\V{h}_{\rm c}\odot \Psi(\Set{M})$. A direct corollary is that the asymptotically optimal mutual information is given by
\begin{align}
I(\RV{x};\RV{y}_{\rm c})&=\nonumber \frac{1}{2}\log \pi +\sum_{i=1}^N\log |[\V{h}_{\rm c}]_i|^2-\frac{2N-1}{2}\log(e\sigma_{\rm c}^2)\\
&\hspace{3mm}+\log\int_{\Set{P}} \sqrt{\sum_{i=1}^N p_i|[\V{h}_{\rm c}]_i|^{-2}}{\rm d}\V{p}^\prime \!+\!\log 2\!+\!o(1). \label{opt_mi}
\end{align}

Another corollary is that for constant magnitude channels, also known as frequency-flat channels, namely $|[\V{h}_{\rm c}]|_i^2=|\bar{h}|^2,~\forall i=1,\dotsc,N$ for some constant $|\bar{h}|^2$, the asymptotic capacity-achieving input distribution is the uniform distribution over $\Set{M}$:
\begin{equation}
p_{\RV{u}}^{\rm cm}(\V{u}) = \frac{1}{|\Set{M}|},~\V{u}\in\Set{M}.
\end{equation}
In this case, \eqref{opt_mi} can be further simplified as
\begin{align}
I(\RV{x};\RV{y}_{\rm c})&= \frac{2N-1}{2}\log\Bigl(\frac{|\bar{h}|^2}{e\sigma_{\rm c}^2}\Bigr)+\log |\Set{P}|\nonumber\\
&\hspace{3mm}+\log (2\sqrt{\pi})+\frac{1}{2}\log N+o(1), \label{capacity_withp}
\end{align}
for which a detailed estimate of $|\Set{P}|$ will be deferred to Sec.~\ref{sec:acfk}. At this point, we would like to highlight that the \ac{psl} control factor $\zeta_{\rm ACF}$ would indeed have an impact on the achievable rate, in the sense that a larger value of $\zeta_{\rm ACF}$ will result in a smaller $|\Set{P}|$.

From a practical perspective, a problem that follows is that how one could generate codebooks corresponding to such input distributions in an \ac{isac} system. In particular, if we are given a codebook $\Set{C}_p=\{\V{p}^\prime_\ell\}_{\ell=1}^L$ containing uniformly distributed codewords in $\Set{P}$, we could use a distribution matcher \cite{7322261}, as widely applied in probabilistic constellation shapers, to assign a probability of
\begin{equation}
\mathbb{P}\{\RV{p}^\prime=\V{p}_\ell^{\prime}\}=\frac{\sqrt{\sum_{i=1}^N [\V{p}_{\ell}]_i|[\V{h}_{\rm c}]_i|^{-2}}}{\sum_{j=1}^L\sqrt{\sum_{i=1}^N [\V{p}_{j}]_i|[\V{h}_{\rm c}]_i|^{-2}}}
\end{equation}
to each codeword $\V{p}_\ell^{\prime}$, such that the transformed codebook matches the input distribution of \eqref{high_snr_opt_input}. Therefore, the essential complexity of the problem lies in the task of generating the codebook $\Set{C}_p$ and encoding/decoding using that codebook.

In light of this, let us take a closer look at the structure of the feasible region $\Set{M}$. To elaborate, under the $\RV{u}$-coordinate, it consists of the following constraints
\begin{subequations}
\begin{align}
&\RV{\theta}\in\mathbb{T}^N,~\RV{p}\in\mathbb{R}^N,\\
&\rv{p}_i\geq 0,~\forall i=1,\dotsc,N,\\
&|\V{f}_i^{\rm H}\RV{p}|^2\leq \zeta_{\rm ACF}^{-1},~\forall i\in\Set{S}_{\rm ACF}, \label{psd_upper}\\
&\sum_{i=1}^N \rv{p}_i={N}. \label{psd_normalize}
\end{align}
\end{subequations}
We observe that \eqref{psd_upper} and \eqref{psd_normalize} are not entrywise separable constraints, which constitute the main challenge. Indeed, if all constraints were entrywise separable, the feasible region would admit a Cartesian-product representation. In that case, the uniform distribution over the feasible region would factorize into a product of one-dimensional uniform distributions, thereby significantly simplifying the codebook. To this end, we consider the inverse \ac{dft} that maps the power spectrum $\V{p}$ to the \ac{pacf}, $\frac{1}{\sqrt{N}}\M{F}^{\rm H}\RV{p}=\RV{\alpha}$, under which the independent variables are $(\RV{a},\RV{\theta})$, satisfying
\begin{equation}
\RV{a}=\left\{
\begin{array}{ll}
\begin{aligned}
&\bigl[{\rm Re}([\RV{\alpha}]_2),{\rm Im}([\RV{\alpha}]_2), \dotsc,\\ &\hspace{3mm}{\rm Re}([\RV{\alpha}]_{\frac{N}{2}}),{\rm Im}([\RV{\alpha}]_{\frac{N}{2}}),[\RV{\alpha}]_{\frac{N}{2}+1}\bigr]^{\rm T}
\end{aligned}, &\hbox{$N$ even;}\\
\begin{aligned}
&\bigl[{\rm Re}([\RV{\alpha}]_2),{\rm Im}([\RV{\alpha}]_2), \dotsc,\\ &\hspace{3mm}{\rm Re}([\RV{\alpha}]_{\frac{N+1}{2}}),{\rm Im}([\RV{\alpha}]_{\frac{N+1}{2}})\bigr]^{\rm T}
\end{aligned}, &\hbox{$N$ odd,}
\end{array}
\right.
\end{equation}
which follows from the symmetric property of \ac{dft} applied on real-valued vectors. The constraints can therefore be rewritten as
\begin{subequations}
\begin{align}
&\RV{\theta}\in\mathbb{T}^N,\\
&\sqrt{[\RV{a}]_{2i-1}^2+[\RV{a}]_{2i}^2}\leq {\frac{1}{\sqrt{N\zeta_{\rm ACF}}}},~\forall i=1,\dotsc,\Bigl\lfloor \frac{N-1}{2}\Bigr\rfloor,\label{disks}\\
&|[\RV{a}]_{N-1}|\leq {\frac{1}{\sqrt{N\zeta_{\rm ACF}}}},~\mathrm{if~}N~\mathrm{even}, \\
&[\M{F}_N]_{1:N-1,2:N}\M{A}\RV{a}+[\V{f}_1]_{1:N-1}\succeq \V{0},\label{nonneg}
\end{align}
\end{subequations}
where
\begin{align}
\M{A} &=\left\{\begin{array}{ll}
\begin{bmatrix}
\M{I}_{\frac{N-1}{2}}\otimes [1,~i]\\
\M{P}_{\frac{N-1}{2}}\otimes [1,~-i]
\end{bmatrix}, & \hbox{if $N$ odd;} \\
\begin{bmatrix}
\M{I}_{\frac{N}{2}-1}\otimes [1,~i] & 0\\
\V{0}_{1\times (N-2)} & 1\\
\M{P}_{\frac{N}{2}-1}\otimes [1,~-i] & 0
\end{bmatrix}, &\hbox{if $N$ even,}
\end{array}
\right.
\end{align}
and $\M{P}_m\in\mathbb{R}^{m\times m}$ denotes an $m$-order anti-diagonal matrix whose anti-diagonal entries are all equal to unity, with the subscript $m$ being omitted when it is clear from the context. In this coordinate system, we see that only the non-negativity constraint \eqref{nonneg} is not entrywise separable. Next, we show that the coordinate transformation between $\RV{a}$ and $\RV{p}^\prime$ preserves the uniform distribution.
\begin{proposition}
The mapping from $\RV{a}$ to $\RV{p}^\prime$ is a bijective real-valued affine mapping taking the form of
\begin{equation}\label{affine_trans}
\RV{p}^\prime = \sqrt{N}(\M{B}\RV{a}+[\V{f}_1]_{1:N-1}),
\end{equation}
where $\M{B} := [\M{F}_N]_{1:N-1,2:N}\M{A}$, and hence preserves uniform distribution, since $\det \M{B}$ is constant with respect to $\RV{p}^\prime$.
\begin{IEEEproof}
Equation \eqref{affine_trans} follows from previous discussions. It thus suffices to show that $\M{B}$ is non-singular. Note that $[\M{F}_N]_{1:N-1,2:N}$ is a Vandermonde matrix with mutually different nodes, which is invertible. When $N$ is odd, $\M{A}$ is equivalent to $\overline{\M{A}}:=\M{I}\otimes \begin{bmatrix}
1 & i\\
1 & -i
\end{bmatrix}$ up to a row permutation, and hence is also invertible. Similarly, when $N$ is even, $\M{A}$ is equivalent to $\begin{bmatrix}
\overline{\M{A}} & \V{0}\\
\V{0} & 1
\end{bmatrix}$, therefore is also invertible. Thus the product $\M{B}$ is also invertible, which completes the proof.
\end{IEEEproof}
\end{proposition}

We may now conclude that one can obtain a uniform distribution under the $\RV{u}$-coordinate via a uniform distribution under the $(\RV{a},\V{\theta})$-coordinate. Moreover, if the non-negativity constraint \eqref{nonneg} is inactive, sampling from the uniform distribution under the $(\RV{a},\RV{\theta})$-coordinate is substantially simplified. This observation suggests exploiting the two sets of degrees of freedom separately, by modulating \ac{iid} symbols onto $\RV{a}$ and \ac{iid} \ac{psk} symbols onto $\RV{\theta}$, which motivates us to propose the modulation scheme detailed in Sec.~\ref{sec:acfk}.

\section{Auto-correlation Function Keying and Its Sensing Performance}\label{sec:acfk}
In this section, we elaborate on the technical design of the \ac{acfk} and discuss its sensing performance. Simply put, \ac{acfk} modulates symbols drawn from finite constellations under the $(\RV{a},\RV{\theta})$ coordinate, which constitutes a discretized approximation of the continuous uniform distribution over the \ac{acf} domain.
\subsection{Structure of \ac{acfk} and Its Nominal Sensing Performance}
The basic structure of the \ac{acfk} is defined as follows.
\begin{definition}
The \ac{acfk} is defined as
\begin{equation}
\RV{x}=\RV{x}_{\rm{p}}\odot\RV{x}_{\rm{a}}=\RV{x}_{\rm{p}}\odot\sqrt{\sqrt{N}\M{F}_N\RV{x}_{\rm{ACF}}},
\end{equation}
where $\RV{x}_{\rm{p}}\in\mathbb{C}^N$ denotes the phase component of the transmitted communication symbol vector $\RV{x}$. These phases are drawn from a \ac{psk} constellation $\Set{S}_{\rm{p}}$ in an \ac{iid} manner. $\RV{x}_{\rm{a}} \in \mathbb{R}^N$ denotes the amplitude component of the communication symbol vector $\RV{x}$. $\RV{x}_{\rm{ACF}} \in \mathbb{C}^N$ denotes a vector with the special structure, given by
\begin{equation}
\begin{aligned}\label{acf}
\RV{x}_{\rm{ACF}} =
\begin{cases}
\left[1;\widetilde{\RV{x}}_{\rm{s}};\M{P}{\widetilde{\RV{x}}_{\rm{s}}}^*\right], & N \text{ is odd} \\[3pt]
\left[1;\widetilde{\RV{x}}_{\rm{s}};0;\M{P}{\widetilde{\RV{x}}_{\rm{s}}}^*\right], & N \text{ is even}
\end{cases}
\end{aligned}
\end{equation}
with
\begin{equation} \label{normalization}
\widetilde{\RV{x}}_{\rm{s}} = \frac{\RV{x}_{\rm{s}}}{\sqrt{N\beta_{\rm ACF}\zeta_{\rm ACF}}} \in \mathbb{C}^{\lfloor\frac{N-1}{2}\rfloor},
\end{equation}
where $\RV{x}_{\rm{s}} \in \mathbb{C}^{\lfloor\frac{N-1}{2}\rfloor}$ denotes communication symbol vector in the \ac{acf} domain, which is drawn from a constellation $\Set{S}_{\rm{s}}$ in an \ac{iid} manner. $\beta_{\rm ACF} = \max_{\rv{x}_{\rm{s}}\in {\Set{S}_{\rm{s}}}}|\rv{x}_{\rm{s}}|^2 $ denotes the peak symbol power of the constellation $\Set{S}_{\rm{s}}$.
\end{definition}

Following the definition of \ac{acfk}, we observe that the vector \eqref{acf} satisfies the conjugate symmetric property, given by
\begin{equation}
[\RV{x}_{\rm{ACF}}]_{i}=[\RV{x}_{\rm{ACF}}]^*_{1+\langle N+1-i \rangle_N},\quad i = 1,\cdots,N.
\end{equation}
It is also related to the independent variable $\RV{a}$ discussed in Sec.~\ref{sec:highsnr_input}, in the following manner
\begin{equation}
\RV{x}_{\rm ACF}=\begin{bmatrix}
1;~\M{A}\RV{a}
\end{bmatrix},
\end{equation}
which takes a similar form as the \ac{pacf} of $\M{F}_N^{\rm H}\RV{x}$. Due to these reasons, we refer to $\RV{x}_{\rm{ACF}}$ and $\sqrt{N}\M{F}_N\RV{x}_{\rm{ACF}}$ as the \emph{nominal \ac{pacf}} and the \emph{nominal power spectrum}, respectively. Indeed, by using the conjugate symmetry of the Fourier transform, we can observe that $\sqrt{N}\M{F}_N\RV{x}_{\rm{ACF}}$ is in fact a real-valued sequence. However, the definition of \ac{acfk} itself does not guarantee that $\M{F}_N\RV{x}_{\rm ACF}$ satisfies the non-negativity constraint, thereby the nominal \ac{pacf} (resp. power spectrum) may not actually be a valid \ac{pacf}. When this constraint is satisfied, according to the discussion in Sec.~\ref{sec:highsnr_input}, \ac{acfk} with $\RV{x}_{\rm s}$ having uniformly distributed entries within the disks
\begin{equation}\label{disks_tilde}
|[\widetilde{\RV{x}}_{\rm s}]_i| \leq {\frac{1}{\sqrt{N\zeta_{\rm ACF}}}},~\forall i=1,\dotsc, \Bigl\lfloor \frac{N-1}{2}\Bigr\rfloor
\end{equation}
would be asymptotically capacity-achieving in the high-\ac{snr} regime for frequency-flat channels\footnote{Note that the $\lfloor(N+1)/2\rfloor$-th entry of the even-$N$ $\RV{x}_{\rm ACF}$ in \eqref{acf} is set to $0$ only for the simplicity of implementation. One could replace this entry with a uniformly distributed real-valued random variable in order to achieve a better approximation of the asymptotically optimal input distribution.}.

Fortunately, the probability that the nominal \ac{pacf} of an \ac{acfk} sequence violates the power spectral non-negativity constraint is negligible, as long as the sidelobe attenuation factor $\zeta_{\rm ACF}$ is sufficiently large, as depicted by the following proposition.

\begin{proposition}[Non-negativity Violation Probability Bound]\label{prop:violation_prob}
The probability that the nominal \ac{pacf} of an \ac{acfk} sequence violates the power spectral non-negativity constraint $\M{F}_N\RV{x}_{\rm{ACF}}\succeq\V{0}$ at a specific frequency bin $k$ is bounded by
\begin{equation}\label{violation_prob}
\mathbb{P}\left\{\V{f}_k^{\rm T}\RV{x}_{\rm{ACF}}< 0\right\} \leq \exp\left(-\alpha\zeta_{\rm ACF}\right),
\end{equation}
where $\alpha$ takes different values for $\RV{x}_{\rm s}$ following different constellations/distributions, as follows
\begin{equation}\label{alpha_branches}
\alpha \!=\! \left\{
\begin{array}{ll}
1,&\hbox{Uniform distribution;}\\[3pt]
\frac{3(M_q\!+\!L)[L^2+\!(2M_q\!-\!L\!-\!1)^2]}{2[M_q(4M_q^2\!-\!1)\!+\!L(4(M_q\!-\!L)^2\!-\!1)]}, & \hbox{$4(M_q^2\!\!-\!\!L^2)$-\ac{qam};}\\[3pt]
\frac{1}{2},&\hbox{8-\ac{qam}, $M_p$-\ac{psk};}\\[3pt]
\frac{1}{4},&\hbox{Others,}
\end{array}
\right.
\end{equation}
where $M_p\geq 4$, with $M_q\in\mathbb{Z}_+$, and $L\leq M_q/2\in\mathbb{N}$ denoting the size of the ``corner cut'' of \ac{qam} constellations, as exemplified by Fig.~\ref{fig:qam32}. ``Uniform distribution'' corresponds to the uniform distribution whose support satisfies \eqref{disks_tilde}.
\begin{IEEEproof}
Please refer to Appendix \ref{sec:proof_violation_prob}.
\end{IEEEproof}
\end{proposition}

Proposition \ref{prop:violation_prob} suggests that the ratio between the volume $|\Set{P}|$ and the relaxed feasible region
\begin{equation}
\Set{P}_{\rm relax} = \{\RV{p}^{\prime}\in\mathbb{R}^{N-1}:~|\V{f}_k^{\rm H}\RV{p}|^2\leq \zeta_{\rm ACF}^{-1},~k\in\Set{S}_{\rm ACF}\}
\end{equation}
can be bounded by constants. In particular, we have
\begin{equation}
\frac{|\Set{P}|}{|\Set{P}_{\rm relax}|} = 1-\mathbb{P}\{\exists k(\V{f}_k^{\rm T}\RV{x}_{\rm ACF}<0)\},
\end{equation}
where the $\RV{x}_{\rm ACF}$ corresponds to the \ac{acfk} sequence drawn from the uniform distribution over $\Set{P}_{\rm relax}$. In light of this, by application of Proposition~\ref{prop:violation_prob} and the union bound, we see that
\begin{subequations}
\begin{align}
\log |\Set{P}_{\rm relax}| &= \log |\Set{P}| \!-\! \log (1\!-\!\mathbb{P}\{\exists k(\V{f}_k^{\rm T}\RV{x}_{\rm ACF}\!<\!0)\})   \\
&\leq \log |\Set{P}| - \log (1-Ne^{-\zeta_{\rm ACF}}) \label{unit_alpha}\\
&\leq \log |\Set{P}| + \log\Bigl(1+\frac{1}{\zeta_{\rm ACF}-\log N}\Bigr),
\end{align}
\end{subequations}
where \eqref{unit_alpha} follows from the fact that $\alpha=1$ for the uniform distribution. Therefore, by choosing $\zeta_{\rm ACF}\in \Omega(\log N)$, and note that $|\Set{P}|\leq |\Set{P}_{\rm relax}|$, we have
$$
\log |\Set{P}|= \log |\Set{P}_{\rm relax}| + O(1).
$$
The volume $|\Set{P}_{\rm relax}|$ can be readily computed as
\begin{equation}
|\Set{P}_{\rm relax}|=\left\{\begin{array}{ll}
\frac{(2\pi\zeta_{\rm ACF}^{-1})^{\frac{N-1}{2}}}{\sqrt{N}}, & \hbox{$N$ odd;}\\
\frac{(2\pi\zeta_{\rm ACF}^{-1})^{\frac{N}{2}-1}}{\sqrt{N}}\cdot 2\sqrt{\zeta_{\rm ACF}^{-1}}, & \hbox{$N$ even;}
\end{array}
\right.
\end{equation}
In either case, we have $\log|\Set{P}_{\rm relax}|=-\frac{N-1}{2}\log (\zeta_{\rm ACF})+\lfloor\frac{N}{2}\rfloor\log 2+\lfloor\frac{N-1}{2}\rfloor \log \pi-\frac{1}{2}\log N$.
Now, using \eqref{capacity_withp}, we see that for frequency-flat channels, the capacity in the high-\ac{snr} regime may be expressed as
\begin{align}
I(\RV{x};\RV{y}_{\rm c})&\!=\!\frac{2N-1}{2}\log\Bigl(\frac{|\bar{h}|^2}{e\sigma_{\rm c}^2\sqrt{\zeta_{\rm ACF}}}\Bigr)\!+\!\frac{1}{4}\log\zeta_{\rm ACF}\nonumber \\
&\hspace{3mm}\!+\!\Bigl(\Bigl\lfloor\frac{N}{2}\Bigr\rfloor+1\Bigr)\log 2\!+\!\Bigl(\Bigl\lfloor\frac{N\!-\!1}{2}\Bigr\rfloor\!+\!\frac{1}{2}\Bigr) \log \pi \!+\! O(1).
\end{align}
This implies that in the high-\ac{snr} regime, from the perspective of capacity computation, the effective \ac{snr} is attenuated by a factor of $\sqrt{\zeta_{\rm ACF}}$ due to the \ac{psl} constraint, compared to the conventional Gaussian input distribution that is capacity-achieving in the absence of the \ac{psl} constraint.

Furthermore, when $N$ is sufficiently large, according to the central limit theorem, we may obtain a better estimate of the non-negativity violation probability based on Gaussian approximation. The approximated non-negativity violation probability is derived in the following result.
\begin{remark}
[Approximated Non-negativity Violation Probability]\label{rem:approx_violation_prob}
When $N$ is sufficiently large, the probability that the nominal \ac{pacf} of an \ac{acfk} sequence violates the power spectral non-negativity constraint $\M{F}_N\RV{x}_{\rm{ACF}}\succeq \V{0}$ at a specific frequency bin $k$ can be approximated as
\begin{equation}
\mathbb{P}\left\{[\M{F}_N{\RV{x}}_{\rm{ACF}}]_{k}<0\right\} \approx \Phi\left(-\frac{1}{\sigma}\right)
\end{equation}
where the function $\Phi(x)=\frac{1}{\sqrt{2\pi}}\int_{-\infty}^{x}e^{-\frac{t^2}{2}}{\rm d}x$ denotes the \ac{cdf} of the standard Gaussian distribution $\mathcal{N}(0,1)$, and $\sigma$ is a fitted standard deviation given in \eqref{sigma_def}. Detailed discussion of this Gaussian approximation is deferred to Sec.~\ref{ssec:ber_phase}.
\end{remark}

In this paper, we first focus on the nominal \ac{pacf} satisfying the power spectral non-negativity constraint, and then evaluate the implications of the non-negativity constraint violation on the sensing and communication performance. At this point, it is worthwhile to demonstrate the advantage of the proposed data modulation scheme in terms of its sensing performance. To elaborate, we start by analyzing the \ac{psl} and \ac{esl} of the random nominal \ac{pacf}, as presented in the following result.


\begin{figure}[t]
     \centering
     \includegraphics[width=0.6\linewidth]{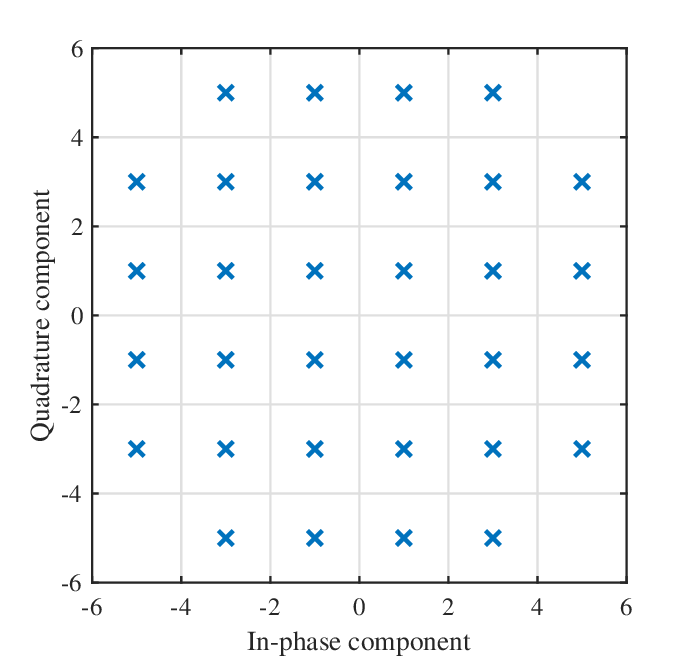}
     \caption{The (unnormalized) 32-\ac{qam} constellation, a ``corner cut'' $4(M_q^2-L^2)$-\ac{qam} constellation with $M_q=3$ and $L=1$.}
     \label{fig:qam32}
\end{figure}

\begin{proposition} \label{prop:nominal_pacf}
The squared nominal \ac{pacf} is
\begin{equation} \label{squared_pacf}
\begin{aligned}
\left|\RV{x}_{\rm{ACF}}\right|^2=
\begin{cases}
\left[1;\frac{|\RV{x}_{\rm{s}}|^2}{{N\beta_{\rm ACF}\zeta_{\rm ACF}}};\M{P}\frac{|\RV{x}_{\rm{s}}|^2}{{N\beta_{\rm ACF}\zeta_{\rm ACF}}}\right], & N \text{ is odd}, \\[7pt]
\left[1;\frac{|\RV{x}_{\rm{s}}|^2}{{N\beta_{\rm ACF}\zeta_{\rm ACF}}}; 0; \M{P}\frac{|\RV{x}_{\rm{s}}|^2}{{N\beta_{\rm ACF}\zeta_{\rm ACF}}}\right], & N \text{ is even}.
\end{cases}
\end{aligned}
\end{equation}
The \ac{psl} and \ac{esl} of the nominal \ac{pacf} at a specific non-zero-valued sidelobe bin $k$ respectively are
\begin{subequations}
\begin{align}
\ac{psl} &= \frac{1}{N\zeta_{\rm ACF}}\\
[\ac{esl}]_k &= \frac{1}{N\zeta_{\rm ACF}} \cdot \frac{\mathbb{E}\left\{|{\rv{x}_{\rm{s}}}|^2\right\} }{\max_{\rv{x}_{\rm{s}}\in\Set{S}_{\rm{s}}}|\rv{x}_{\rm{s}}|^2}
\end{align}
\end{subequations}
\end{proposition}
\begin{IEEEproof}
Please refer to Appendix \ref{sec:proof_nominal_pacf}.
\end{IEEEproof}

From Proposition~\ref{prop:nominal_pacf}, it is clear that both the \ac{psl} and \ac{esl} of the nominal \ac{pacf} are precisely controllable, which is achieved by directly modulating normalized and attenuated communication symbols, along with their conjugates, onto the \ac{acf} domain. These symbols collectively constitute the sidelobes of the nominal \ac{pacf}.

\begin{remark}[Sidelobe attenuation factor and nominal \ac{psl}]
Note that the non-negativity violation probability bound in Proposition~\ref{prop:violation_prob} exhibits an exponential decay with respect to the sidelobe attenuation factor $\zeta_{\rm{ACF}}$. According to Proposition \ref{prop:nominal_pacf}, a larger sidelobe attenuation factor yields lower nominal \ac{psl} and hence improved sensing performance. Therefore, in practical systems, the sidelobe attenuation factor may be set sufficiently large to achieve reliable sensing performance, implying that the non-negativity constraint is often satisfied with high probability. For example, by configuring the constellation $\Set{S}_{\rm{s}}$ as the 32-\ac{qam} constellation shown in Fig.~\ref{fig:qam32} and setting the sidelobe attenuation factor to $\zeta_{\rm{ACF}}=10~{\rm{dB}}$, the non-negativity violation probability bound can be calculated as $e^{-8.5}\approx2.035\times 10^{-4}$.
\end{remark}

Furthermore, it is straightforward to verify that when the nominal \ac{pacf} satisfies the non-negativity constraint, it coincides with the actual \ac{pacf}. Indeed, by using the Wiener-Khinchine theorem, we can express the actual \ac{pacf} as
\begin{subequations}
\begin{align}
\RV{r}_{\RV{x}} &= \frac{1}{\sqrt{N}}\M{F}_N^{\rm{H}}|\RV{x}|^2\\
& = \frac{1}{\sqrt{N}}\M{F}_N^{\rm{H}}\left|\RV{x}_{\rm{p}}\odot\sqrt{\sqrt{N}\M{F}_N\RV{x}_{\rm{ACF}}}\right|^2\\
& = \frac{1}{\sqrt{N}}\M{F}_N^{\rm{H}}\left(|\RV{x}_{\rm{p}}|^2\odot\left|\sqrt{N}\M{F}_N\RV{x}_{\rm{ACF}}\right|\right).
\end{align}
\end{subequations}
By exploiting the constant modulus property of \ac{psk} constellations, we have
\begin{equation} \label{pacf}
\RV{r}_{\RV{x}}=\frac{1}{\sqrt{N}}\M{F}_N^{\rm{H}}\left|\sqrt{N}\M{F}_N\RV{x}_{\rm{ACF}}\right|.
\end{equation}
Based on the assumption that \ac{acfk} sequences satisfy the power spectral non-negativity constraint, we can safely ignore the magnitude operation in \eqref{pacf}. Therefore, when the spectral non-negativity constraint is satisfied, the entire \ac{pacf} can be precisely controlled. This constitutes the key feature of \ac{acfk} in terms of sensing performance.


 \begin{figure}[t]
     \centering
     \includegraphics[width=1\linewidth]{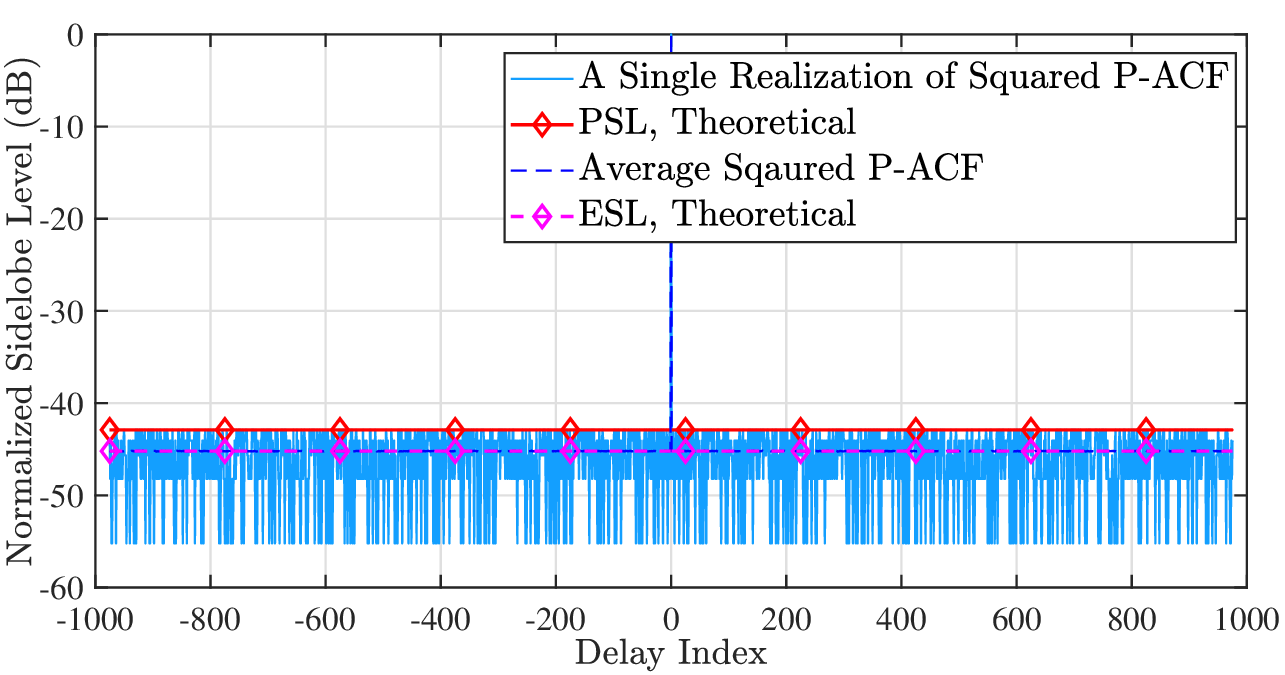}
     \caption{A single realization and the average squared nominal \ac{pacf} of \ac{acfk}, with the sidelobe attenuation factor $\zeta_{\rm{ACF}}=10~{\rm{dB}}$, $N=1951$, and the 32-\ac{qam} constellation.}
     \label{fig:pacf}
 \end{figure}

To facilitate an intuitive understanding of the precise \ac{psl} control capability, we perform numerical simulations to visualize the squared nominal \ac{pacf} of \ac{acfk} in Fig.~\ref{fig:pacf}. We consider \ac{acfk} with $N=1951$ subcarriers. The sidelobe attenuation factor is set to $\zeta_{\rm{ACF}}=10~{\rm{dB}}$. The constellation $\Set{S}_{\rm{s}}$ is configured as the ``corner cut'' of the 32-\ac{qam} constellation, as shown in Fig.~\ref{fig:qam32}. The average squared \ac{pacf} is obtained by averaging over $10000$ independent random realizations. It can be observed that the theoretical values match very well with the simulation result. Each non-zero-valued sidelobe of the structured \ac{pacf} corresponds to a specific constellation point of the 32-\ac{qam}. Both the \ac{psl} and \ac{esl} are precisely controlled as designed.

\subsection{Impact of Non-negativity Constraint Violation}\label{subsec:impact}
We would like to highlight that the discussion so far focuses on the nominal quantities, e.g. nominal \ac{pacf}. In this subsection, we consider the impact of nominal \ac{pacf} violating the power spectral non-negativity constraint. Specifically, the non-negativity violation implies that there must exist at least one frequency bin $k$ such that the corresponding entry $[\M{F}_N{\RV{x}}_{\rm{ACF}}]_{k}< 0$. To facilitate subsequent analysis, let us denote the set of all frequency bins $k$ that violate the power spectral non-negativity constraint, namely $[\M{F}_N{\RV{x}}_{\rm{ACF}}]_{k}< 0$ as $\RS{K}$, which is naturally a random set. We present the structure of the actual \ac{pacf} of \ac{acfk} in the following proposition.
\begin{proposition}[Actual \ac{pacf} under potential violation of non-negativity]\label{prop:pacf_real}
The \ac{pacf} of \ac{acfk} is
\begin{equation}
\RV{r}_{\RV{x}}=\RV{x}_{\rm{ACF}}+\sqrt{N}\M{F}_N^{\rm{H}}\RV{e},
\end{equation}
where $\RV{e}$ denotes the error vector contributed by non-negativity violated entries at frequency bins $k\in\RS{K}$, taking the values of $\frac{2\left|[\M{F}_N{\RV{x}}_{\rm{ACF}}]_{k}\right|}{\sqrt{N}}$, while those entries satisfying the non-negativity constraint are zero.
\end{proposition}
\begin{IEEEproof}
Please refer to Appendix \ref{sec:proof_pacf_real}.
\end{IEEEproof}

It is clear that the \ac{pacf} coincides with its nominal \ac{pacf}, when the nominal \ac{pacf} satisfies the non-negativity constraint, as proven in Proposition~\ref{prop:nominal_pacf}. The structure of the actual \ac{pacf} can be regarded as the superposition of the nominal \ac{pacf} $\RV{x}_{\rm{ACF}}$ and the random fluctuation term $\sqrt{N}\M{F}_N^{\rm{H}}\RV{e}$.

Now note that the non-negativity constraint violation induces random fluctuations in the mainlobe of \ac{pacf}, thereby we use \ac{pslr} as the sensing metric for a fair comparison. For notational simplicity in subsequent analysis, let us denote the $k$-th entry of $\RV{e}$ as $\rv{e}_k$. According to the definition of \ac{dft}, the \ac{pacf} of \ac{acfk} can be expressed as
\begin{equation}
[\RV{r}_{\RV{x}}]_{n}=[\RV{x}_{\rm{ACF}}]_{n}+\sum_{k\in\RS{K}}\rv{e}_ke^{j\theta_{k,n}}
\end{equation}
Using the triangle inequality, we bound the \ac{pslr} as follows
\begin{subequations}
\begin{align}
\ac{pslr}&=\max_{n\in \Set{S}_{\rm ACF}}\left\{\left|\frac{\left[\RV{r}_{\RV{x}}\right]_n}{\left[\RV{r}_{\RV{x}}\right]_1}\right|^2\right\}\\
&=\max_{n\in \Set{S}_{\rm ACF}}\left\{\frac{\left|[\RV{x}_{\rm{ACF}}]_n+\sum_{k\in\RS{K}}\rv{e}_ke^{j\theta_{k,n}}\right|^2}{\left|[\RV{x}_{\rm{ACF}}]_1+\sum_{k\in\RS{K}}\rv{e}_ke^{j\theta_{k,1}}\right|^2}\right\}\\
&\leq \frac{\max_{n\in \Set{S}_{\rm ACF}}\left(\left|[\RV{x}_{\rm{ACF}}]_n\right|+\sum_{k\in\RS{K}}\rv{e}_k\right)^2}{\left(1+\sum_{k\in\RS{K}}\rv{e}_k\right)^2}\\
&=\frac{\left(\max_{n\in \Set{S}_{\rm ACF}}\left|[\RV{x}_{\rm{ACF}}]_n\right|+\sum_{k\in\RS{K}}\rv{e}_k\right)^2}{\left(1+\sum_{k\in\RS{K}}\rv{e}_k\right)^2}\\
&=\frac{\left(\frac{1}{\sqrt{N\zeta_{\rm{ACF}}}}+\sum_{k\in\RS{K}}\rv{e}_k\right)^2}{\left(1+\sum_{k\in\RS{K}}\rv{e}_k\right)^2}\triangleq \overline{g}(\{\rv{e}_k\}_{k\in\RS{K}}). \label{upper_bound_pslr}
\end{align}
\end{subequations}
Clearly, for a given set $\RS{K}$, the upper bound on the \ac{pslr} is an increasing function of each $\rv{e}_k$ with $k\in \RS{K}$. This property indicates that larger error amplitudes $\rv{e}_k$ would cause a higher \ac{pslr}. Since these amplitudes are determined by $\left|[\M{F}_N{\RV{x}}_{\rm{ACF}}]_k\right|$, smaller values of $[\M{F}_N{\RV{x}}_{\rm{ACF}}]_k$ may severely degrade sensing performance. In brief, the non-negativity constraint violation introduces the random fluctuation term $\sqrt{N}\M{F}_N^{\rm{H}}\RV{e}$ into the \ac{pacf}, thus weakening the capability of controlling the \ac{pslr} of its \ac{pacf}. This motivates the following analysis on the distribution of the \ac{pslr}.


In particular, in the presence of potential non-negativity violation, the \ac{pslr} can no longer precisely be controlled by $\frac{1}{N\zeta_{\rm ACF}}$, but is instead upper-bounded by the quantity $\overline{g}(\{\rv{e}_k\}_{k\in\RS{K}})$. Ideally, with this upper bound, we may derive a lower bound on the probability
$\mathbb{P}\left\{\ac{pslr}\leq \eta\right\}$ for some $\eta>0$, which characterizes the controllability of the \ac{pslr} in the worst-case scenario. Nevertheless, the probability is generally intractable, since the number of possible error patterns, marked by the set $\RS{K}$, is enormous.

In light of this, we focus on the most likely violation event, namely at most one entry (with random index) of the vector $\M{F}_N\RV{x}_{\rm{ACF}}$ is negative, and hence $|\RS{K}|\leq 1$. Such event would be the dominant violation event for practical \ac{acfk} systems having a sufficiently large sidelobe attenuation factor $\zeta_{\rm{ACF}}$. Under this assumption, the \ac{pacf} can be simplified as
\begin{equation}
[\RV{r}_{\RV{x}}]_n=[\RV{x}_{\rm{ACF}}]_n+\rv{e}_ke^{j\rv{\theta}_{\rv{k},n}}.
\end{equation}
and the corresponding \ac{pslr} can be bounded by
\begin{equation}
{\ac{pslr}}\leq \frac{\left(\frac{1}{\sqrt{N\zeta_{\rm{ACF}}}}+\rv{e}_k\right)^2}{\left(1+\rv{e}_k\right)^2}  \triangleq g(\rv{e}_k),
\end{equation}
where ${k}$ denotes the only entry that violates the non-negativity constraint. Intuitively, the probability $\mathbb{P}\left\{\ac{pslr}\leq g(\gamma)\right\}$ is related to the outage probability $\mathbb{P}\left\{\rv{e}_k>\delta\right\}$ for some $\delta>0$. Specifically, a sufficient condition for $\ac{pslr}\leq g(\gamma)$ is that $g(\rv{e}_k)\leq g(\gamma)$ holds for all $k$, due to the fact that $\ac{pslr}\leq g(\rv{e}_k)$. Note that $g(\rv{e}_k)$ is an increasing function of $\rv{e}_k$, and thus the sufficient condition is equivalent to $\rv{e}_k\leq\gamma$ for all $\rv{k}$.

Using a similar approach as employed in the proof of Proposition \ref{prop:violation_prob}, we derive an upper bound on the outage probability $\mathbb{P}\left\{\rv{e}_k>\delta\right\}$.
\begin{corollary} \label{prop:violation_prob_corollary}
The outage probability that the amplitude $\rv{e}_k$ at a specific frequency bin $k$ exceeds the positive value $\delta>0$ is bounded by
\begin{equation}
\mathbb{P}\left\{\rv{e}_k>\delta\right\}\leq \exp\left(-\alpha\zeta_{\rm ACF}{\left(1+\frac{N\delta}{2}\right)^2}\right)
\end{equation}
\begin{IEEEproof}
Please refer to Appendix \ref{sec:proof_violation_prob_corollary}.
\end{IEEEproof}
\end{corollary}

Next, by using the propositional logic reasoning and the union bound, we obtain the following theoretical lower bound.
\begin{proposition}[\ac{pslr} bound]\label{prop:psl_prop}
The probability that the \ac{pslr} is no greater than the positive value $g(\gamma)$ is bounded by
\begin{equation}
\mathbb{P}\left\{\ac{pslr}\leq g (\gamma)\right\}\geq 1- \mathbb{P}\{|\RS{K}|\geq 2\}-N e^{-\alpha\zeta_{\rm ACF}{\left(1+\frac{N\gamma}{2}\right)^2}}.
\end{equation}
In particular, when at most one entry of the vector $\M{F}_N\RV{x}_{\rm{ACF}}$ is negative, namely $|\RS{K}|\leq 1$, we have
$$
\mathbb{P}\left\{\ac{pslr}\leq g (\gamma)\right\}\geq 1-N e^{-\alpha\zeta_{\rm ACF}{\left(1+\frac{N\gamma}{2}\right)^2}}.
$$
\begin{IEEEproof}
Please refer to Appendix \ref{sec:proof_psl_prop}.
\end{IEEEproof}
\end{proposition}

Similar to the discussion in Remark~\ref{rem:approx_violation_prob}, when $N$ is sufficiently large, the central limit theorem-based approximation can also be employed. Therefore, we can readily obtain the following result.
\begin{remark}\label{prop:approx_psl}
When $N$ is sufficiently large and at most one entry of the vector $\M{F}_N\RV{x}_{\rm{ACF}}$ is negative, the probability that the \ac{pslr} is no greater than the positive value $g(\gamma)$ is approximated as
\begin{equation}
\mathbb{P}\left\{\ac{pslr} \leq g (\gamma)\right\} \approx 1-N
\Phi\left(-\frac{1+\frac{N\gamma}{2}}{\sigma}\right).
\end{equation}

To be specific, relying on Gaussian approximation, we have
\begin{subequations}
\begin{align*}
\mathbb{P}\left\{\rv{e}_k>\delta\right\}
&=\mathbb{P}\left\{[\M{F}_N\RV{x}_{\rm{ACF}}]_{k}<-\frac{\delta\sqrt{N}}{2}\right\} \\
&\approx \mathbb{P}\left\{\frac{\sqrt{N}[\M{F}_N{\RV{x}}_{\rm{ACF}}]_{k}-1}{\sigma}<\frac{-\frac{N\delta}{2}-1}{\sigma}\right\}\\
&=\Phi\left(-\frac{1+\frac{N\delta}{2}}{\sigma}\right), \label{approx_psl} \tag{\theequation}
\end{align*}
\end{subequations}
where the standard deviation $\sigma$ regarding this central-limit approximation is defined in \eqref{sigma_def}, with detailed discussion in Sec.~\ref{ssec:ber_phase}. With the aid of the derivation in Appendix \ref{sec:proof_psl_prop}, we obtain this remark.
\end{remark}

\begin{figure}[t]
	\centering
	\includegraphics[width=0.9\linewidth]{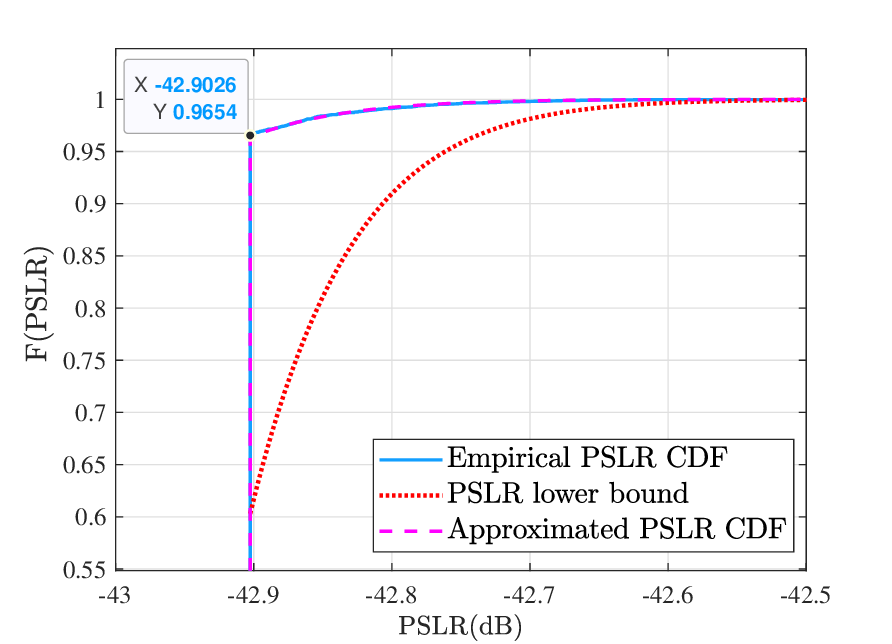}
	\caption{The empirical \ac{cdf}, theoretical lower bound and approximated \ac{cdf} of the \ac{pslr}, for \ac{acfk} with sidelobe attenuation factor $\zeta_{\rm{ACF}}=10~{\rm{dB}}$, $N=1951$, and 32-\ac{qam} constellation.}
	\label{fig:psl_bound}
\end{figure}

To provide an intuitive insight into the impacts of the power spectral non-negativity constraint violation on the sensing performance, we carry out numerical simulations to visually demonstrate the distribution of the $\ac{pslr}$. We consider \ac{acfk} with $N=1951$ subcarriers. The constellation $\Set{S}_{\rm{s}}$ is configured as a 32-\ac{qam} constellation, while the sidelobe attenuation factor is set to $\zeta_{\rm{ACF}}=10~{\rm{dB}}$. We plot the empirical \ac{cdf} of the \ac{pslr} of its \ac{pacf} in Fig.~\ref{fig:psl_bound} by performing $10000$ independent random realizations. With Proposition \ref{prop:nominal_pacf}, the value of the \ac{pslr} of its nominal \ac{pacf} is calculated as ${\ac{pslr}}=-42.9026~{\rm{dB}}$. As shown in Fig.~\ref{fig:psl_bound}, the probability that the nominal \ac{pacf} satisfies the power spectral non-negativity constraint reaches $0.9654$. The empirical \ac{cdf} converges to unity when the \ac{pslr} is approximately $-42.6~{\rm{dB}}$, which implies that the effect of the random fluctuation term $\sqrt{N}\M{F}_N^{\rm{H}}\RV{e}$ is often negligible. The theoretical lower bound also validates this phenomenon. Besides, it can be observed that the approximated \ac{cdf} closely coincides with the empirical \ac{cdf}, which verifies that the scenario with only one negative entry in the vector $\M{F}_N\RV{x}_{\rm{ACF}}$ is dominantly prevalent in practical \ac{acfk} systems, when the sidelobe attenuation factor $\zeta_{\rm{ACF}}$ is sufficiently large.


\subsection{Coherent Processing Over Multiple Channel Uses: Ambiguity Function Keying}
In the previous subsection we have considered the \ac{acfk} under a single channel use. In a different context, coherent integration across multiple channel uses (also known as ``\ac{ofdm} symbols'') has been shown to improve the \ac{esl} of \ac{ofdm} \cite{iceberg}. To elaborate, assuming that the targets remain stationary across $M$ channel uses, the \ac{isac} \ac{tx} randomly generates $M$ \ac{iid} symbol sequences from a given constellation, and averages over their corresponding \ac{mf} output at the sensing \ac{rx}. Through the coherent integration, the \ac{esl} can be reduced by a factor of $M$ \cite{iceberg}. Moreover, for non-stationary targets, under the assumption that the Doppler shift is negligible compared to the subcarrier spacing, coherent integration can be generalized to the \ac{2dfft} processing, under which the coherent processing over $M$ channel uses would reduce the \ac{esl} of the \ac{fstaf} by a factor of $M$ \cite{zhang2025discrete}.

In contrast to \ac{esl}, since averaging across large samples would yield near-Gaussian sidelobe distributions, according to the central limit theorem, coherent integration does not guarantee a controllable \ac{psl} reduction. To address this problem, we propose a generalization of the \ac{acfk}, termed as the \acf{afk}, to achieve deterministic \ac{psl} reductions when multiple channel uses are available, as characterized by the following definition.
\begin{definition}
The \ac{afk} over $N$ subcarriers and $M$ channel uses (also termed as ``transmission slots'') is defined as
\begin{equation}\label{afk}
\RM{X}=\RM{X}_{\rm{p}}\odot\RM{X}_{\rm{a}}=\RM{X}_{\rm{p}}\odot\sqrt{\sqrt{MN}\M{F}_N\RM{X}_{\rm{AF}}\M{F}_M^{\rm H}}
\end{equation}
where $\RM{X}_{\rm{p}} \in \mathbb{C}^{N \times M}$ denotes the phase component of the transmitted communication symbol matrix $\RM{X}$. These phases are drawn from a \ac{psk} constellation $\Set{S}_{\rm{p}}$ in an \ac{iid} manner. $\RM{X}_{\rm{a}} \in \mathbb{C}^{N \times M}$ denotes the amplitude component of the communication symbol matrix $\RM{X}$. $\RM{X}_{\rm{AF}} \in \mathbb{C}^{N \times M}$ denotes a matrix in the delay-Doppler domain, satisfying the conjugate symmetric property about the origin, given by
\begin{equation}
[\RM{X}_{\rm{AF}}]_{i,j} = [\RM{X}_{\rm{AF}}]^*_{1+\langle N+1-i \rangle _N, 1+\langle M+1-j \rangle _M},
\end{equation}
where $i = 1,\cdots,N$ and $j = 1,\cdots,M$.
\end{definition}

We refer to $\RM{X}_{\rm{AF}}$ and $\sqrt{MN}\M{F}_N\RM{X}_{\rm{AF}}\M{F}_M^{\rm H}$ as the \emph{nominal \ac{fstaf}} and the \emph{nominal two-dimensional power spectrum}, respectively. These names follow from the fact that the two-dimensional power spectral non-negativity constraint violation may still occur. The nominal \ac{fstaf} design becomes more sophisticated in the delay-Doppler domain. Hence, we classify such designs into four categories and enumerate the corresponding structural configurations of the nominal \ac{fstaf} as follows.
\begin{definition} \label{dd_symbols} The structures of the nominal \ac{fstaf} include four categories,
\begin{enumerate}
\item $N$ is odd and $M$ is odd
\begin{equation}
\RM{X}_{\rm{AF}} = \left[
{\begin{array}{*{20}{c}}
1 & \RV{x}_a^{\rm T} & (\M{P}_{n_c}\RV{x}_a)^{\rm H}\\
\RV{x}_b & \RM{X}_A & \RM{X}_B\\
\M{P}_{n_r}\RV{x}_b^* & \M{P}_{n_r}\RM{X}_B^*\M{P}_{n_c} & \M{P}_{n_r}\RM{X}_A^*\M{P}_{n_c}
\end{array}}
\right]
\end{equation}
\item $N$ is odd and $M$ is even
\begin{equation}
\RM{X}_{\rm{AF}} \!=\! \left[
{\begin{array}{*{20}{c}}
1 & \RV{x}_a^{\rm T}  & 0 & (\M{P}_{n_c}\RV{x}_a)^{\rm H} \\
\RV{x}_b & \RM{X}_A  & \RV{x}_d & \RM{X}_B\\
\M{P}_{n_r}\RV{x}_b^* & \M{P}_{n_r}\RM{X}_B^*\M{P}_{n_c} & \M{P}_{n_r}\RV{x}_d^* & \M{P}_{n_r}\RM{X}_A^*\M{P}_{n_c}
\end{array}}
\right]
\end{equation}
\item $N$ is even and $M$ is odd
\begin{equation}
\RM{X}_{\rm{AF}} = \left[
{\begin{array}{*{20}{c}}
1 & \RV{x}_a^{\rm T} & (\M{P}_{n_c}\RV{x}_a)^{\rm H}\\
\RV{x}_b & \RM{X}_A & \RM{X}_B\\
0 & \RV{x}_c^{\rm T} & (\M{P}_{n_c}\RV{x}_c)^{\rm H}\\
\M{P}_{n_r}\RV{x}_b^* & \M{P}_{n_r}\RM{X}_B^*\M{P}_{n_c} & \M{P}_{n_r}\RM{X}_A^*\M{P}_{n_c}
\end{array}}
\right]
\end{equation}
\item $N$ is even and $M$ is even
\begin{equation}
\RM{X}_{\rm{AF}} \!=\! \left[
{\begin{array}{*{20}{c}}
1 & \RV{x}_a^{\rm T}  & 0 & (\M{P}_{n_c}\RV{x}_a)^{\rm H} \\
\RV{x}_b& \RM{X}_A  & \RV{x}_d & \RM{X}_B\\
0 & \RV{x}_c^{\rm T} & 0 & (\M{P}_{n_c}\RV{x}_c)^{\rm H}\\
\M{P}_{n_r}\RV{x}_b^* & \M{P}_{n_r}\RM{X}_B^*\M{P}_{n_c} & \M{P}_{n_r}\RV{x}_d^* & \M{P}_{n_r}\RM{X}_A^*\M{P}_{n_c}
\end{array}}
\right]
\end{equation}
\end{enumerate}
where $\RV{x}_a,\RV{x}_c \in \mathbb{C}^{n_c}$ and $\RV{x}_b,\RV{x}_d \in \mathbb{C}^{n_r} $ denote the symbol vectors, while $\RM{X}_A, \RM{X}_B \in \mathbb{C}^{n_r\times n_c}$ denote the symbol matrices, with $n_r=\lfloor\frac{N-1}{2}\rfloor$ and $n_c=\lfloor\frac{M-1}{2}\rfloor$. We denote each entry of the vectors and matrices as $\widetilde{\rv{x}}_{\rm{s}}$, which can be expressed as
\begin{equation}\label{afksymbols}
\widetilde{\rv{x}}_{\rm{s}} = \frac{\rv{x}_{\rm{s}}}{\sqrt{MN\beta_{\rm AF}\zeta_{\rm AF}}} \in \mathbb{C}
\end{equation}
where $\rv{x}_{\rm{s}}$ denotes a communication symbol in the delay-Doppler domain, which is drawn from a constellation $\Set{S}_{\rm{s}}$ in an \ac{iid} manner. $\zeta_{\rm AF}$ denotes a positive sidelobe attenuation factor, $\beta_{\rm AF} = \max_{\rv{x}_{\rm{s}}\in{\Set{S}_{\rm{s}}}}|\rv{x}_{\rm{s}}|^2 $ denotes the peak symbol power of the constellation $\Set{S}_{\rm{s}}$.
\end{definition}

We proceed to evaluate the nominal sensing performance by analyzing the \ac{psl} and \ac{esl} of the random nominal \ac{fstaf}, as presented in the following proposition.
\begin{proposition}\label{prop:afkpsl}
The \ac{psl} and \ac{esl} of the nominal \ac{fstaf} at a specific non-zero-valued sidelobe bin $(k,l)$ are respectively given by
\begin{subequations}
\begin{align}
\ac{psl} &= \frac{1}{MN\zeta_{\rm AF}},\\
[\ac{esl}]_{k,l} &= \frac{1}{MN\zeta_{\rm AF}} \cdot \frac{\mathbb{E}\left\{|{\rv{x}_{\rm{s}}}|^2\right\} }{\max_{\rv{x}_{\rm{s}}\in\Set{S}_{\rm{s}}}|\rv{x}_{\rm{s}}|^2}.
\end{align}
\end{subequations}
When the nominal \ac{fstaf} satisfies the two-dimensional power spectral non-negativity constraint, the \ac{fstaf} of \ac{afk} coincides with its nominal \ac{fstaf}, as given by $\RM{R}_{\RM{X}} = \RM{X}_{\rm{AF}}$.
\end{proposition}
\begin{IEEEproof}
Please refer to Appendix \ref{sec:proof_afkpsl}.
\end{IEEEproof}

\begin{figure}[t]
	\centering
	\includegraphics[width=1\linewidth]{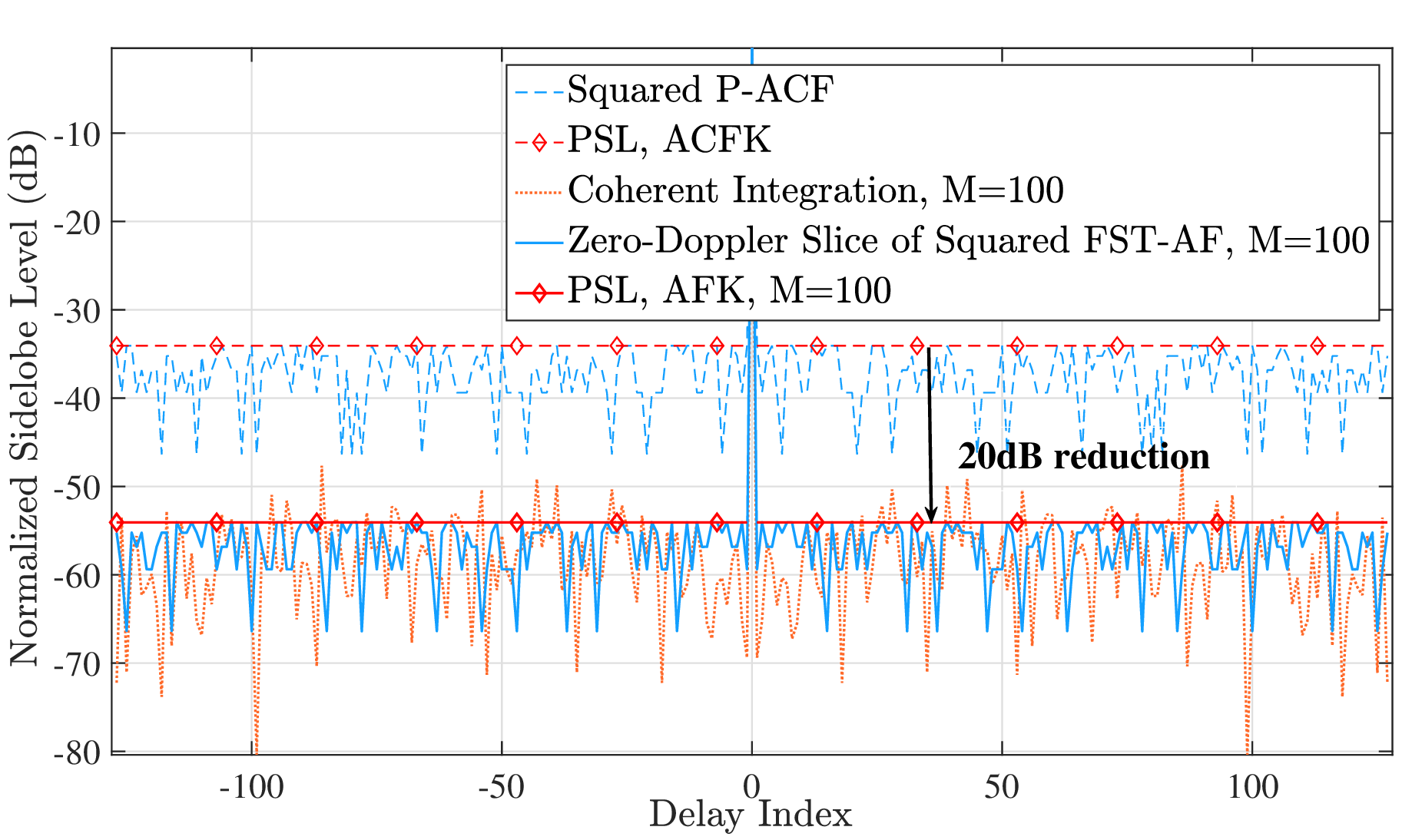}
	\caption{A single realization of the squared nominal \ac{pacf} and the zero-Doppler slice of a single realization of the squared nominal \ac{fstaf}, for \ac{acfk} and \ac{afk}, with the sidelobe attenuation factors $\zeta_{\rm AF}=\zeta_{\rm{ACF}}=10~{\rm{dB}}$, $N=257$, under 32-\ac{qam} constellation. \ac{acfk} corresponds to $M=1$, whereas \ac{afk} uses $M=100$.}
	\label{fig:fstaf}
\end{figure}

Clearly, by exploiting $M$ transmission slots, the \ac{afk} can suppress the \ac{psl} of the nominal \ac{fstaf} by a factor of $M$. To facilitate an intuitive understanding of \ac{afk}, the zero-Doppler slice of the squared nominal \ac{fstaf} is plotted in Fig.~\ref{fig:fstaf}. In particular, we consider \ac{acfk} and \ac{afk} with $N=257$ subcarriers, and especially for \ac{afk} we use $M=100$ transmission slots. The sidelobe attenuation factor is set to $\zeta_{\rm{AF}}=10~{\rm{dB}}$. It can be seen that the coherent integration fails to achieve precise \ac{psl} control, whereas \ac{afk} enables precise control of the \ac{psl} of the nominal \ac{fstaf}. As indicated by Proposition \ref{prop:afkpsl}, compared with its \ac{acfk} counterpart, a $20~{\rm{dB}}$ \ac{psl} reduction can be observed. When the sidelobe attenuation factor $\zeta_{\rm{AF}}$ is sufficiently large, the \ac{fstaf} coincides with its nominal \ac{fstaf} with high probability, and hence a controllable \ac{psl} reduction of its \ac{fstaf} can be obtained.

\section{Transceiver Design for \ac{acfk}} \label{sec:architecture}
In this section, we present a reference design for monostatic \ac{isac} transceivers employing \ac{acfk} over quasi-static multipath channels. We establish the signal processing pipeline for both sensing and communication tasks, as illustrated in Fig.~\ref{fig:system_design}.

\begin{figure*}[t]
    \centering
    \includegraphics[width=0.95\linewidth]{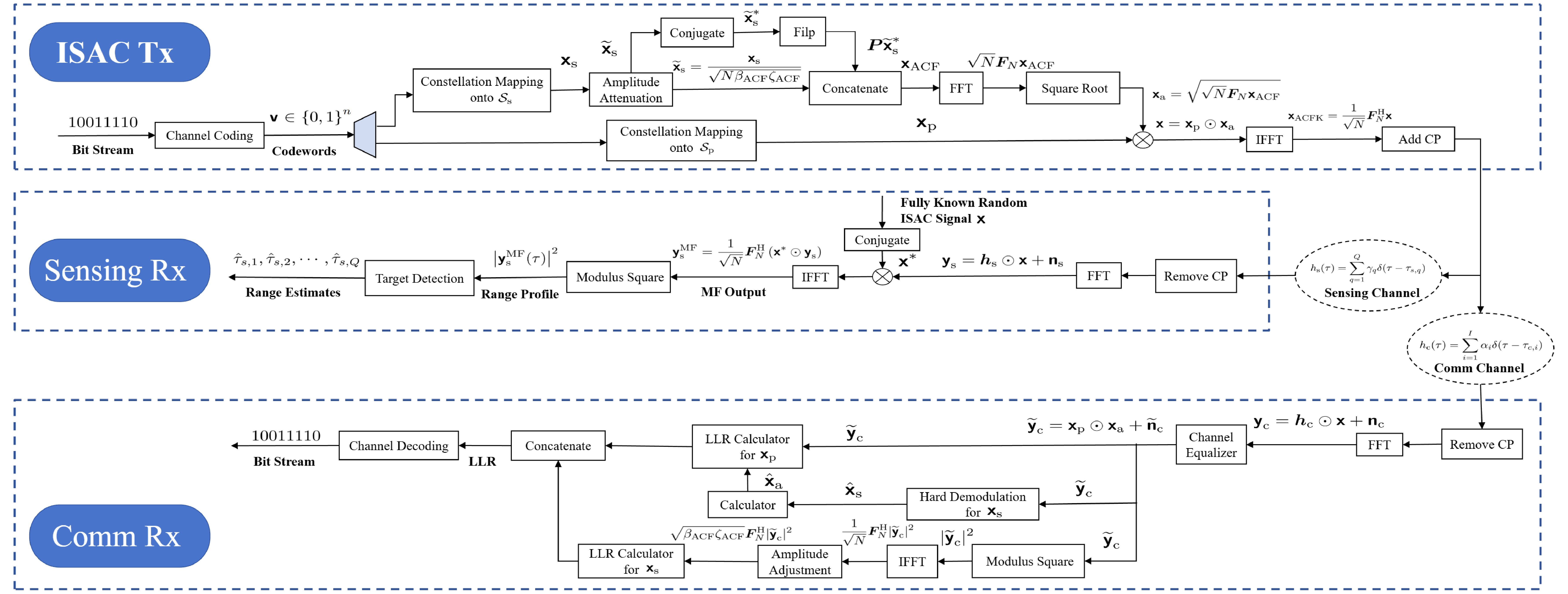}

    \caption{A signal processing pipeline for a monostatic \ac{isac} system employing \ac{acfk} over quasi-static multipath channels.}
    \label{fig:system_design}
\end{figure*}

\subsection{\ac{isac} Transmitter Design}
As can be observed from the top row of Fig.~\ref{fig:system_design}, at the transmitter side, a random bit stream is first generated from the source. Each block of $k$ information bits is then encoded into a codeword $\RV{v}\in\left\{0,1\right\}^n$ with the rate given by $R_{\rm c}=\frac{k}{n}<1$, where $n$ denotes the codeword length. The codeword is then split into two separate streams, corresponding to amplitude and phase components of the spectrum, respectively. Specifically, one of the data streams is mapped onto the \ac{psk} constellation $\Set{S}_{\rm{p}}$, generating the phase component $\RV{x}_{\rm{p}}$. The other stream is utilized to construct the amplitude component $\RV{x}_{\rm{a}}$ through a series of transformations, as prescribed in the design of \ac{acfk}, including constellation mapping onto $\Set{S}_{\rm{s}}$, amplitude attenuation, the conjugate operation, \ac{dft} and square-root operation. After the two streams are generated, they are combined by an element-wise multiplication, which generates the frequency-domain communication symbol vector $\RV{x}$. Finally, these symbols are modulated onto subcarriers. The time-domain representation of the \ac{isac} signal $\RV{x}_{\rm {ACFK}} \in \mathbb{C}^{N}$ is thus given by
\begin{equation}
\RV{x}_{\rm{ACFK}} = \frac{1}{\sqrt{N}}\M{F}_N^{\rm H}\RV{x} = \frac{1}{\sqrt{N}}\M{F}_N^{\rm H}\left(\RV{x}_{\rm{p}}\odot\sqrt{\sqrt{N}\M{F}_N\RV{x}_{\rm{ACF}}}\right).
\end{equation}

From a practical perspective, \ac{acfk} can be implemented by modifying and amalgamating the existing \ac{ofdm} and \ac{sc} modulation schemes. As shown in Fig.~\ref{fig:system_design}, the generation of the phase component $\RV{x}_{\rm{p}}$ completely aligns with the \ac{ofdm} modulation scheme, while the generation of the amplitude component $\RV{x}_{\rm{a}}$ is a variant of the \ac{sc} modulation scheme, with the modification that a square-root operation is introduced at the \ac{isac} \ac{tx}, along with the corresponding magnitude-squared operation at the communication \ac{rx}.

For the channel coding, we employ a joint coding scheme, in which the information bits modulated on amplitude and phase are jointly encoded using a single codebook. Naturally, an alternative approach is to modulate information bits separately on amplitude and phase, referred to as the ``separate coding scheme'' hereafter. As we shall see in Sec.~\ref{sec:simulations}, the joint coding scheme enables encoded bits under favorable channel conditions to assist the decoding of bits suffering from poor channel quality, thereby yielding a lower total \ac{ber}.

\begin{remark}[Generalized transmitter design for \ac{afk}]
The proposed reference design can be readily extended to \ac{afk}. To elaborate, the amplitude component $\RM{X}_{\rm{a}}$ is generated by mapping communication symbols in a structured manner onto the delay-Doppler domain in accordance with Definition \ref{dd_symbols}, followed by \ac{2dfft}. After performing an element-wise multiplication between $\RM{X}_{\rm{a}}$ and $\RM{X}_{\rm{p}}$, the time-frequency-domain communication symbol matrix $\RM{X}$ is obtained. Finally, these symbols are modulated onto subcarriers by taking inverse \ac{dft} on each column of $\RM{X}$. To conclude, \ac{afk} can be implemented by modifying and amalgamating the existing \ac{ofdm} and \ac{otfs} modulation schemes. The generation of the phase component $\RM{X}_{\rm{p}}$ completely aligns with the \ac{ofdm} modulation scheme, while the generation of the amplitude component $\RM{X}_{\rm{a}}$ is a variant of the \ac{otfs} modulation scheme.
\end{remark}

\subsection{Communication Receiver Design}\label{ssec:comm_rx}
The architecture of the communication \ac{rx} is shown in the bottom row of Fig.~\ref{fig:system_design}. The frequency-domain representation of the received communication signal is expressed as
\begin{equation}
\RV{y}_{\rm{c}} = \V{h}_{\rm {c}} \odot \RV{x}_{\rm{p}}\odot\RV{x}_{\rm{a}} + \RV{n}_{\rm c} = \V{h}_{\rm {c}} \odot \RV{x}_{\rm{p}}\odot\sqrt{\sqrt{N}\M{F}_N\RV{x}_{\rm{ACF}}} + \RV{n}_{\rm c}.
\end{equation}
Here, we propose a minimalist communication \ac{rx} design for \ac{acfk}, serving as a proof of concept as well as a benchmark for the \ac{ber} performance analysis. We assume that $\V{h}_{\rm {c}}$ is perfectly known at the communication \ac{rx} by means of channel estimation prior to the transmission. To recover the transmitted communication symbols, the \ac{isi} needs to be eliminated through channel equalization. Without loss of generality, we employ the frequency-domain \ac{zf} equalizer, which is the most basic linear channel equalization approach. After the frequency-domain \ac{zf} equalization, the received signal is represented as
\begin{equation}\label{afawgn}
\widetilde{\RV{y}}_{\rm{c}} = \RV{x}_{\rm{p}}\odot\RV{x}_{\rm{a}} + \widetilde{\RV{n}}_{\rm{c}} =\RV{x}_{\rm{p}}\odot\sqrt{\sqrt{N}\M{F}_N\RV{x}_{\rm{ACF}}} + \widetilde{\RV{n}}_{\rm{c}}
\end{equation}
where $\widetilde{\RV{y}}_{\rm{c}} = \frac{\RV{y}_{\rm{c}}}{\V{h}_{\rm{c}}}$ and $\widetilde{\RV{n}}_{\rm{c}} = \frac{\RV{n}_{\rm c}}{\V{h}_{\rm{c}}}$.
Since linear transformations preserve the complex Gaussian distribution, the noise $\widetilde{\RV{n}}_{\rm{c}}$ remains a complex Gaussian distributed random vector, namely $\widetilde{\RV{n}}_{\rm{c}}\sim\mathcal{CN}$($\V{0}, \rm{diag}(\widetilde{\V{\sigma}}^2$)), where $\widetilde{\V{\sigma}}^2=\left[\widetilde{\sigma}_1^2;\cdots;\widetilde{\sigma}_{N}^2\right] \in \mathbb{R}^{N}$ and $\widetilde{\sigma}_{k}^2=\sigma_{\rm c}^2|[\V{h}_{\rm c}]_k|^{-2}$, $k=1,2,\cdots,N$.

The communication \ac{rx} aims to accurately decode all the communication symbols from the received signal. Since \ac{acfk} comprises two separated streams of modulated communication signals, it is required to reconstruct these two communication symbol vectors $\RV{x}_{\rm{p}}$ and $\RV{x}_{\rm{s}}$, which convey $N$ \ac{iid} \ac{psk} symbols in the frequency domain and $\lfloor\frac{N-1}{2}\rfloor$ \ac{iid} symbols in the \ac{acf} domain, respectively.

\begin{figure*}[!b]
\vspace{-4mm}
\hrule
\setcounter{MYtempeqncnt}{\value{equation}}
\setcounter{equation}{63}
\begin{subequations} \label{qamdemoifft}
\begin{align*}
\frac{1}{\sqrt{N}}\M{F}_N^{\rm H}|\widetilde{\RV{y}}_{\rm{c}}|^2 &= \frac{1}{\sqrt{N}}\M{F}_N^{\rm{H}} |\sqrt{N}\M{F}_N\RV{x}_{\rm{ACF}}| +\frac{1}{\sqrt{N}}\M{F}_N^{\rm H}|\widetilde{\RV{n}}_{\rm{c}}|^2+\frac{1}{\sqrt{N}}\M{F}_N^{\rm H}\left(2\RV{x}_{\rm{a}}\odot\widetilde{\RV{n}}_{\rm{x}}\right)\\
&= \RV{x}_{\rm{ACF}} + \sqrt{N}\M{F}_N^{\rm{H}}\RV{e} +\frac{1}{\sqrt{N}}\M{F}_N^{\rm H}|\widetilde{\RV{n}}_{\rm{c}}|^2+\frac{1}{\sqrt{N}}\M{F}_N^{\rm H}\left(2\RV{x}_{\rm{a}}\odot\widetilde{\RV{n}}_{\rm{x}}\right). \tag{\theequation}
\end{align*}
\end{subequations}
\setcounter{equation}{\value{MYtempeqncnt}}
\end{figure*}

\begin{figure*}[!b]
\vspace{-4mm}
\hrule
\setcounter{MYtempeqncnt}{\value{equation}}
\setcounter{equation}{64}
\begin{equation}\label{qamdemodraw}
\sqrt{\beta_{\rm ACF}\zeta_{\rm ACF}}\M{F}_N^{\rm H}|{\widetilde{\RV{y}}_{\rm{c}}}|^2\!=\! \sqrt{N\beta_{\rm ACF}\zeta_{\rm ACF}}\RV{x}_{\rm{ACF}} \!+\! N\sqrt{\beta_{\rm ACF}\zeta_{\rm ACF}}\M{F}^{\rm{H}}\RV{e} \!+\!\sqrt{\beta_{\rm ACF}\zeta_{\rm ACF}}\M{F}_N^{\rm H}|{\widetilde{\RV{n}}_{\rm c}}|^2\!+\!\sqrt{\beta_{\rm ACF}\zeta_{\rm ACF}}\M{F}_N^{\rm H}\left(2\RV{x}_{\rm{a}}\odot\widetilde{\RV{n}}_{\rm{x}}\right).
\end{equation}
\setcounter{equation}{\value{MYtempeqncnt}}
\end{figure*}

\subsubsection{The recovery of the \ac{acf} symbol vector}
The sidelobes of the nominal \ac{pacf} $\RV{x}_{\rm{ACF}}$ consist of $\lfloor\frac{N-1}{2}\rfloor$ \ac{iid} communication symbols to be demodulated and their conjugates. We recover these symbols simply through the reversed process of the modulation scheme adopted at the \ac{isac} \ac{tx}. Specifically, by taking the squared magnitude of the received signal and using the constant-modulus property of \ac{psk} constellations, we have
\begin{subequations}
\begin{align*}
|\widetilde{\RV{y}}_{\rm{c}}|^2 &= \left(\RV{x}_{\rm{p}}\odot\RV{x}_{\rm{a}} + \widetilde{\RV{n}}_{\rm{c}}\right)^*\odot\left(\RV{x}_{\rm{p}}\odot\RV{x}_{\rm{a}} + \widetilde{\RV{n}}_{\rm{c}}\right)\\
& = |\RV{x}_{\rm{p}}|^2|\RV{x}_{\rm{a}}|^2 + |\widetilde{\RV{n}}_{\rm{c}}|^2+2\RV{x}_{\rm{a}}\odot{\rm{Re}}\left(\widetilde{\RV{n}}_{\rm{c}}\odot\RV{x}_{\rm{p}}^*\right)\\
& = |\sqrt{N}\M{F}_N\RV{x}_{\rm{ACF}}| + |\widetilde{\RV{n}}_{\rm{c}}|^2 + 2\RV{x}_{\rm{a}}\odot{\rm{Re}}\left(\widetilde{\RV{n}}_{\rm{c}}\odot\RV{x}_{\rm{p}}^*\right).\tag{\theequation} \label{demod_vector}
\end{align*}
\end{subequations}
To facilitate further analysis, let us denote ${\rm{Re}}\left({\widetilde{\RV{n}}_{\rm {c}}}\odot\RV{x}_{\rm{p}}^*\right)$ as $\widetilde{\RV{n}}_{\rm{x}}\in\mathbb{R}^{N}$. By performing the inverse \ac{dft} on the processed signal, we obtain the demodulation model, represented as \eqref{qamdemoifft}. Upon multiplying both sides by the normalized and attenuated coefficient $\sqrt{N\beta_{\rm ACF}\zeta_{\rm ACF}}$, we have \eqref{qamdemodraw}.


Now, according to \eqref{acf} and \eqref{normalization}, we may extract the \ac{acf} symbols $\RV{x}_{\rm s}$ as follows\setcounter{equation}{65}
\begin{align}
\widetilde{\RV{y}}_{\rm demod}:&=\sqrt{\beta_{\rm ACF}\zeta_{\rm ACF}}[\M{F}_N^{\rm H}|\widetilde{\RV{y}}_{\rm c}|^2]_{2:\lfloor(N+1)/2\rfloor} \nonumber \\
&=\RV{x}_{\rm s}+\widetilde{\RV{n}}_1+\widetilde{\RV{n}}_2+\widetilde{\RV{n}}_3, \label{demod_extract}
\end{align}
where
\begin{align*}
\widetilde{\RV{n}}_1&=\sqrt{\beta_{\rm ACF}\zeta_{\rm ACF}}[\M{F}_N^{\rm H}(2\RV{x}_{\rm a}\odot \widetilde{\RV{n}}_{\rm x})]_{2:\lfloor(N+1)/2\rfloor},\\
\widetilde{\RV{n}}_2&=\sqrt{\beta_{\rm ACF}\zeta_{\rm ACF}}[\M{F}_N^{\rm H}|\widetilde{\RV{n}}_{\rm c}|^2]_{2:\lfloor(N+1)/2\rfloor},\\
\widetilde{\RV{n}}_3&=N\sqrt{\beta_{\rm ACF}\zeta_{\rm ACF}}[\M{F}_N^{\rm H}\RV{e}]_{2:\lfloor(N+1)/2\rfloor}.
\end{align*}
The exact a posteriori \ac{llr} can be computed as
\begin{align}
\rv{\lambda}_{k,i}^{\rm a} &=\ln \frac{\mathbb{P}\{\rv{b}_{k,i}=1|[\widetilde{\RV{y}}_{\rm demod}]_k\}}{\mathbb{P}\{\rv{b}_{k,i}=0|[\widetilde{\RV{y}}_{\rm demod}]_k\}}\nonumber \\
&=\ln \frac{\sum_{s \in \Set{S}_{{\rm s},i,1}} \exp(\Lambda_k(s))}{\sum_{s \in \Set{S}_{{\rm s},i,0}} \exp(\Lambda_k(s))},\label{exact_llr_acf}
\end{align}
where $\RV{b}_k=[\rv{b}_{k,1},\dotsc,\rv{b}_{k,n_{\rm c}}]^{\rm T}$ is the bit vector corresponding to the \ac{acf} symbol $[\RV{x}_{\rm s}]_k$, $n_{\rm c}$ denotes the number of bits per \ac{acf} symbol, $\Set{S}_{{\rm s},i,1}$ and $\Set{S}_{{\rm s},i,0} $ denote the subsets of constellation $\Set{S}_{\rm{s}}$ with communication symbols satisfying $\rv{b}_i=1$ and $\rv{b}_i=0$, respectively, and $\Lambda_k(s)$ represents the log-likelihood corresponding to the observation $\widetilde{\RV{y}}_{\rm demod}|[\RV{x}_{\rm s}]_k=s$.

Now, observe that the first term on the right hand side of \eqref{demod_extract}, $\RV{x}_{\rm s}$, is exactly the desired term for demodulation, while the remaining are error terms. In general, these error terms are non-Gaussian, and are mutually dependent across entries. Therefore, the log-likelihood $\Lambda_k(s)$ would be intractable. Despite this, following the common practice in the literature of practical receiver design \cite{1006557,1268000,5779722}, we compute the approximate \emph{a posteriori} \acp{llr} used in soft demodulation by treating these terms as Gaussian, neglecting the inter-symbol dependence and omitting small error terms. In particular, the term $\widetilde{\RV{n}}_3$ is negligible when the violation probability of the non-negativity constraint is small, while the term $\widetilde{\RV{n}}_2$ is negligible in the high-\ac{snr} regime since it is on the order of $O(\sigma_{\rm c}^2)$. Thus the dominant error term is $\widetilde{\RV{n}}_1$. The corresponding covariance matrix (conditioned on $\RV{x}_{\rm a}$) is given by
\begin{subequations}
\begin{align*}
\M{R}&=\beta_{\rm ACF}\zeta_{\rm ACF}\M{S}\M{F}_N^{\rm H}\mathbb{E}_{\RV{n}_{\rm x}}\left\{\left(2\RV{x}_{\rm{a}}\odot\RV{n}_{\rm{x}}\right)(\left(2\RV{x}_{\rm{a}}\odot\RV{n}_{\rm{x}}\right))^{\rm T}\right\}\M{F}_{N}\M{S}^{\rm H}\\
&= 4\beta_{\rm ACF}\zeta_{\rm ACF}\M{S}\M{F}_N^{\rm H}\mathbb{E}_{\RV{n}_{\rm x}}\left\{{{\rm {diag}}(\RV{x}_{\rm{a}})}{\RV{n}_{\rm{x}}}{\RV{n}_{\rm{x}}}^{\rm T}{{\rm {diag}}(\RV{x}_{\rm{a}})}\right\}\M{F}_{N}\M{S}^{\rm H}\\
&= 4\beta_{\rm ACF}\zeta_{\rm ACF}\M{S}\M{F}_N^{\rm H}{{\rm {diag}}(\RV{x}_{\rm{a}})}\mathbb{E}_{\RV{n}_{\rm x}}\left\{\RV{n}_{\rm{x}}\RV{n}_{\rm{x}}^{\rm T}\right\}{{\rm {diag}}(\RV{x}_{\rm{a}})}\M{F}_{N}\M{S}^{\rm H}\\
&= 2\sigma_{\rm{c}}^2\beta_{\rm ACF}\zeta_{\rm ACF}\M{S}\M{F}_N^{\rm H}{{\rm {diag}}(|\RV{x}_{\rm{a}}|^2\odot|\V{h}_{\rm c}|^{-2})}\M{F}_{N}\M{S}^{\rm H}\\
&=\frac{2}{\sqrt{N}}\sigma_{\rm{c}}^2\beta_{\rm ACF}\zeta_{\rm ACF}\M{S}{\rm circ}\left(\RV{r}_{\RV{x}}\circledast (\M{F}_N^{\rm H} |\V{h}_{\rm c}|^{-2})\right)\M{S}^{\rm H} \tag{\theequation}
\end{align*}
\end{subequations}
where $\M{S}$ denotes the selection matrix satisfying $\M{S}\V{z}=[\V{z}]_{2:\lfloor(N+1)/2\rfloor}$.
We observe that the $(k,k)$-th diagonal entry of $\M{R}$ reads
\begin{equation}\label{actual_var}
[\M{R}]_{k,k} = \frac{2\sigma_{\rm{c}}^2\beta_{\rm ACF}\zeta_{\rm ACF}}{N}\V{1}^{\rm T}(|\V{x}_{\rm{a}}|^2\odot|\V{h}_{\rm c}|^{-2}),
\end{equation}
which is identical for all $k$. Similarly, the pseudo-covariance matrix of $\widetilde{\RV{n}}_1$, conditioned on $\RV{x}_{\rm a}$, is given by
\begin{equation}
\M{C} \!=\! \frac{2}{\sqrt{N}}\sigma_{\rm{c}}^2\beta_{\rm ACF}\zeta_{\rm ACF}\M{S}{\rm circ}\left(\RV{r}_{\RV{x}}\circledast (\M{F}_N^{\rm H} |\V{h}_{\rm c}|^{-2})\right)\M{F}_N^2\M{S}^{\rm H},
\end{equation}
where $\M{F}_N^2$ is involutory, and hence $\M{F}_N^2\M{S}^{\rm H}$ selects the conjugate-symmetric positions of a vector corresponding to those selected by $\M{S}$, in the sense that $\V{z} = [1;\M{S}\V{z};\M{S}\M{F}_N^2\V{z}]$. The diagonal entries of $\M{C}$ read
\begin{equation}\label{pseudo_var}
[\M{C}]_{k,k} = \frac{2\sigma_{\rm{c}}^2\beta_{\rm ACF}\zeta_{\rm ACF}}{N}\sum_{n=1}^N\frac{|[\RV{x}_{\rm a}]_n|^2}{|[\V{h}_{\rm c}]_n|^2} e^{j\frac{4\pi}{N}(k-1)(n-1)}.
\end{equation}

According to \eqref{actual_var} and \eqref{pseudo_var}, using the Gaussian approximation and neglecting inter-symbol dependence, the log-likelihood can be approximated by

\begin{align}
\Lambda_k(s)& \!\approx\! -\frac{[\M{R}]_{k,k}|[\widetilde{\RV{y}}_{\rm demod}]_k\!-\!s|^2\!-\!{\rm Re}[[\M{C}]_{k,k}^*([\widetilde{\RV{y}}_{\rm demod}]_k\!-\!s)^2]}{[\M{R}]_{k,k}^2-|[\M{C}]_{k,k}|^2} \nonumber  \\
&\hspace{3mm}-\frac{1}{2}\ln([\M{R}]_{k,k}^2-|[\M{C}]_{k,k}|^2)-\ln\pi. \label{lambda_ks}
\end{align}

In practice, the expression \eqref{lambda_ks} remains not directly computable, since it relies on the \ac{acf} symbols $\RV{x}_{\rm s}$ through $\RV{x}_{\rm a}$. In light of this, we use the further approximation of $\RV{x}_{\rm a}\propto \V{1}$, which is reasonable when the sidelobe attenuation factor $\zeta_{\rm ACF}$ is large. When this is not the case, $\RV{x}_{\rm a}\propto \V{1}$ could serve as the initialization of an iterative symbol detection process.

\begin{remark}[Redundancy of the conjugate \ac{acf}-domain entries]
The proposed receiver demodulates only the non-redundant \ac{acf}-domain symbol vector $\RV{x}_{\rm s}$. Since $|\widetilde{\RV{y}}_{\rm c}|^2$ is real-valued, its inverse \ac{dft} is conjugate symmetric. Hence, the noisy entries corresponding to the conjugate counterpart of $\RV{x}_{\rm s}$ are completely determined by the noisy entries corresponding to $\RV{x}_{\rm s}$, rather than being independent observations. Exploiting these conjugate entries would therefore only duplicate the same soft information and does not provide any diversity gain.
\end{remark}

\subsubsection{The recovery of the phase symbol vector}
We next consider the demodulation of the phase symbol vector $\RV{x}_{\rm{p}}$, which comprises $N$ \ac{iid} \ac{psk} symbols. After frequency-domain \ac{zf} equalization, the $k$-th subcarrier can be written as
\begin{equation}\label{pskdemod}
\widetilde{\rv{y}}_k = \rv{a}_k \rv{x}_{{\rm p},k}+\widetilde{\rv{n}}_k,\qquad k=1,\ldots,N,
\end{equation}
where
\begin{align*}
\rv{x}_{{\rm p},k} = \left[\RV{x}_{\rm{p}}\right]_{k}, \widetilde{\rv{y}}_k=\left[\widetilde{\RV{y}}_{\rm{c}}\right]_{k},\widetilde{\rv{n}}_k=\left[\widetilde{\RV{n}}_{\rm{c}}\right]_{k}, \rv{a}_k\!=\! \left[\RV{x}_{\rm{a}}\right]_{k},
\end{align*}
with $\widetilde{\rv{n}}_k\sim\mathcal{CN}(0,\widetilde{\sigma}_{k}^2)$ and $\widetilde{\sigma}_{k}^2=\sigma_{\rm c}^2|[\V{h}_{\rm c}]_k|^{-2}$. If $\rv{a}_k$ were perfectly known, the conditional likelihood is given by
\begin{equation}
p_{\widetilde{\rv{y}}_k|\rv{x}_{{\rm p},k},\rv{a}_k}(y_k|x,a_k)=\frac{1}{\pi \widetilde{\sigma}_{k}^2}\exp\Bigl(-\frac{|y_k-a_k x|^2}{\widetilde{\sigma}_{k}^2}\Bigr).
\end{equation}
For soft demodulation, the conditional a posteriori \ac{llr} of the $i$-th bit carried by $\rv{x}_{{\rm p},k}$ can be written as
\begin{equation}
\rv{\lambda}_{k,i}^{\rm p} = \ln \frac{\sum_{s\in\Set{S}_{{\rm p},i,1}}\exp\Bigl(-\frac{|\widetilde{\rv{y}}_k-\rv{a}_k s|^2}{\widetilde{\sigma}_{k}^2}\Bigr)}{\sum_{s\in\Set{S}_{{\rm p},i,0}}\exp\Bigl(-\frac{|\widetilde{\rv{y}}_k-\rv{a}_k s|^2}{\widetilde{\sigma}_{k}^2}\Bigr)},
\end{equation}
where $\Set{S}_{{\rm p},i,1}$ and $\Set{S}_{{\rm p},i,0}$ denote the subsets of \ac{psk} symbols whose corresponding $i$-th bit equals $1$ and $0$, respectively.


Similar to the derivation in the \ac{acf} symbol case, the exact a posteriori \ac{llr} is not directly computable, since $\rv{a}_k$'s are determined by the \ac{acf} symbols through $\RV{x}_{\rm a}=\sqrt{\sqrt{N}\M{F}_N\RV{x}_{\rm ACF}}$, which is unknown at the communication \ac{rx}. Therefore, we use a two-step estimator as follows. First, hard demodulation based on the approximate \ac{acf} symbol log-likelihood \eqref{lambda_ks} is adopted to estimate $\RV{x}_{\rm{s}}$. Then, using these estimated communication symbols, an approximation of the coefficient $\rv{a}_k$ can be computed and denoted as $\widehat{\rv{a}}_k$. Finally, soft demodulation on the model \eqref{pskdemod} is conducted, and the approximate a posteriori \ac{llr} may be expressed as
\begin{equation}\label{pskllrapp}
\rv{\lambda}_{k,i}^{\rm p} \approx \ln \frac{\sum_{s\in\Set{S}_{{\rm p},i,1}}\exp\Bigl(-\frac{|\widetilde{\rv{y}}_k-\widehat{\rv{a}}_k s|^2}{\widetilde{\sigma}_{k}^2}\Bigr)}{\sum_{s\in\Set{S}_{{\rm p},i,0}}\exp\Bigl(-\frac{|\widetilde{\rv{y}}_k-\widehat{\rv{a}}_k s|^2}{\widetilde{\sigma}_{k}^2}\Bigr)},
\end{equation}

Based on the approximate \acp{llr} given by \eqref{exact_llr_acf}, \eqref{lambda_ks}, and \eqref{pskllrapp}, the communication \ac{rx} can carry out the joint decoding scheme to recover all transmitted information bits.

\subsection{Sensing Receiver Design}
The sensing receiver design is portrayed in the middle row of Fig.~\ref{fig:system_design}. The frequency-domain representation of the received echo signal is expressed as
\begin{equation}\label{senchannel}
\RV{y}_{\rm s} = \V{h}_{\rm s} \odot \RV{x}_{\rm{p}}\odot\RV{x}_{\rm{a}} + \RV{n}_{\rm s} = \V{h}_{\rm s} \odot \RV{x}_{\rm{p}}\odot\sqrt{\sqrt{N}\M{F}_N\RV{x}_{\rm{ACF}}} + \RV{n}_{\rm s}.
\end{equation}

To extract the target range parameters corresponding to path delays, we employ the standard \ac{mf} approach. In the frequency domain, the \ac{mf} operation can be expressed as follows
\begin{equation}\label{mf}
\begin{aligned}
\RV{y}_{\rm s}^{\rm MF} &= \frac{1}{\sqrt{N}}\M{F}_N^{\rm H}\left(\RV{x}^*\odot\RV{y}_{\rm s}\right).
\end{aligned}
\end{equation}
Substituting \eqref{senchannel} into \eqref{mf}, the output signal after \ac{mf} can be represented as
\begin{subequations}
\begin{align*}
\RV{y}_{\rm s}^{\rm MF} &= \frac{1}{\sqrt{N}}\M{F}_N^{\rm H}\left(\RV{x}^*\odot\RV{y}_{\rm s}\right)\\
&= \frac{1}{\sqrt{N}}\M{F}_N^{\rm H}\left[\RV{x}_{\rm{p}}^*\odot\RV{x}_{\rm{a}}\odot\left(\V{h}_{\rm s} \odot \RV{x}_{\rm{p}}\odot\RV{x}_{\rm{a}} + \RV{n}_{\rm s}\right)\right]\\
&= \frac{1}{\sqrt{N}}\M{F}_N^{\rm H}\left(\V{h}_{\rm s}\odot|\RV{x}_{\rm{a}}|^2+\RV{x}_{\rm{p}}^*\odot\RV{x}_{\rm{a}}\odot\RV{n}_{\rm s}\right)\\
&= \frac{1}{\sqrt{N}}\M{F}_N^{\rm H}\left(\V{h}_{\rm s}\odot|\sqrt{N}\M{F}_N\RV{x}_{\rm{ACF}}|\right)+\RV{n}_{\rm s}^{\rm MF}, \tag{\theequation}
\end{align*}
\end{subequations}
where $\RV{n}_{\rm s}^{\rm MF}$ denotes the output Gaussian noise. For a specific sensing scenario involving $Q$ point-like targets, we can further obtain
\begin{subequations}
\begin{align*}
\RV{y}_{\rm s}^{\rm MF}(\tau)
&= \V{h}_{\rm s}(\tau)\circledast\RV{r}_{\RV{x}}(\tau) +\RV{n}_{\rm s}^{\rm MF}(\tau)\\
&= \sum_{q=1}^Q\gamma_{q}\RV{r}_{\RV{x}}\left(\tau-\tau_{s,q}\right)+\RV{n}_{\rm s}^{\rm MF}(\tau).\tag{\theequation}
\end{align*}
\end{subequations}

The output signal can be interpreted as a linear combination of $Q$ time-shifted versions of the \ac{pacf} $\RV{r}_{\RV{x}}(\tau)$, with additive noise, which is typically referred to as the \emph{range profile} in the radar literature \cite{richards2005fundamentals}. To detect targets of interest, it is typical to identify $Q$ peaks in the squared \ac{mf} output signal $\left|\RV{y}_{\rm s}^{\rm MF}(\tau)\right|^2$. The multi-target detection can be achieved by using adaptive thresholding algorithms such as \ac{cfar} detectors. Finally, the range estimates of these targets are obtained.


\section{Approximate Uncoded \ac{ber} Analysis Over Frequency-Flat Channels} \label{sec:performance}
In this section, we provide an approximate uncoded \ac{ber} analysis for \ac{acfk} over frequency-flat channels, namely for $|\V{h}_{\rm c}|^2\propto \V{1}$. Without loss of generality, we consider the constant channel gain $|\V{h}_{\rm c}|^2=\V{1}$, which is equivalent to incorporating $|\bar{h}|^2$ into the noise variance $\sigma_{\rm{c}}^2$. The purpose of this section is to provide a tractable benchmark, as well as an intuitive understanding, for the basic communication receiver design proposed in Sec.~\ref{ssec:comm_rx}, rather than an exact characterization of its performance, or that of the optimal receiver. The \ac{ber} performance of the \ac{acfk} can be characterized by the \ac{ber} of the symbol vectors $\RV{x}_{\rm{p}}$ and $\RV{x}_{\rm{s}}$. The overall \ac{ber} of \ac{acfk} can be expressed as a linear combination of these two components, while the corresponding combination coefficients depend on the modulation order in practical systems.


\subsection{The \ac{ber} of the \ac{acf} Symbol Vector}
When the non-negativity violation probability is small, the actual \ac{pacf} $\RV{r}_{\RV{x}}$ can be approximated as the nominal \ac{pacf} $\RV{x}_{\rm ACF}$. Since the channel is frequency-flat, we now have
\begin{subequations}
\begin{align}
\M{R} &\approx 2\sigma_{\rm{c}}^2\beta_{\rm ACF}\zeta_{\rm ACF}\M{S}{\rm circ}\left(\RV{x}_{\rm ACF}\right)\M{S}^{\rm H},\\
\M{C} &\approx 2\sigma_{\rm{c}}^2\beta_{\rm ACF}\zeta_{\rm ACF}\M{S}{\rm circ}\left(\RV{x}_{\rm ACF}\right)\M{F}_N^2\M{S}^{\rm H}.
\end{align}
\end{subequations}
Note that the diagonal entries of ${\rm circ}\left(\RV{x}_{\rm ACF}\right)$ are equal to one, while its off-diagonal entries are determined by the attenuated \ac{acf} sidelobes. Under the assumption that $\zeta_{\rm ACF}$ is sufficiently large, we may apply the low-sidelobe approximation of
\begin{equation}
\M{S}{\rm circ}\left(\RV{x}_{\rm ACF}\right)\M{S}^{\rm H}\approx \M{I},~\M{S}{\rm circ}\left(\RV{x}_{\rm ACF}\right)\M{F}_N^2\M{S}^{\rm H}\approx \M{0}.
\end{equation}
Thus, the \ac{acf}-domain demodulation model can be approximated as
\begin{equation}
\widetilde{\RV{y}}_{\rm demod} \approx \RV{x}_{\rm s} +\widetilde{\RV{n}}_1,
\end{equation}
where $\widetilde{\RV{n}}_1\sim \mathcal{CN}(\V{0},2\sigma_{\rm{c}}^2\beta_{\rm{ACF}}\zeta_{\rm{ACF}}\M{I})$.

Observe that this approximated model can be interpreted as a classical \ac{awgn} channel model \cite{proakis2001digital}. Using this interpretation, the \ac{snr} of this channel can be written as
\begin{equation}\label{effective_snr}
\ac{snr} = \frac{\mathbb{E}\left\{|{\rv{x}_{\rm{s}}}|^2\right\}}{2\sigma_{\rm{c}}^2\beta_{\rm ACF}\zeta_{\rm ACF}}.
\end{equation}
The corresponding uncoded \ac{ber} can then be evaluated using the standard detector under \ac{awgn} channels. In light of this, we refer to the \ac{snr} in \eqref{effective_snr} as the \emph{effective \ac{snr}} of the \ac{acf} symbols $\RV{x}_{\rm s}$. It is worth noting that the effective \ac{snr} is relatively low, since the communication symbols are modulated onto the sidelobes of $\RV{x}_{\rm{ACF}}$, whereas the mainlobe does not convey any information. Furthermore, the result reveals an inherent tradeoff between communication and sensing performance under the \ac{acfk} framework. Specifically, a larger sidelobe attenuation factor $\zeta_{\rm{ACF}}$ improves sensing performance, but degrades communication performance.

\subsection{The \ac{ber} of the Phase Symbol Vector}\label{ssec:ber_phase}
Next, we consider the phase symbol vector $\RV{x}_{\rm p}$. Note that the received signal on the $k$-th subcarrier can still be expressed in the form of \eqref{pskdemod}. In the practical receiver design described in Sec.~\ref{ssec:comm_rx}, $\rv{a}_k$ is unknown before the demodulation of $\RV{x}_{\rm s}$, thereby an estimate of $\RV{x}_{\rm s}$ is employed. To obtain a tractable high-\ac{snr} benchmark, in this subsection, we assume that $\RV{x}_{\rm s}$ has been perfectly recovered. Under this idealized assumption, we see that estimation of \ac{psk} symbols $\rv{x}_{{\rm p},k}$ under the model \eqref{pskdemod} can be interpreted as \ac{psk} demodulation over a fading channel with channel coefficients given by $\RV{x}_{\rm{a}}$.


Using the aforementioned interpretation, the power gain of the channel can be represented as
\begin{equation}
\RV{r}_{\rm h}=|\RV{x}_{\rm{a}}|^2 =\left|\sqrt{\sqrt{N}\M{F}_N\RV{x}_{\rm{ACF}}}\right|^2=\sqrt{N}\left|\M{F}_N\RV{x}_{\rm{ACF}}\right|.
\end{equation}
Note that for all $k=1,2,\cdots,N$, we have
$$
[\V{f}_k]_1[\RV{x}_{\rm{ACF}}]_1 = \frac{1}{\sqrt{N}}.
$$
According to the definition of \ac{dft}, we can obtain
\begin{equation}
\widetilde{\V{f}}_k^{\rm T}\widetilde{\RV{x}}_{\rm{ACF}} = \frac{1}{\sqrt{N}}\sum_{n=2}^N [\widetilde{\RV{x}}_{\rm{ACF}}]_n e^{-j\theta_{k,n}}.
\end{equation}
where we denote $\theta_{k,n}=\frac{2\pi(k-1)(n-1)}{N}$ and
$$
\begin{aligned}
\widetilde{\V{f}}_k  &= [[\V{f}_k]_2,\dotsc, [\V{f}_k]_N]^{\rm T}, \\
\widetilde{\RV{x}}_{\rm{ACF}}  &=[[\RV{x}_{\rm{ACF}}]_2,\dotsc,[\RV{x}_{\rm{ACF}}]_N]^{\rm T}.
\end{aligned}
$$
Using the conjugate symmetry of $\RV{x}_{\rm{ACF}}$, we further obtain
\begin{equation}
\widetilde{\V{f}}_k^{\rm T}\widetilde{\RV{x}}_{\rm{ACF}} =\frac{2}{\sqrt{N}}\sum_{n=2}^{1+\lfloor\frac{N-1}{2}\rfloor} \rv{z}_n,
\end{equation}
where
\begin{equation}
\rv{z}_n={\rm Re}\left\{[\widetilde{\RV{x}}_{\rm{ACF}}]_n e^{-j\theta_{k,n}}\right\}=\rv{r}_n\cos\left(\theta_{k,n}\right)+\rv{i}_n\sin\left(\theta_{k,n}\right),
\end{equation}
with $\rv{r}_n = {\rm Re}\{[\widetilde{\RV{x}}_{\rm{ACF}}]_n\}$ and $\rv{i}_n = {\rm Im}\{[\widetilde{\RV{x}}_{\rm{ACF}}]_n\}$. Thus, we have
\begin{subequations}\label{power_spec}
\begin{align*}
\left[\M{F}_N\RV{x}_{\rm{ACF}}\right]_k&=\V{f}_k^{\rm T}\RV{x}_{\rm{ACF}}=[\V{f}_k]_1[\RV{x}_{\rm{ACF}}]_1+\widetilde{\V{f}}_k^{\rm T}\widetilde{\RV{x}}_{\rm{ACF}}\\
&=\frac{1}{\sqrt{N}}+\frac{2}{\sqrt{N}}\sum_{n=2}^{1+\lfloor\frac{N-1}{2}\rfloor}\rv{z}_n. \tag{\theequation}
\end{align*}
\end{subequations}
The exact distribution of $\left[\M{F}_N\RV{x}_{\rm{ACF}}\right]_k$ is clearly intractable. We observe that $\rv{z}_n$ are independent, non-identically distributed and uniformly bounded real random variables. By leveraging the Lindeberg-Feller central limit theorem, the sum of such random variables can be approximated by a Gaussian distribution for large $N$. To proceed, we first introduce the definition of rotational symmetry.
\begin{definition}[Rotational Symmetry]
The expectation and pseudo-variance of the constellation are zero, namely
\begin{equation}
\mathbb{E}\{\rv{s}\}=0,\quad \mathbb{E}\{\rv{s}^2\}=0.
\end{equation}
\end{definition}
It is worth noting that most practical constellations satisfy this property, including all the \ac{psk} and \ac{qam} constellations except for \ac{bpsk} and 8-\ac{qam}.

Now, from \eqref{power_spec}, we have
\begin{equation}
\sqrt{N}\left[\M{F}_N\RV{x}_{\rm{ACF}}\right]_k=1+2\sum_{n=2}^{1+\lfloor\frac{N-1}{2}\rfloor}\rv{z}_n.
\end{equation}
Assuming that the constellation $\Set{S}_{\rm{s}}$ is rotationally symmetric, we can obtain
\begin{equation}
\mu_{\rv{z}}=\mathbb{E}\left\{\rv{z}_n\right\}=\mathbb{E}\{\rv{r}_n\}\cos\left(\theta_{k,n}\right)+\mathbb{E}\{\rv{i}_n\}\sin\left(\theta_{k,n}\right)=0, \label{mu}
\end{equation}
and
\begin{subequations}
\begin{align}
\sigma^2_{\rv{z}}&=\mathbb{E}\left\{|\rv{z}_n|^2\right\}\\
&= \mathbb{E}\{\rv{r}_n^2\}\cos^2\left(\theta_{k,n}\right)+\mathbb{E}\{\rv{i}_n^2\}\sin^2\left(\theta_{k,n}\right)\\
&\quad+2\mathbb{E}\{\rv{r}_n\rv{i}_n\}\cos\left(\theta_{k,n}\right)\sin\left(\theta_{k,n}\right)\\
&=\frac{1}{2}\mathbb{E}\left\{|[\widetilde{\RV{x}}_{\rm{ACF}}]_n|^2\right\} = \frac{\mathbb{E}\{|\rv{x}_{\rm{s}}|^2\}}{2N\beta_{\rm{ACF}}\zeta_{\rm{ACF}}},\label{sigma}
\end{align}
\end{subequations}
where \eqref{mu} follows from the zero-mean property of the constellation, namely $\mathbb{E}\{[\widetilde{\RV{x}}_{\rm{ACF}}]_n\}=\mathbb{E}\{\rv{r}_n\}+j\mathbb{E}\{\rv{i}_n\}=0$, while \eqref{sigma} follows from the zero pseudo-variance property of the constellation, namely $\mathbb{E}\{[\widetilde{\RV{x}}_{\rm{ACF}}]_n^2\}=\mathbb{E}\{\rv{r}^2_n\}-\mathbb{E}\{\rv{i}^2_n\}+2j\mathbb{E}\{\rv{r}_n\rv{i}_n\}=0$. It is worth noting that for all $k=1,2,\cdots,N$, the random variables $\rv{z}_n$ have identical means and variances.

Since $\rv{z}_n$ are uniformly bounded real random variables, these random variables satisfy the Lindeberg condition, given by
\begin{equation}
\forall \epsilon>0, \quad \lim_{N\to\infty} \frac{1}{\sigma_N^2} \sum_{n=2}^{1+\lfloor\frac{N-1}{2}\rfloor} \int_{|\rv{z}_n| \ge \epsilon \sigma_N} \rv{z}_n^2 {\rm{d}}\mathbb{P} = 0
\end{equation}
where
\begin{equation}
\sigma_N^2=\sum_{n=2}^{1+\lfloor\frac{N-1}{2}\rfloor}\sigma_{\rv{z}}^2=
\begin{cases}
\frac{(N-1)\mathbb{E}\left\{|\rv{x}_{\rm{s}}|^2\right\}}{4N\beta_{\rm{ACF}}\zeta_{\rm ACF}}, & N \text{ is odd},\\[7pt]
\frac{(N-2)\mathbb{E}\left\{|\rv{x}_{\rm{s}}|^2\right\}}{4N\beta_{\rm{ACF}}\zeta_{\rm ACF}}, & N \text{ is even}.
\end{cases}
\end{equation}
Therefore, according to the Lindeberg-Feller central limit theorem, the sum of $\rv{z}_n$'s is asymptotically distributed as a real Gaussian distribution with zero mean and variance $\sigma_N^2$. Neglecting inter-symbol dependence, we further approximate the nominal power spectrum $\sqrt{N}\M{F}_N\RV{x}_{\rm{ACF}}$ by a real Gaussian distribution, namely $\sqrt{N}\M{F}_N{\RV{x}}_{\rm{ACF}} \sim \mathcal{N}\left(\V{1}_{N},\sigma^2\M{I}_N\right)$, with
\begin{equation}\label{sigma_def}
\sigma^2=
\begin{cases}
\frac{(N-1)\mathbb{E}\left\{|\rv{x}_{\rm{s}}|^2\right\}}{N\beta_{\rm{ACF}}\zeta_{\rm ACF}}, & N \text{ is odd},\\[7pt]
\frac{(N-2)\mathbb{E}\left\{|\rv{x}_{\rm{s}}|^2\right\}}{N\beta_{\rm{ACF}}\zeta_{\rm ACF}}, & N \text{ is even}.
\end{cases}
\end{equation}
when $N$ is sufficiently large.

We may now conclude that, for rotationally symmetric constellations, each entry of $\RV{r}_{\rm h}$ approximately follows a folded normal distribution, with its probability density function given by
\begin{equation}\label{folded_normal}
p(\rv{r}_{\rm h})=\frac{1}{\sqrt{2\pi\sigma^2}}e^{-\frac{(\rv{r}_{\rm h}-1)^2}{2\sigma^2}}+\frac{1}{\sqrt{2\pi\sigma^2}}e^{-\frac{(\rv{r}_{\rm h}+1)^2}{2\sigma^2}}
\end{equation}
Then, the \ac{mgf} of this distribution can be derived as
\begin{subequations}
\begin{align*}
M_{\rv{r}_{\rm h}}(t) &= \mathbb{E}_{\rv{r}_{\rm h}}\{e^{t\rv{r}_{\rm h}}\}
=\int_0^{+\infty}e^{t\rv{r}_{\rm h}}p(\rv{r}_{\rm h}){\rm d}\rv{r}_{\rm h} \\
&= \frac{1}{2} e^{-\frac{1}{2\sigma^2}} \!\left\{ {\rm{erfcx}}\left(\frac{-1-\sigma^2 t}{\sigma\sqrt{2}}\right) \!+\!{\rm{erfcx}}\left(\frac{1-\sigma^2 t}{\sigma\sqrt{2}}\right) \!\right\}, \tag{\theequation}
\end{align*}
\end{subequations}
where
\begin{equation}
{\rm{erfcx}}(x) = e^{x^2} {\rm{erfc}}(x)=e^{x^2} \cdot \frac{2}{\sqrt{\pi}}\int_{x}^{\infty} e^{-t^2} {\rm d}t,
\end{equation}
with ${\rm{erfc}}(x)$ denoting complementary error function.

Next, we derive a closed-form approximation of the \ac{ber}. The \ac{ser} of $M_p$-\ac{psk} over an \ac{awgn} channel can be expressed as \cite{proakis2001digital}
\begin{equation}
P_{\rm {a}}(\rho) = \frac{1}{\pi} \int_{0}^{\frac{(M_p-1)\pi}{M_p}} \exp\left(-\rho \cdot \frac{\sin^2\left(\frac{\pi}{M_p}\right)}{\sin^2 \theta} \right) {\rm d}\theta,
\end{equation}
where $\rho=\frac{1}{\sigma^2_c}$ denotes the \ac{snr}, and $P_{\rm a}(\rho)$ denotes the \ac{ser} of \ac{psk} over an \ac{awgn} channel. $M_p\in\mathbb{Z}_+$ denotes the size of \ac{psk} constellations. By the assumption that the channel coefficient $\RV{x}_{\rm{a}}$ is perfectly known, the instantaneous \ac{snr} over the effective fading channel can be expressed as $\rho_{\rm h}=\rv{r}_{\rm h}\rho=\frac{\rv{r}_{\rm h}}{\sigma^2_c}$. Then, the \ac{ser} over the effective fading channel can be computed as
\begin{subequations}
\begin{align*}
P_{\rm h}(\rho)&= \mathbb{E}_{\rv{r}_{\rm h}}\{P_{\rm a}(\rho_{\rm h})\}\\
&=\frac{1}{\pi} \int_{0}^{\frac{(M_p-1)\pi}{M_p}} \mathbb{E}\left\{\exp\left(-\rv{r}_{\rm h}\rho \cdot \frac{\sin^2\left(\frac{\pi}{M_p}\right)}{\sin^2 \theta} \right)\right\} {\rm d}\theta\\
&=\frac{1}{\pi} \int_{0}^{\frac{(M_p-1)\pi}{M_p}} M_{\rv{r}_{\rm h}}\left(-\rho \cdot \frac{\sin^2\left(\frac{\pi}{M_p}\right)}{\sin^2 \theta}\right) {\rm d}\theta. \tag{\theequation}
\end{align*}
\end{subequations}
A simple one-dimensional numerical integration can thus be employed to compute $P_{\rm h}(\rho)$. When the Gray code mapping is adopted to label $M_p$-\ac{psk} symbols, the \ac{ber} can further be approximately computed from the \ac{ser} in the high-\ac{snr} regime, given by
\begin{equation}
\ac{ber} \approx \frac{P_{\rm {h}}(\rho)}{\log_2 M_p}.
\end{equation}

\begin{remark} \label{prop:deep_fading}
Finally, we would like to discuss the relationship between the non-negativity violation and \ac{acfk}-induced deep fading. As shown in \eqref{folded_normal}, the effective fading channel, caused by taking the magnitude of the Fourier transform of the \ac{acf} symbol vector, follows a folded normal distribution $P(\rv{r}_h)$. Since the probability mass of the normal distribution over $(-\infty,0]$ is folded onto $[0,\infty)$ , the non-negativity violation raises the occurrence probability of deep fading. In fact, the ``folding effect"  would eventually result in \ac{ber} deterioration of \ac{psk} symbols, especially in the high-\ac{snr} regime. To elaborate, when the communication noise becomes considerably weak, the conditional \ac{ber} in the case of high instantaneous channel gain is negligible, and thus the low instantaneous channel gain case becomes dominant in terms of \ac{ber}. This behavior corresponds exactly to the folding effect producing substantial probability mass in the low channel gain regime, marked by the non-negativity violation probability. Therefore, the non-negativity violation probability can be viewed as a quality indicator of the \ac{psk} sub-channel.
\end{remark}

\begin{figure*}[t]
	\centering
	\includegraphics[width=0.95\linewidth]{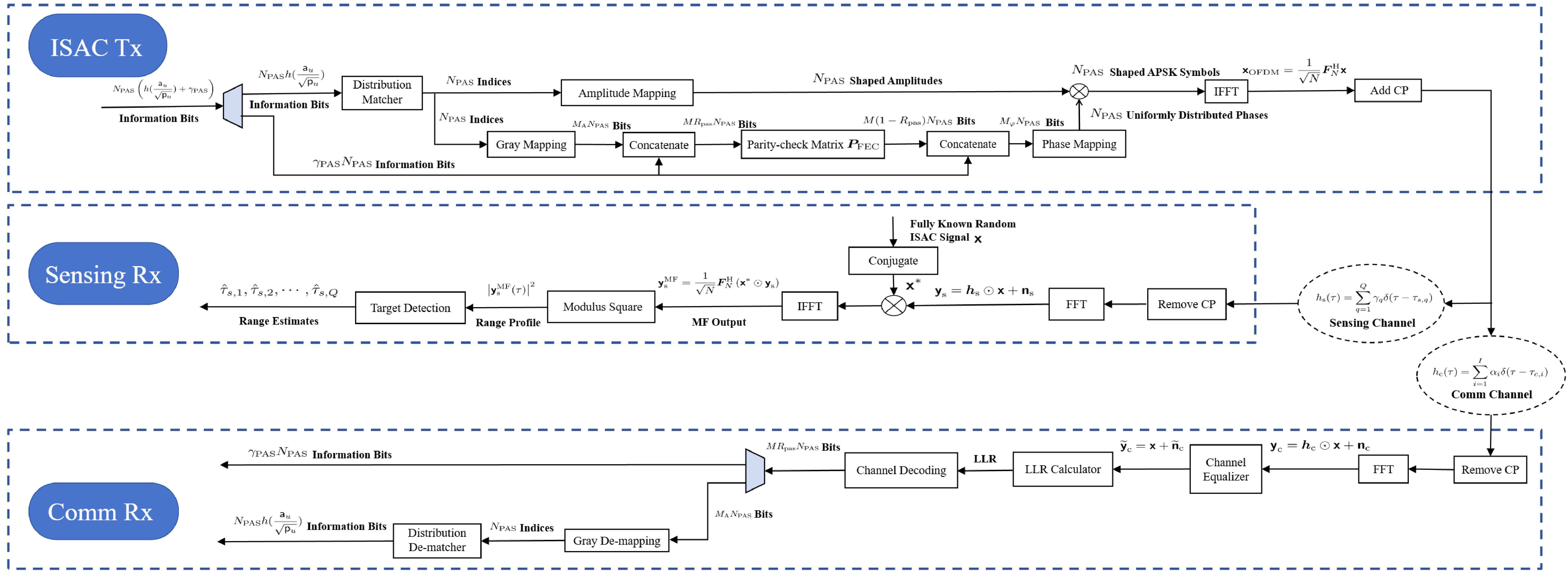}
	\caption{The signal processing pipeline for the generalized \ac{pas} system over the quasi-static multipath channels.}
	\label{fig:pas_system}
\end{figure*}

\section{Numerical Results} \label{sec:simulations}
In this section, we present numerical results to validate our theoretical analysis. Unless otherwise specified, we consider \ac{acfk} with $N=1951$ subcarriers. Both uncoded and coded transmission schemes are evaluated in the simulations. For the channel coding, we employ the \ac{ldpc} code with codeword length $n=16200$ and rate $R_{\rm c}^{\rm ACF}=\frac{2}{3}$. For the implementation of \ac{acfk}, the constellations $\Set{S}_{\rm{p}}$ and $\Set{S}_{\rm{s}}$ are configured as 64-\ac{psk} and 32-\ac{qam}, respectively. The sidelobe attenuation factor is set to $\zeta_{\rm{ACF}}=10~{\rm{dB}}$. Simulations are carried out over both the \ac{awgn} channel and a $3$-path Rician fading channel with Rician factor $K_{r}=10~{\rm{dB}}$. Under these configurations, the \ac{acfk} system is capable of carrying
\begin{equation}
R_{\rm{total}}^{\rm ACF}=\frac{1}{1951}(1951\cdot\log_2 64+\frac{1950}{2}\cdot\log_2 32) = 8.4987
\end{equation}
bits per subcarrier, while the actual \ac{ldpc} code requires
\begin{equation}
R_{\rm{used}}^{\rm ACF}=\frac{16200}{1951} = 8.3034
\end{equation}
bits per subcarrier. The adaptation between these two rates is implemented by zero-padding. Finally, all the simulation results are obtained by averaging over $10000$ random realizations.

\subsection{Baseline: Generalized Probabilistic Amplitude Shaping}
To facilitate understanding of numerical results, we first elaborate on the baseline scheme before presenting simulations. To achieve a beneficial tradeoff between sensing and communication performance, recent research has explored \ac{pcs} schemes to optimize the input distribution of \ac{ofdm} signals \cite{du2024reshaping, zhao2026input, 11417988}. In this paper, we adopt the generalized \ac{pas} scheme proposed in \cite{11417988} as the baseline scheme, which will also be referred to as the ``\ac{ofdm}-\ac{pcs}'' scheme hereafter. The generalized \ac{pas} extends the conventional \ac{pas} framework presented in \cite{7307154} from \ac{ask} to \ac{apsk} constellations, thus yielding circularly symmetric constellations that are more suitable for sensing task. Readers are referred to \cite{11417988} for more technical details. It is noteworthy that \cite{11417988} only provides conceptual design of the generalized \ac{pas} which is insufficient for \ac{ber} evaluation. For the completeness of presentation, we consider a specific implementation detailed as follows.

The signal processing pipeline for the generalized \ac{pas} system is illustrated in Fig.~\ref{fig:pas_system}. We employ the \ac{ldpc} code with codeword length $n=16200$ and rate $R_{\rm c} ^{\rm PAS}>\frac{2}{5}$ for channel coding. In view of the transmitted symbol vector length determined by codeword length, we configure an \ac{ofdm}-\ac{pcs} system with $N_{\rm{PAS}}=1620$ subcarriers for simulations. Note that this is different from the configuration in the \ac{acfk} case where $N=1951$. In order to ensure fairness of comparison, we consider the rescaled \ac{esl}, defined as $N \cdot\ac{esl}$, as the sensing performance evaluation metric, since the \ac{esl} decreases at a rate of $O(\frac{1}{N})$ as $N$ grows large \cite{11087656}. For the constellation design, we first design a 1024-\ac{apsk} constellation $\Set{S}_3$ with $M_{\rm A}=4$ amplitude bits and $M_{\rm \varphi}=6$ phase bits, as shown in Fig.~\ref{fig:apsk1024}. These \ac{apsk} symbols can be expressed by the amplitude-phase factorization $\rv{x}_u=\frac{\rv{a}_u}{\sqrt{\rv{p}_u}}\cdot \rv{s}_u$, where the amplitude $\rv{a}_u\in[1,2,\cdots,16]$ can be shaped independently, while the phase $\rv{s}_u$ is uniformly distributed on a unit circle. $\rv{p}_u=\mathbb{E}\left\{|\rv{x}_u|^2\right\}, \rv{x}_u\in\Set{S}_3$ denotes the mean power of the constellation.

Next, since the \ac{esl} of \ac{ofdm} modulation is related to the kurtosis of constellations \cite{11087656, du2024reshaping}, we optimize the input distribution by solving the following entropy maximization problem, formulated as
\begin{subequations}\label{pcs_islr}
\begin{align}
\max_{t\in[t_{\min},t_{\max}]}~\max_{p_{\rv{x}_u}(x_u)}&~~H\left(\rv{x}_u\right),\\
{\rm s.t.}&~~ \mathbb{E}\{|\rv{x}_u|^4\}\leq \zeta_{\rm PAS} t^2(\mu_4-1)+t^2,\\
&~~\mathbb{E}\{|\rv{x}_u|^2\}=t,\\
&~~\rv{x}_u\in\Set{S}_3,
\end{align}
\end{subequations}
where $t_{\min}= \min_{\rv{x}_u\in \Set{S}_3}|\rv{x}_u|^2$ and $t_{\max}= \max_{\rv{x}_u\in \Set{S}_3}|\rv{x}_u|^2$. $\mu_4$ denotes the kurtosis of the constellation, defined as
\begin{equation}
\mu_4 \triangleq \frac{\mathbb{E}\left\{ \left| \rv{s} - \mathbb{E}(\rv{s}) \right|^4 \right\}}{\mathbb{E}\left\{ \left| \rv{s} - \mathbb{E}(\rv{s}) \right|^2 \right\}^2}.
\end{equation}
$\zeta_{\rm pas}$ denotes the sidelobe reduction ratio, which serves to ensure identical rescaled \ac{esl} for both the \ac{acfk} and generalized \ac{pas} systems, thus facilitating fair communication performance evaluation. Clearly, the inner subproblem is a convex optimization problem, which may be efficiently solved by off-the-shelf convex optimization solvers \cite{grant2014cvx}, while the outer subproblem may be addressed through a one-dimensional line search. Finally, the optimal amplitude distribution $p(\frac{\rv{a}_u}{\sqrt{\rv{p}_u}})$, derived from the optimal input distribution $p(\rv{x}_u)$, can be implemented by \ac{ccdm} \cite{7322261}.

\begin{figure}[t]
    \vspace{-3mm}
	\centering
    \subfloat[The original normalized 1024-\ac{apsk} constellation]{
	\includegraphics[width=0.42\linewidth]{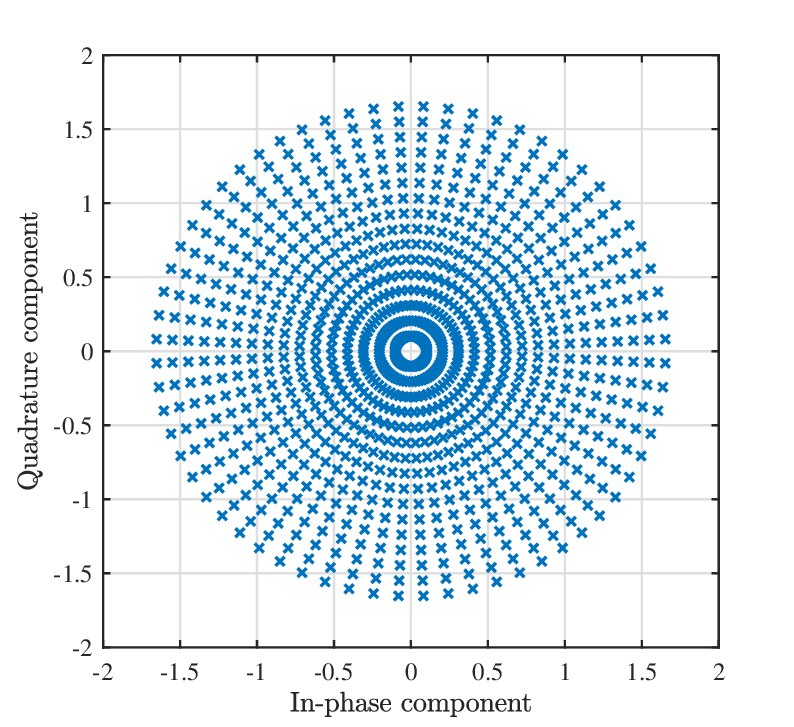}
    \label{fig:apsk1024}
    }\hspace{0.5mm}
    \subfloat[The maximal-entropy input distribution $p(\rv{x}_u)$]{
	\includegraphics[width=0.48\linewidth]{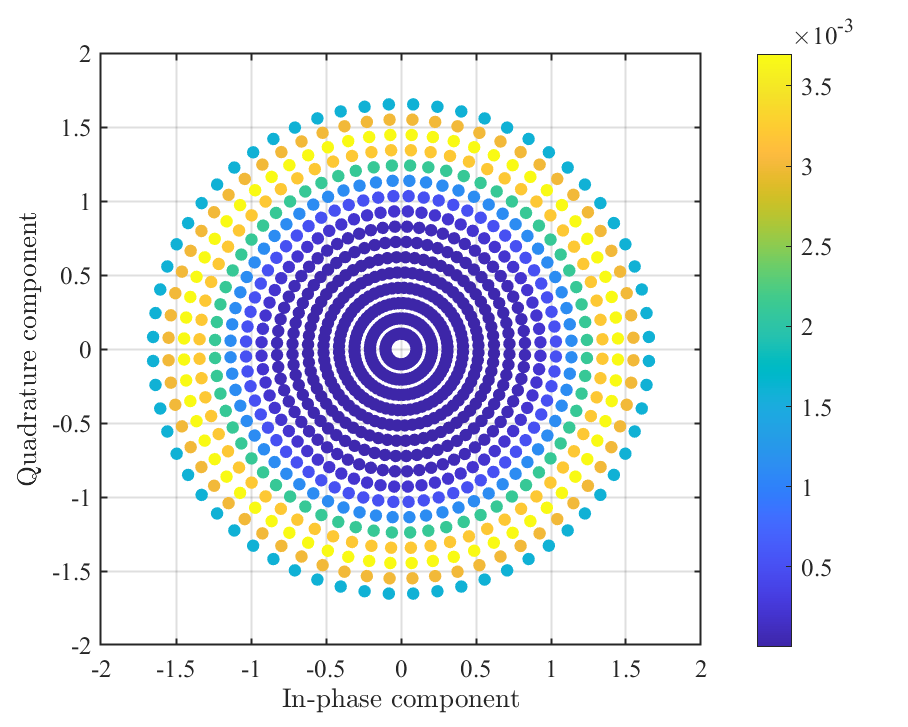}
    \label{fig:pcs_apsk1024}
    }
	\caption{The original normalized 1024-\ac{apsk} constellation and the probabilistically weighted constellation obtained by generalized \ac{pas} solving \eqref{pcs_islr}.}
\end{figure}

Under these configurations, the transmission rate per subcarrier is expressed as
\begin{equation} \label{pas_rate}
R^{\rm{PAS}}=\frac{H(\frac{\rv{a}_u}{\sqrt{\rv{p}_u}})+\gamma_{\rm{PAS}}}{R_{\rm c}^{\rm{PAS}}}=\frac{H(\frac{\rv{a}_u}{\sqrt{\rv{p}_u}})+10 R_{\rm c}^{\rm{PAS}}-4}{R_{\rm c}^{\rm{PAS}}}~{\rm{bits}},
\end{equation}
where $0<H(\frac{\rv{a}_u}{\sqrt{\rv{p}_u}})\leq M_{\rm{A}}=4$ denotes the entropy of the normalized amplitude. The parameter $0<\gamma_{\rm{PAS}}=10 R_{\rm c}^{\rm{PAS}}-4\leq 6$ governs the number of information bits that bypass the \ac{ccdm}. For a fixed optimal input distribution $p(\rv{x}_u)$, the transmission rate $R^{\rm{PAS}}$ is an increasing function of the code rate $R_{\rm{c}}^{\rm PAS}$. In addition, the entropy $H(\frac{\rv{a}_u}{\sqrt{\rv{p}_u}})$ implemented by \ac{ccdm} is typically lower than the theoretical optimum, due to the finite \ac{ccdm} block length $N_{\rm{PAS}}=1620$.

We next show the optimal input distribution $p(\rv{x}_u)$ derived from the optimization problem \eqref{pcs_islr}, as depicted in Fig.~\ref{fig:pcs_apsk1024}. The colors of these constellation points represent the corresponding symbol probabilities. The constellation exhibits a ring-like shape with a kurtosis of $1.0588$, which implies favorable sensing performance. The theoretical entropies of the optimal amplitude and input distributions are calculated as $H(\frac{\rv{a}_u}{\sqrt{\rv{p}_u}})=2.7482$ bits and $H(\rv{x}_u)=8.7482$ bits, respectively, while the achievable amplitude entropy $H(\frac{\rv{a}_u}{\sqrt{\rv{p}_u}})$ by \ac{ccdm} is $2.7222$ bits.

\subsection{\texorpdfstring{\ac{ber}}{BER} Performance}
We first demonstrate the \ac{ber} of the uncoded \ac{acfk} system over \ac{awgn} and Rician fading channels in Fig.~\ref{fig:ber_awgn_un} and Fig.~\ref{fig:ber_rician_un}, respectively.  In Fig.~\ref{fig:ber_awgn_un}, it can be observed that the analytical values agree well with the simulation results in the high-\ac{snr} regime, which validates the correctness of our derivation in Sec. \ref{sec:performance}. The $P_{\rm{neg}}$ curve in the legend corresponds to the non-negativity violation probability. It is worth noting that the \ac{ber} curve of 64-\ac{psk} in the \ac{acfk} system exhibits a prominent deviation from its standard 64-\ac{psk} counterpart over the same \ac{awgn} channel, particularly in the region below the $P_{\rm{neg}}$ curve. This observation confirms that the non-negativity violation probability indeed serves as a valid quality indicator of the \ac{psk} sub-channel. As discussed in Sec.~\ref{ssec:ber_phase}, the \ac{ber} of \ac{psk} symbols in the high-\ac{snr} region is dominated by the low instantaneous channel gain case, which is closely related to the non-negativity violation probability $P_{\rm{neg}}$, due to the ``folding effect" induced by taking the magnitude of the symbols. In addition, the total \ac{ber} over the Rician channel is much higher than that over the \ac{awgn} channel.

 \begin{figure}[t]
     \centering
     \includegraphics[width=0.8\linewidth]{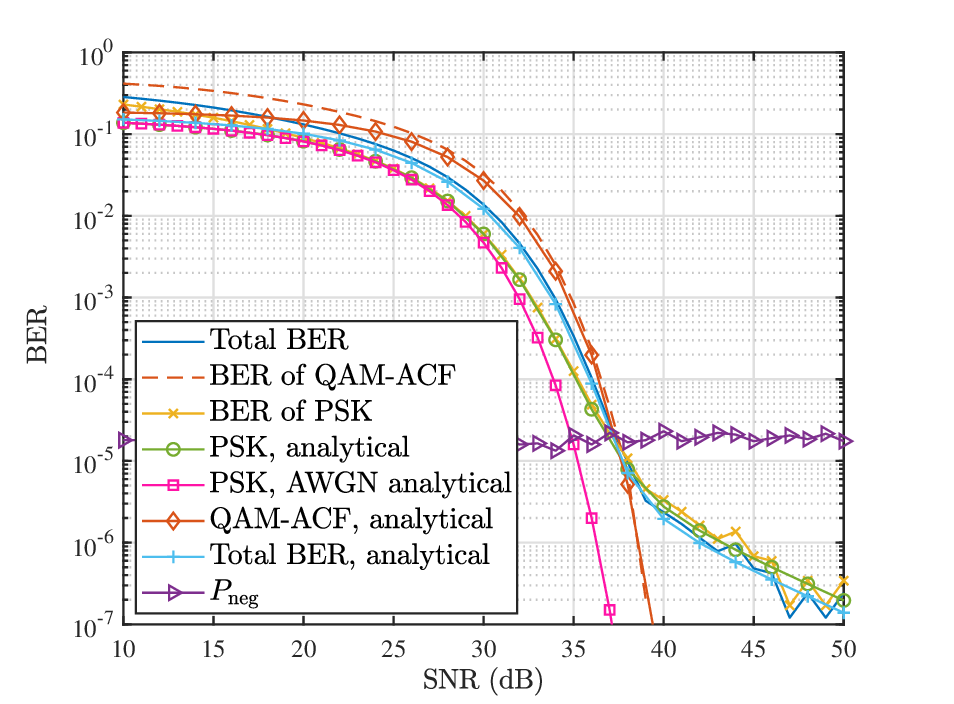}
     \caption{The \ac{ber} performance of the uncoded \ac{acfk} system over \ac{awgn} channels, with analytical and simulation results.}
 \label{fig:ber_awgn_un}
 \end{figure}

 \begin{figure}[t]
     \centering
     \includegraphics[width=0.8\linewidth]{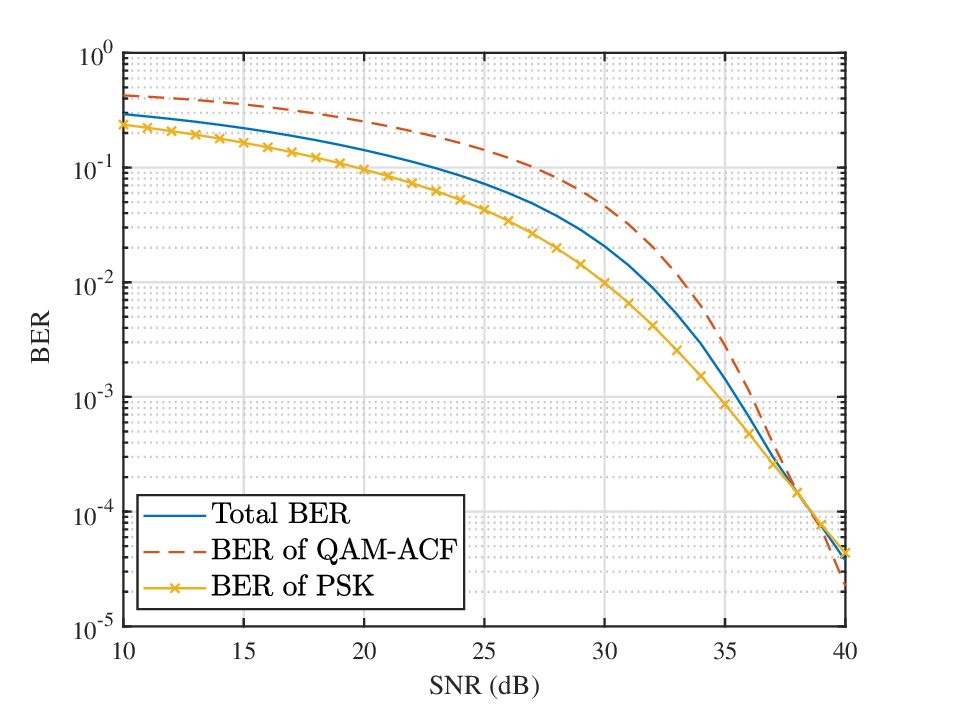}
     \caption{The \ac{ber} performance of the uncoded \ac{acfk} system over Rician channels with the Rician factor $K_{r}=10~{\rm{dB}}$.}
     \label{fig:ber_rician_un}
 \end{figure}

Next, we illustrate the \ac{ber} of the coded \ac{acfk} system over \ac{awgn} and Rician fading channels in Fig.~\ref{fig:ber_all}, respectively. It is clear that the coded scheme achieves significantly better performance than its uncoded counterpart. The \ac{ber} exhibits a steep decline at about $\ac{snr} = 23~{\rm{dB}}$ over the \ac{awgn} channel. This phenomenon is termed as the waterfall effect. Then, we focus on the \ac{ber} performance comparison with the generalized \ac{pas} systems. To ensure consistent error correcting capability, the same \ac{ldpc} code with rate $R_{\rm c}^{\rm{PAS}}=R_{\rm{c}}^{\rm ACF}=\frac{2}{3}$ is employed for channel coding. However, according to \eqref{pas_rate}, the transmission rate per subcarrier implemented by \ac{ccdm} is computed as $R^{\rm{PAS}}=8.0833$ bits, which is lower than that of the coded \ac{acfk} system $R_{\rm{used}}^{\rm ACF}=8.3034$ bits. Note that such a difference is inevitable since the transmission rate in the \ac{pas} scheme is not flexibly adjustable. To this end, we carry out simulations under various code rates to improve the fairness of the comparison, as summarized in Table \ref{table}. As can be seen, compared with the \ac{ofdm}-\ac{pcs} system, \ac{acfk} achieves comparable coded \ac{ber} performance over \ac{awgn} channels at similar spectral efficiencies. Over Rician channels, \ac{acfk} exhibits a moderate \ac{snr} loss and a flatter waterfall slope, as portrayed in Fig.~\ref{fig:ber_all}\subref{fig:ber_rician}, which is consistent with the folded-normal effective fading mechanism caused by the \ac{acf}-domain amplitude construction.

 \begin{figure}[t]
 	\centering
     \subfloat[the \ac{awgn} channel]{
     \includegraphics[width=0.47\linewidth]{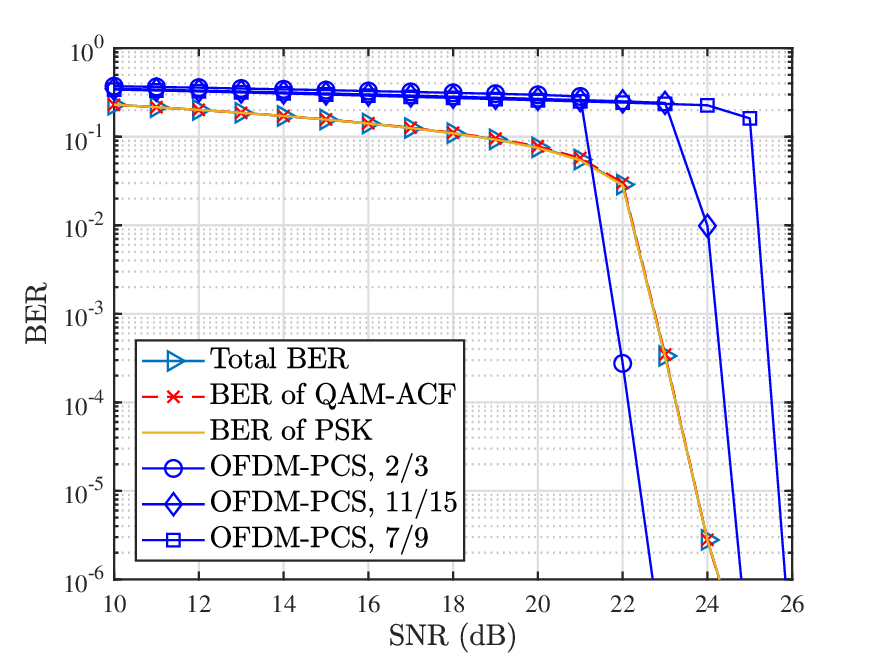}
     }
	\hfill
     \subfloat[the Rician channel]{
     \includegraphics[width=0.47\linewidth]{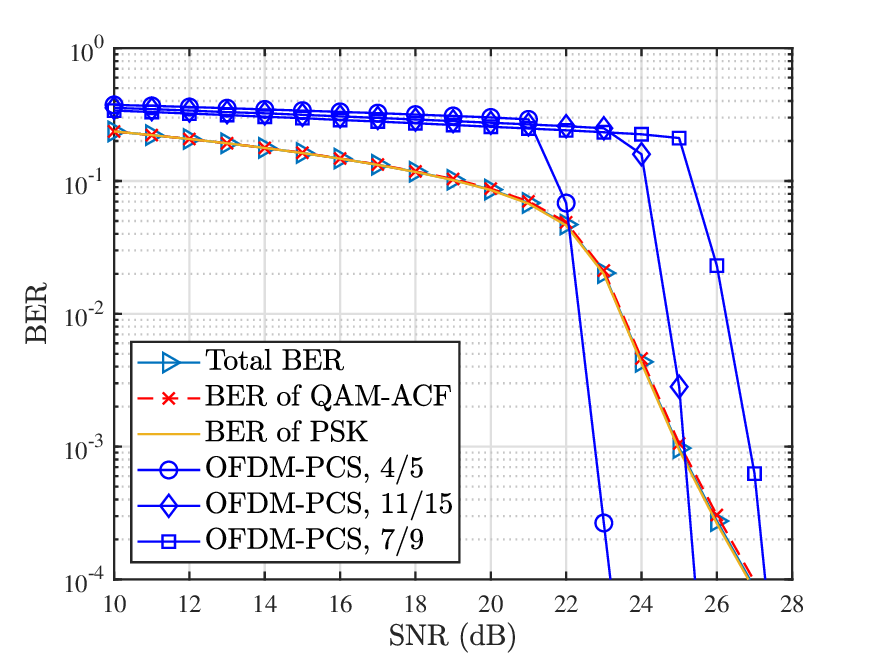}
     \label{fig:ber_rician}
     }
     \caption{The \ac{ber} performance comparison over \ac{awgn} and Rician channels between the coded \ac{acfk} system with the code rate $R_{\rm{c}}^{\rm ACF}=\frac{2}{3}$ and the generalized \ac{pas} system with the code rates $R_{\rm{c}}^{\rm PAS}=\frac{2}{3}, \frac{11}{15}$ and $\frac{7}{9}$.}
     \label{fig:ber_all}
 \end{figure}


  \begin{table}[!t]
 	\begin{center}
 		\caption{The transmission rates per subcarrier under various code rates}\label{table}
 		\renewcommand\arraystretch{1.2}
 		\begin{tabular}{|c|c|c|}
 			\hline
 			\textbf{Systems} & \textbf{Code rates} & \textbf{Transmission rates per subcarrier} \\
 			\cline{1-3}
 			\ac{acfk} & $2/3$ & $8.3034~{\rm{bits}}$ \\
 			\cline{1-3}
 			\multirow{3}{*}{\ac{ofdm}-\ac{pcs}} & $2/3$ & $8.0833~{\rm{bits}}$  \\
 			\cline{2-3}
 			 & $11/15$ & $8.2576~{\rm{bits}}$\\
 			\cline{2-3}
 			 & $7/9$ & $8.3571~{\rm{bits}}$\\
 			 \cline{1-3}
 		\end{tabular}
 	\end{center}
 \end{table}

 \begin{figure}[t]
	\centering
	\includegraphics[width=0.8\linewidth]{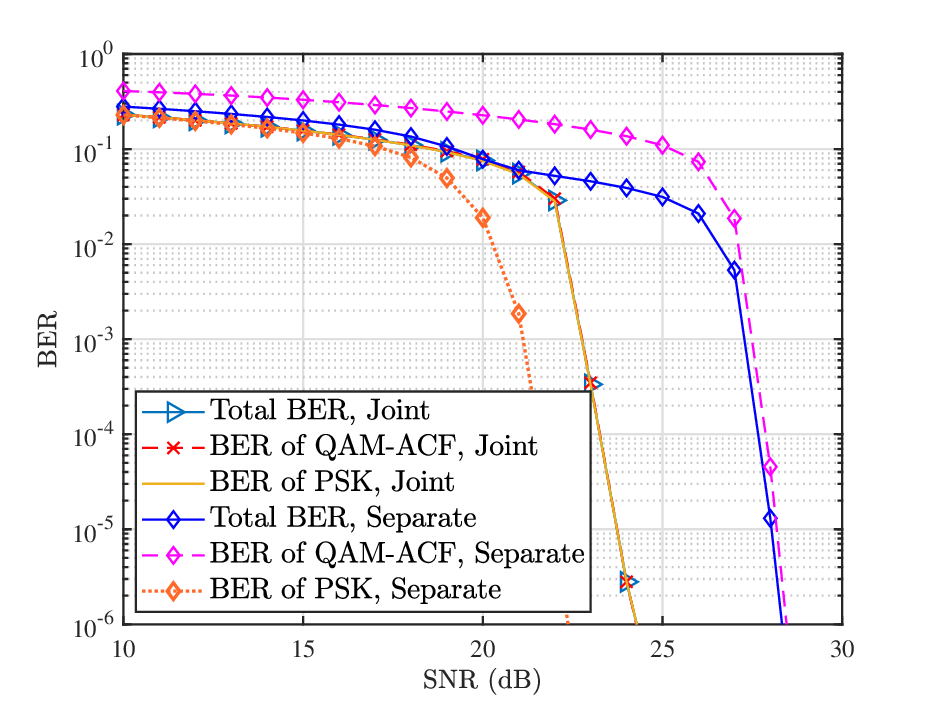}
	\caption{The \ac{ber} performance comparison between the joint coding and separate coding schemes with the same \ac{ldpc} code over \ac{awgn} channels.}
	\label{fig:separate_coding}
\end{figure}

Finally, we examine the \ac{ber} performance between the joint coding and separate coding schemes, as portrayed in Fig.~\ref{fig:separate_coding}. Note that the joint coding significantly outperforms the separate coding by $4~{\rm{dB}}$, which validates the reliability of the \ac{acfk} transceiver design in Sec. \ref{sec:architecture}. Joint coding achieves better total \ac{ber} performance by balancing the inferior \ac{ber} performance of 32-\ac{qam} and the superior \ac{ber} performance of 64-\ac{psk} in the separate coding scheme. Specifically, the gain brought by the joint decoding for 32-\ac{qam} outweighs the loss caused by the joint decoding for 64-\ac{psk}. 

\subsection{Target Detection Performance}
To evaluate the sensing performance of \ac{acfk}, we first depict the empirical \ac{cdf} of the \ac{pslr} of \ac{pacf} for the \ac{acfk} and generalized \ac{pas} systems in Fig.~\ref{fig:pas_cdf}. Under the identical \ac{esl} constraint, the \ac{pslr} under the generalized \ac{pas} system is $5 \!\sim\! 10~{\rm{dB}}$ higher than the well-controlled \ac{pslr} achieved by the \ac{acfk} system. This result confirms that the generalized \ac{pas} system has the capability of precisely controlling the \ac{esl}, whereas it fails to control the \ac{pslr}.

We next demonstrate how the \ac{pslr} reduction translates into target detection performance gain. In particular, we consider a two-target scenario, where two targets, one stronger and one weaker, are located at $14.065~{\rm{m}}$ and $28.125~{\rm{m}}$ away from the \ac{tx}, respectively. The bandwidth of the waveform is configured as $800~{\rm{MHz}}$. The echo power from the weaker target is fixed at $0~{\rm{dB}}$, while we denote the echo power from the stronger target as $\alpha_{\rm strong}~{\rm{dB}}$. To fairly evaluate the interference that the sidelobe of the stronger target inflicts on the weaker target, we impose the same \ac{esl} constraints on both the \ac{acfk} and generalized \ac{pas} systems in this subsection.  Accordingly, we define the \ac{sir} as the ratio between the echo power from the weaker target $(0~{\rm{dB}})$ and the average sidelobe response level of the strong target, which can be expressed as
\begin{equation}
{\rm SIR}=-(\text{\ac{esl}}+\alpha_{\rm strong})~{\rm{dB}}.
\end{equation}

Besides, we focus on the effective \ac{snr} of the weak target, defined as $\frac{P}{\sigma^2_{\rm{s}}}$, where $P$ represents the transmit power. Since the \ac{mf} processing at the sensing \ac{rx} yields a \ac{snr} gain proportional to the number of subcarriers, and the \ac{acfk} and generalized \ac{pas} adopt different subcarrier configurations, we employ the post-\ac{mf} \ac{snr} for fair comparison of detection performance, which can be expressed as
\begin{equation}
\ac{snr} = (10\log_{10}\frac{P}{\sigma^2_{\rm{s}}} + 10\log_{10}N_{\rm{sub}})~{\rm{dB}},
\end{equation}
where $10\log_{10}N_{\rm{sub}}$ corresponds to the \ac{mf}-induced \ac{snr} gain and $N_{\rm{sub}}$ denotes the number of subcarriers of the evaluated system. The \ac{cacfar} detector is adopted at the sensing \ac{rx}.

 \begin{figure}[t]
	\centering
	\includegraphics[width=0.8\linewidth]{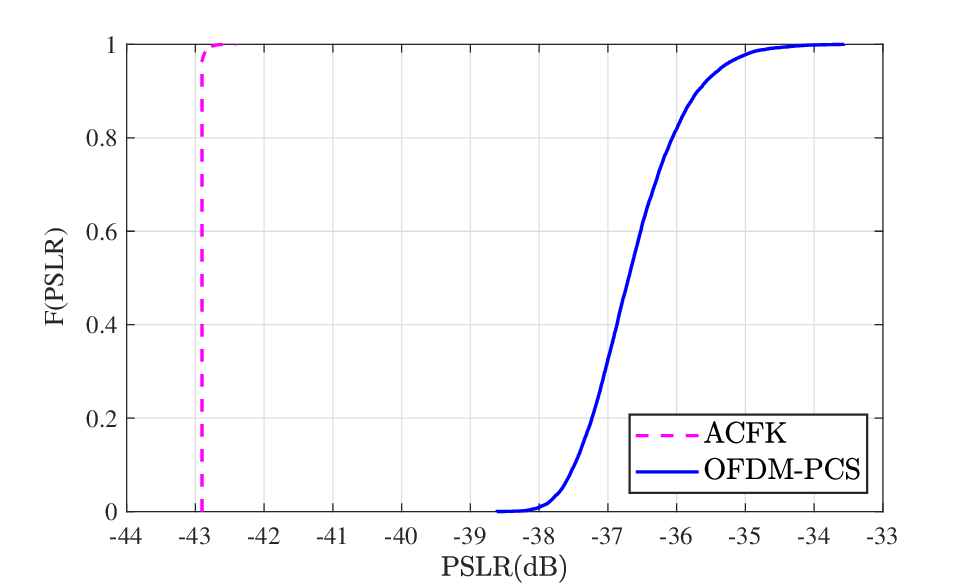}
	\caption{The empirical \ac{cdf} of the \ac{pslr} of \ac{pacf} for the \ac{acfk} and generalized \ac{pas} systems, under the same \ac{esl} constraints $(\ac{esl}=-45.2071~{\rm{dB}})$.}
	\label{fig:pas_cdf}
\end{figure}

To intuitively illustrate the detection performance gain, we plot the range profiles of two targets and target detection results under \ac{acfk} and generalized \ac{pas} systems in Fig.~\ref{fig:profile_all}. For fair comparison, the detection probability of the weaker target is configured as $P_{\rm{D}}=0.7$ with identical $\ac{snr}=22.1~{\rm{dB}}$ and $\ac{sir}=7.2~{\rm{dB}}$ settings for both \ac{acfk} and generalized \ac{pas} systems. The detection threshold is calculated by \ac{cacfar}. As can be seen, the range profile of generalized \ac{pas} exhibits a large number of spurious peaks, thereby yielding substantial false alarms. Benefiting from its lower and stable \ac{pslr}, \ac{acfk} produces a cleaner range profile with no false alarms.

\begin{figure}[t]
    \vspace{-3mm}
	\centering
    \subfloat[The \ac{acfk} system]{
	\includegraphics[width=0.95\linewidth]{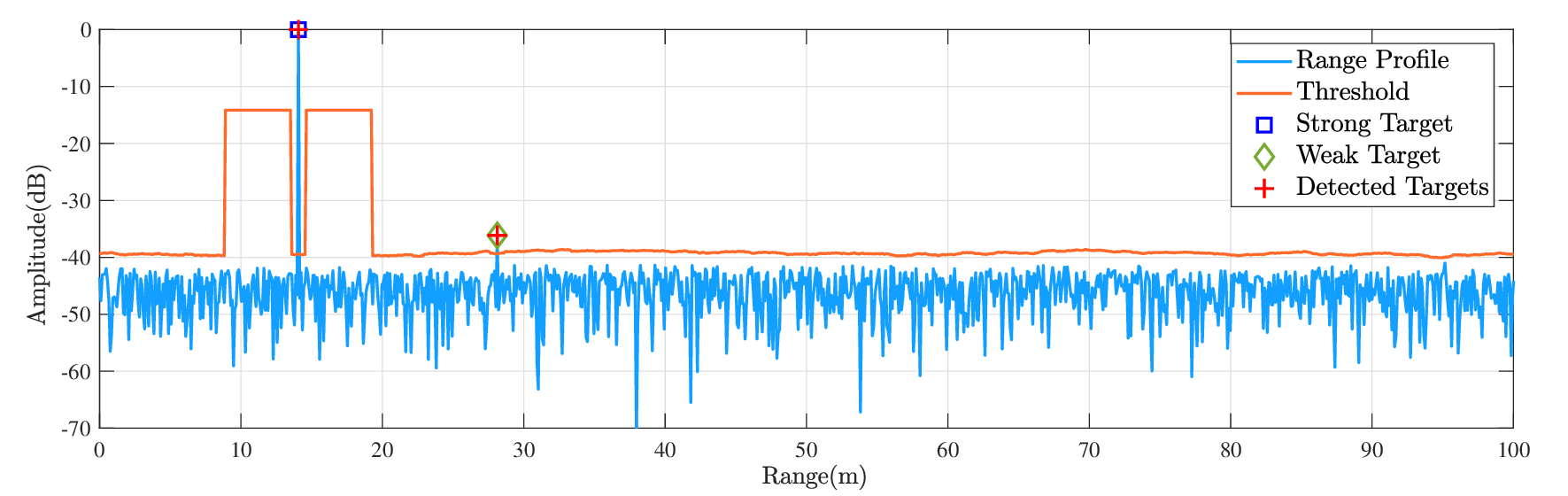}
    \label{fig:ACFK_profile}
    }\hspace{0.5mm}
    \subfloat[The generalized \ac{pas} system]{
	\includegraphics[width=0.95\linewidth]{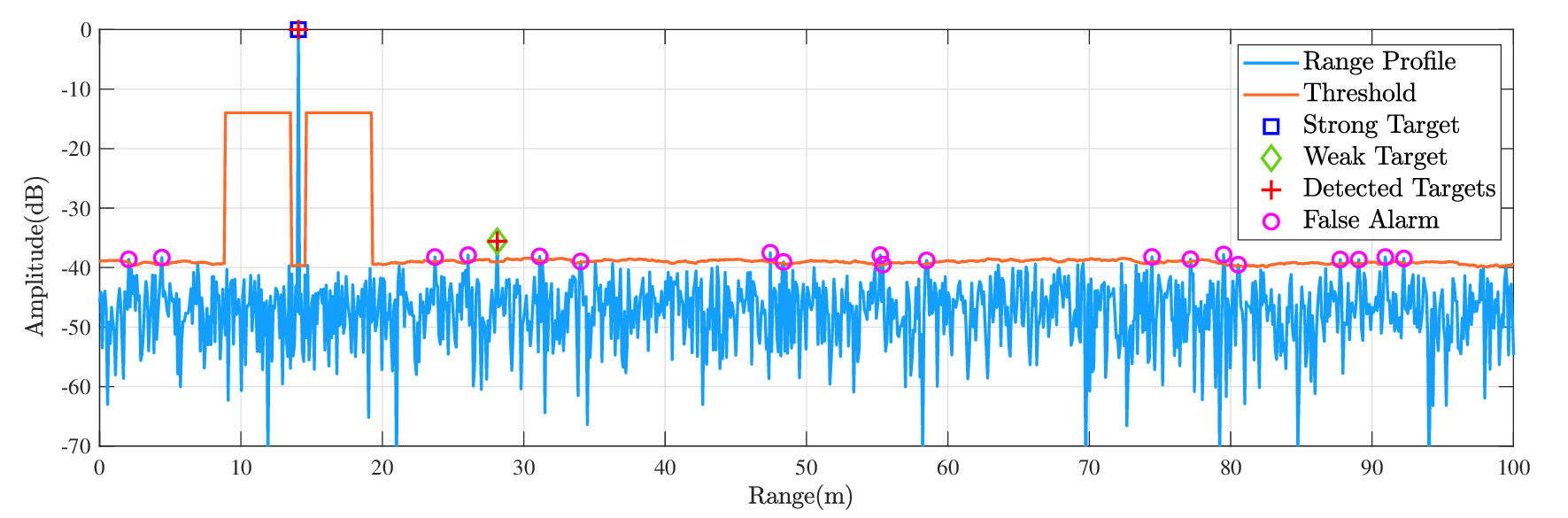}
    \label{fig:PAS_profile}
    }
	\caption{The target detection performance and range profiles of two targets under \ac{acfk} and generalized \ac{pas} systems, with the same detection probability of the weaker target $P_{\rm{D}}=0.7$.}
    \label{fig:profile_all}
\end{figure}

We further evaluate the target detection performance by the \ac{roc} curves, which characterize the tradeoff between the false-alarm probability $P_{\rm{FA}}$ and detection probability $P_{\rm{D}}$ of the weaker target. In Fig.~\ref{fig:roc_snr}, the \ac{snr} of the weaker target is fixed at $22.1~{\rm{dB}}$, while the \ac{sir} varies by adjusting the strength of the stronger target. It can be observed that the \ac{acfk} achieves significantly better target detection performance compared to the generalized \ac{pas} in the low false-alarm probability regime. Specifically, \ac{acfk} can attain a detection probability of  $P_{\rm{D}}=0.7$ at the false-alarm probability of $P_{\rm{FA}}=0.003$, even under $\ac{sir}=5.2~{\rm{dB}}$. As the \ac{sir} decreases, the detection performance of the generalized \ac{pas} degrades more rapidly than that of \ac{acfk}.


In Fig.~\ref{fig:roc_sir}, the \ac{sir} is fixed at $7.2~{\rm{dB}}$, while the \ac{snr} decreases by raising the power of the sensing noise. Clearly, the \ac{acfk} delivers better target detection performance than the generalized \ac{pas} in the low false-alarm probability regime. The performance gap between the two schemes narrows as the \ac{snr} grows.

 \begin{figure}[t]
     \centering
     \includegraphics[width=0.8\linewidth]{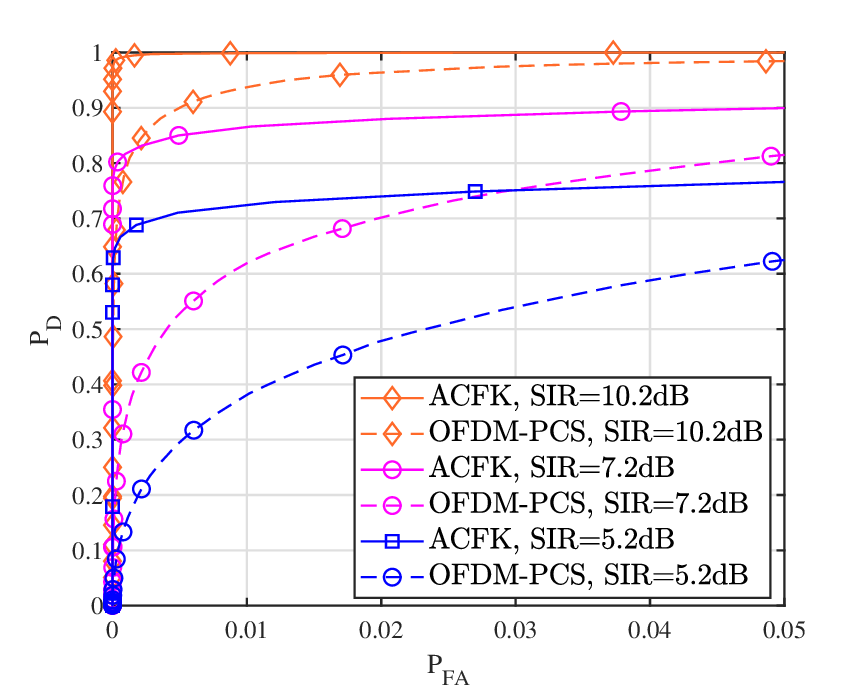}
     \caption{\ac{roc} curves of the weaker target under \ac{acfk} and generalized \ac{pas} systems with a fixed \ac{snr} $(\ac{snr}=22.1~{\rm{dB}})$  and various values of $\ac{sir}$.}
     \label{fig:roc_snr}
 \end{figure}

 \begin{figure}[t]
     \centering
     \includegraphics[width=0.8\linewidth]{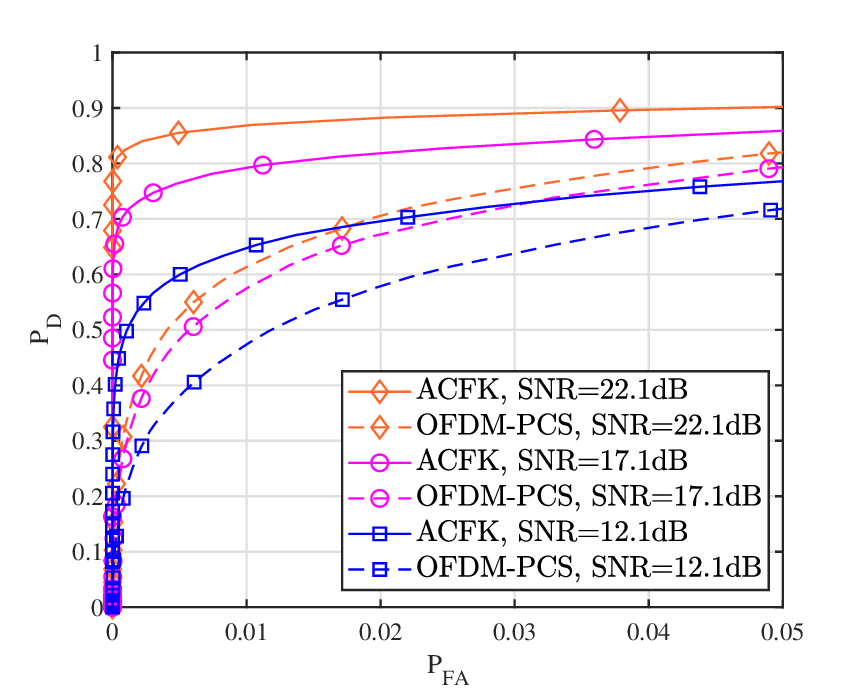}
     \caption{\ac{roc} curves of the weaker target under \ac{acfk} and generalized \ac{pas} systems with a fixed \ac{sir} $(\ac{sir}=7.2~{\rm{dB}})$  and various values of $\ac{snr}$.}
     \label{fig:roc_sir}
 \end{figure}

\section{Conclusions}\label{sec:conclusion}
In this paper, we have investigated the payload-bearing waveform design problem, which imposes stringent constraints on the sidelobes of random \ac{pacf} while maintaining a high achievable communication rate. Specifically, built upon the mutual information maximization framework under \ac{psl} and power budget constraints, we have shown that the continuous \ac{acf}-domain uniform construction serves as an asymptotically optimal high-\ac{snr} design principle under \ac{psl} constraints over quasi-static frequency-flat channels. In this context, we have proposed \ac{acfk}, which can be interpreted as a structured implementation of this principle in practical finite-constellation systems. From the sensing perspective, \ac{acfk} directly embeds data symbols into the \ac{acf}-domain sidelobes, thereby enabling exact control of the nominal \ac{pacf} and high-probability control of the actual \ac{pacf} under potential violation of a power spectral non-negativity constraint. We have provided a reference \ac{isac} transceiver design for \ac{acfk} over quasi-static multipath channels, and analyzed the high-\ac{snr} approximate \ac {ber}. From a practical perspective, \ac{acfk} can be implemented by modifying and amalgamating the existing \ac{ofdm} and \ac{sc} modulation schemes. Numerical results have demonstrated that \ac{acfk} outperforms a generalized \ac{pas} baseline in \ac{pslr} control and weak-target detection at similar data rate and \ac{ber}. Future work will explore further enhancements in \ac{acfk} system design, for example, developing practical coding strategies for data sequences satisfying the non-negativity constraint, as well as investigating the optimal communication receiver design and performance analysis.

\appendices

\section{Proof of Lemma \ref{lem:highsnr}}\label{sec:proof_highsnr}
\begin{IEEEproof}
We first consider a slightly modified feasible region
\begin{align}
\Set{P}_{p_{\min}} &= \bigl\{\RV{p}^\prime\in\mathbb{R}^{N-1}:~\rv{p}_i\geq p_{\min},~i=1\dotsc N, \nonumber \\
&\hspace{3mm}|\V{f}_k^{\rm H}\RV{p}|^2\leq \zeta_{\rm ACF}^{-1},~k\in\Set{S}_{\rm ACF} \bigr\},
\end{align}
where $p_{\rm min}>0$, and denote $\Set{M}_{p_{\min}}=\Set{P}_{p_{\min}}\times \mathbb{T}^N$. Note that $|\Set{M}_{p_{\min}}|=|\Set{P}_{p_{\min}}|(2\pi)^N$, implying that $0<|\Set{M}_{p_{\min}}|<\infty$ since $|\Set{P}_{p_{\min}}|>0$. The boundary of $\Set{P}_{p_{\min}}$ satisfies
\begin{equation}
\partial \Set{P}_{p_{\min}} \subset \bigcup_{i=1}^N \Set{B}_i^{(1)}\cup \bigcup_{k\in\Set{S}_{\rm ACF}}\Set{B}_k^{(2)},
\end{equation}
where $\Set{B}_i^{(1)}=\{\RV{p}^\prime :\RV{p}^\prime\in \Set{P}_{p_{\min}},\rv{p}_i=p_{\rm min}\}$, and $\Set{B}_k^{(2)}=\{\RV{p}^\prime :\RV{p}^\prime\in \Set{P}_{p_{\min}},|\V{f}_k^{\rm H}\RV{p}|^2=\zeta_{\rm ACF}^{-1}\}$. The sets $\Set{B}_i^{(1)}$ are hyperplane patches, whose $\rho$-neighborhoods have volume of $O(\rho)$. The sets $\Set{B}_k^{(2)}$ are smooth hypersurfaces, which also have $\rho$-neighborhoods with volume of $O(\rho)$. This implies that
\begin{equation}\label{asu1}
\bigl|\{\RV{u}:\min_{\V{x}\in\partial\Set{M}_{p_{\min}}}\|\RV{u} -\V{x}\|\leq \rho\}\bigr|\leq c_1(2\pi)^N\rho,
\end{equation}
where $c_1$ is a constant.

Now, denote $\rv{z}_i=[\V{h}_{\rm c}]_i\rv{x}_i$, note that the mapping $\phi(\RV{u})$ is infinitely differentiable, and the Jacobian matrix $\M{J}(\RV{u})$ satisfies
\begin{subequations}
\begin{align}
\frac{\partial}{\partial \rv{p}_i} \RV{z}&=\frac{[\V{h}_{\rm c}]_ie^{j\rv{\theta}_i}}{2\sqrt{\rv{p}_i}}\V{e}_i-\frac{[\V{h}_{\rm c}]_Ne^{j\rv{\theta}_N}}{2\sqrt{\rv{p}_N}}\V{e}_N,\\
\frac{\partial}{\partial \rv{\theta}_i}\RV{z}&=j[\V{h}_{\rm c}]_i\sqrt{\rv{p}_i}e^{j\rv{\theta}_i}\V{e}_i,
\end{align}
\end{subequations}
where $\V{e}_i$ denotes the $N$-dimensional vector with all entries being zero except that the $i$-th entry equals to one. After some manipulations, we obtain
\begin{equation}
\M{G}(\RV{u})=\begin{bmatrix}
\M{G}_{pp}& \M{0}\\
\M{0} & \M{G}_{\theta\theta}
\end{bmatrix},\label{G_structure}
\end{equation}
where
\begin{subequations}\label{G_structure2}
\begin{align}
\M{G}_{pp} &= \frac{1}{4}{\rm diag}\bigl(|\V{h}_{\rm c}^\prime|^2\oslash \RV{p}^\prime\bigr)+ \frac{|[\V{h}_{\rm c}]_N|^2}{4\rv{p}_N}\V{1}\V{1}^{\rm T},\\
\M{G}_{\theta\theta}&={\rm diag}\bigl(|\V{h}_{\rm c}|^2\odot \RV{p}\bigr),
\end{align}
\end{subequations}
and $\V{h}_{\rm c}=[(\V{h}_{\rm c}^\prime)^{\rm T},~[\V{h}_{\rm c}]_N]^{\rm T}$. It is obvious that $\M{G}(\RV{u})$ is positive definite for every $\RV{u}\in\Set{M}_{p_{\min}}$, having constant upper and lower bounds denoted as follows:
$$
c_G\M{I}\preceq \M{G}(\RV{u})\preceq C_G\M{I}.
$$
This implies that $\M{J}(\RV{u})$ has full column rank, and hence $\phi(\RV{u})$ is an immersion. It is also injective since $[\V{h}_{\rm c}]_i\neq 0$ holds for all $i$. In addition, $\phi^{-1}$ can be explicitly constructed as
$$
\rv{p}_i= \frac{|\rv{z}_i|^2}{|[\V{h}_{\rm c}]_i|^2},~e^{j\rv{\theta}_i} = \frac{\rv{z}_i[\V{h}_{\rm c}]_i^{-1}}{|\rv{z}_i[\V{h}_{\rm c}]_i^{-1}|},
$$
which is continuous. Thus we may conclude that $\phi(\RV{u})$ is an embedding.

Next let us construct a tubular coordinate for $\RV{u}$. Let
\begin{equation}
\M{D}:= {\rm diag}(|\V{h}_{\rm c}|^{-2} \otimes \V{1}_2),
\end{equation}
we see that the image of $\phi(\RV{u})$ resides in
\begin{equation}
\Set{T}_h:= \{\V{y}\in\mathbb{R}^{2N}:\V{y}^{\rm T}\M{D}\V{y}=N\}.
\end{equation}
Since $\M{D}$ is positive definite, $\Set{T}_h$ is a smooth hypersurface in $\mathbb{R}^{2N}$. Define $f(\V{y}):=\V{y}^{\rm T}\M{D}\V{y}-N$, we see that $\Set{T}_h$ constitutes the set of solutions to $f(\V{y})=0$, and hence for all tangent vector $\M{J}(\RV{u})\V{v}$ at $\RV{u}$, it holds that
\begin{equation}
[\nabla f(\phi(\RV{u}))]^{\rm T}\M{J}(\RV{u})\V{v}=0,
\end{equation}
which implies that $[\M{D}\phi(\RV{u})]^{\rm T}\M{J}(\RV{u})\V{v}=0$. Thus we may express the normal vector at $\RV{u}$ as
\begin{equation}
\V{b}(\RV{u})=\V{n}(\phi(\RV{u}))=\frac{\M{D}\phi(\RV{u})}{\|\M{D}\phi(\RV{u})\|}.
\end{equation}
Next, let us denote
\begin{equation}
\V{t}(\V{y},r) = \V{y}+r\V{n}(\V{y}),
\end{equation}
and $d_{\min}=\lambda_{\min}(\M{D})$, $d_{\max}=\lambda_{\max}(\M{D})$, thereby we have
\begin{equation}
\|\M{D}\V{y}\|\geq s_{\min} := \sqrt{d_{\min}N}.
\end{equation}
Choose an arbitrary $\rho_0$ from $(0,s_{\min}/(2d_{\max}))$, let $\mu=r/\|\M{D}\V{y}\|$, we have explicitly
$$
\V{t}(\V{y},r) = (\M{I}+\mu\M{D})\V{y}.
$$
Since $|\mu|=|r|/\|\M{D}\V{y}\|\leq \rho_0$, we have $|\mu|d_{\max}< 1/2$, which implies that $1+\mu d_l>0$ holds for every diagonal entry (and hence eigenvalue) $d_l$ of $\M{D}$. Assume by contradiction that $\V{t}(\V{y},r)$ is not one-to-one, such that there exists $(\tilde{\mu},\tilde{\V{y}})\neq (\mu,\V{y})$ satisfying
\begin{equation}
\V{t}(\V{y},r)=(\M{I}+\mu\M{D})\V{y}=(\M{I}+\tilde{\mu}\M{D})\tilde{\V{y}}.
\end{equation}
According to the constant-power constraint, we see that $\mu$ (or $\tilde{\mu}$) should satisfy
\begin{equation}
\psi(\mu)=\V{t}^{\rm T}(\V{y},r)(\M{I}+\mu\M{D})^{-1}\M{D}(\M{I}+\mu\M{D})^{-1}\V{t}(\V{y},r)-N=0.
\end{equation}
However, since $1+\mu d_l>0$ holds for all $d_l$, we have
\begin{equation}
\frac{\partial}{\partial \mu}\psi(\mu)= -2\sum_{l=1}^{2N}\frac{d_l^2 [\V{t}(\V{y},r)]_l^2}{(1+\mu d_l)^3}<0,
\end{equation}
and hence $\psi(\mu)=0$ has at most one solution. Therefore the mapping $\V{t}(\V{y},r)$ is one-to-one. It then turns out that the tubular coordinate transformation
$$
\tilde{\RV{t}}(\RV{u},r)=\RV{t}(\phi(\RV{u}),r)=\phi(\RV{u})+\RV{b}(\RV{u})r
$$
is a diffeomorphism over $\Set{M}\times (-\rho_0,\rho_0)$. Taking derivative we obtain
\begin{equation}
\Bigl(\frac{\partial \tilde{\RV{t}}}{\partial \RV{u}_r}\Bigr)^{\rm T}\frac{\partial \tilde{\RV{t}}}{\partial \RV{u}_r}=\begin{bmatrix}
(\M{J}+r\frac{\partial \RV{b}}{\partial \RV{u}})^{\rm T}(\M{J}+r\frac{\partial \RV{b}}{\partial \RV{u}}) & 0 \\
0 & 1
\end{bmatrix}.
\end{equation}
Therefore the Riemannian volume density reads
$$
v(\RV{u},r)=\sqrt{{\rm det}\Bigl[(\M{J}+r\frac{\partial \RV{b}}{\partial \RV{u}})^{\rm T}(\M{J}+r\frac{\partial \RV{b}}{\partial \RV{u}})\Bigr]},
$$
which degenerates to $\sqrt{{\rm det}\M{G}(\RV{u})}$ when $r=0$.

Next we give an upper bound for the general $v(\RV{u},r)$. The Jacobian from $\RV{y}$ to $\V{n}(\V{y})$ reads
$$
\frac{\partial \V{n}(\V{y})}{\partial \V{y}}=\frac{\M{D}}{\|\M{D}\V{y}\|} - \frac{\M{D}\V{y}\V{y}^{\rm T}\M{D}^2}{\|\M{D}\V{y}\|^3},
$$
which implies
$$
\left\|\frac{\partial \V{n}(\V{y})}{\partial \V{y}}\right\|\leq \frac{2\|\M{D}\|}{\|\M{D}\V{y}\|}\leq \frac{2d_{\max}}{s_{\min}}.
$$
Therefore we have
\begin{equation}
\left\|\frac{\partial \RV{b}(\RV{u})}{\partial \RV{u}}\right\| \leq \frac{2d_{\max}}{s_{\min}} \|\M{J}(\RV{u})\| \leq \frac{2d_{\max}\sqrt{C_G}}{s_{\min}}:=C_b<\infty.
\end{equation}
Denote $\M{G}_r(\RV{u})=(\M{J}+r\frac{\partial \RV{b}}{\partial \RV{u}})^{\rm T}(\M{J}+r\frac{\partial \RV{b}}{\partial \RV{u}})$, and
$$
\M{R}(\RV{u},r) = [\M{G}(\RV{u})]^{-\frac{1}{2}}(\M{G}_r(\RV{u})-\M{G}(\RV{u}))[\M{G}(\RV{u})]^{-\frac{1}{2}},
$$
we obtain
\begin{equation}
\|\M{R}(\RV{u},r)\|\leq \frac{2\sqrt{C_G}C_B}{c_G}|r|+\frac{C_B^2}{c_G}|r|^2.
\end{equation}
We may choose an appropriate $\rho_0$ such that $\|\M{R}(\RV{u},r)\|\leq \frac{1}{2}$ for all $|r|<\rho_0$. Now, since $\M{G}_r(\RV{u}) = [\M{G}(\RV{u})]^{\frac{1}{2}}(\M{I}+\M{R}(\RV{u},r))[\M{G}(\RV{u})]^{\frac{1}{2}}$, we obtain
\begin{equation}
\log \frac{v(\RV{u},r)}{\sqrt{\det \M{G}(\RV{u})}} = \frac{1}{2}\log \det(\M{I}+\M{R}(\RV{u},r)),
\end{equation}
which implies that
\begin{align}
\Bigl|\log \frac{v(\RV{u},r)}{\sqrt{\det \M{G}(\RV{u})}}\Bigr|&\leq (2N-1)\|\M{R}(\RV{u},r)\| \nonumber \\
&= C_v|r|, \label{log_vol_upper}
\end{align}
where $C_v=\frac{2N-1}{c_G}(2\sqrt{C_G}C_B+C_B^2\rho_0)$.

We may rewrite the observation in the $\phi(\RV{u})$-coordinate as
\begin{equation}
\RV{y}_{\sigma_{\rm c}} = \phi(\RV{u}) + \sigma_{\rm c} \RV{\omega},
\end{equation}
where $\RV{\omega}$ follows the Gaussian distribution of
$$
p_{k,\sigma_{\rm c}}(\V{z})=\frac{1}{(\pi\sigma_{\rm c}^2)^{\frac{k}{2}}}\exp\left(-\frac{\|\V{z}\|^2}{\sigma_{\rm c}^2}\right),
$$
with $k=2N$. The density of $\RV{y}_{\sigma_{\rm c}}$ reads
$$
p_{\RV{y}_{\sigma_{\rm c}}}(\V{y})=\int_{\Set{M}} g(\V{u}) p_{k,\sigma_{\rm c}}(\V{y}-\phi(\V{u})){\rm d}\V{u}.
$$
Under the tubular coordinate, we have
\begin{equation}
q(\V{u},r)= p_{\RV{y}_{\sigma_{\rm c}}}(\tilde{\V{t}}(\V{u},r))v(\V{u},r).
\end{equation}
Choose a $b_{\sigma_{\rm c}}=\sigma_{\rm c}\sqrt{L\log(1/\sigma_{\rm c})}$ for a fixed, sufficiently large $L>0$, we define the $b_{\sigma_{\rm c}}$-interior of $\Set{M}_{p_{\min}}$
\begin{equation}
\Set{M}_{p_{\min},\sigma_{\rm c}} \!=\! \{\V{u}\!\in\!\Set{M}_{p_{\min}}:\min_{\V{x}\in\partial\Set{M}_{p_{\min}}}\|\V{u}\!-\!\V{x}\|\! >\! b_{\sigma_{\rm c}}\},
\end{equation}
and the principal domain of integration
\begin{equation}
\Set{E}_{\sigma_{\rm c}} = \{(\V{u},r):\V{u}\in\Set{M}_{p_{\min},\sigma_{\rm c}},|r|\leq b_{\sigma_{\rm c}}\}.
\end{equation}
Consider a small perturbation $\V{w}$ satisfying $\|\V{w}\|\leq \frac{b_{\sigma_{\rm c}}}{\sigma_c}=\sqrt{L\log (1/\sigma_{\rm c})}$, let $\V{v}=\V{u}+\sigma_{\rm c}\V{w}$, taking the Taylor expansion of $\phi(\V{u}+\sigma_{\rm c}\V{w})$ at $\V{u}$, we obtain
\begin{equation}
\phi(\V{u}+\sigma_{\rm c}\V{w}) = \phi(\V{u}) + \sigma_{\rm c}\M{J}(\V{u})\V{w} + \frac{\sigma_{\rm c}^2}{2}\V{w}^{\rm T}\M{H}_\phi\V{w} + \V{r}_3,
\end{equation}
where $\M{H}$ denotes the Hessian of $\phi(\V{u})$ satisfying $\V{w}^{\rm T}\M{H}_\phi\V{w}\leq c_2\|\V{w}\|^2$, and $\V{r}_3$ satisfies $\|\V{r}_3\|\leq c_3 \sigma_{\rm c}^3\|\V{w}\|^3$, for some positive constants $c_2$ and $c_3$, since $\phi(\V{u})$ is a $C^{\infty}$ embedding. Thus the residual term reads
\begin{equation}
\tilde{\V{t}}(\V{u},r) - \phi(\V{u}+\sigma_{\rm c}\V{w}) = \V{b}(\V{u})r -\sigma_{\rm c}\M{J}(\V{u})\V{w} - \frac{\sigma_{\rm c}^2}{2}\V{w}^{\rm T}\M{H}_\phi\V{w} - \V{r}_3.
\end{equation}
The leading term $\V{b}(\V{u})r -\sigma_{\rm c}\M{J}(\V{u})\V{w}$ satisfies
\begin{align}
\|\V{b}(\V{u})r -\sigma_{\rm c}\M{J}(\V{u})\V{w}\|^2 &=r^2\|\V{b}(\V{u})\|^2 +\sigma_{\rm c}^2\|\M{J}(\V{u})\V{w}\|^2 \nonumber \\
&\hspace{3mm}-2\sigma_{\rm c}r\V{b}^{\rm T}(\V{u})\M{J}(\V{u})\V{w} \nonumber \\
&=r^2 +\sigma_{\rm c}^2 \V{w}^{\rm T}\M{G}(\V{u})\V{w}.
\end{align}
We thereby have
\begin{equation}
\frac{\|\tilde{\V{t}}(\V{u},r) - \phi(\V{u}+\sigma_{\rm c}\V{w})\|^2}{\sigma_{\rm c}^2}= \frac{r^2}{\sigma_{\rm c}^2} + \V{w}^{\rm T}\M{G}(\V{u})\V{w}+o(1),
\end{equation}
since $\|\V{w}\|\leq \sqrt{L\log(1/\sigma_{\rm c})}$. It also follows that
$$
g(\V{u}+\sigma_{\rm c}\V{w}) = g(\V{u}) (1+o(1)),
$$
which implies
$$
p_{\RV{y}_{\sigma_{\rm c}}}(\tilde{\V{t}}(\V{u},r)) = \frac{(\pi\sigma_{\rm c}^2)^\frac{2N-1}{2}}{(\pi\sigma_{\rm c}^2)^N\sqrt{\det \M{G}(\V{u})}}e^{-\frac{r^2}{\sigma_{\rm c}^2}}g(\V{u})(1+o(1)),
$$
and hence
$$
q(\V{u},r)=\frac{v(\V{u},r)g(\V{u})}{\sqrt{\det \M{G}(\V{u})}} p_{1,\sigma_{\rm c}}(r)(1+o(1)).
$$
Now using \eqref{log_vol_upper} we have $\frac{v(\V{u},r)}{\sqrt{\det \M{G}(\V{u})}}=1+o(1)$, thereby
\begin{equation}\label{entropy_expansion}
q(\V{u},r)=g(\V{u})p_{1,\sigma_{\rm c}}(r)(1+o(1))
\end{equation}
holds uniformly on $\Set{E}_{\sigma_{\rm c}}$.

Before we estimate the entropy in the main region $\Set{E}_{\sigma_{\rm c}}$, let us first control the entropy outside $\Set{E}_{\sigma_{\rm c}}$. In total there are three exceptions:
\begin{enumerate}
\item Near boundary of $\Set{M}_{p_{\min}}$: Namely $\underline{\Set{M}}_{p_{\min}}=\{\V{u}:\min_{\V{x}\in\partial\Set{M}_{p_{\min}}} \|\V{u}-\V{x}\|\leq b_{\sigma_{\rm c}}\}$. According to \eqref{asu1}, we have $|\underline{\Set{M}}_{p_{\min}}|\leq c_1(2\pi)^N b_{\sigma_{\rm c}}$, implying that
$$
\mathbb{P}\{\min_{\V{x}\in\partial\Set{M}_{p_{\min}}} \|\RV{u}-\V{x}\|\leq b_{\sigma_{\rm c}}\}=O(b_{\sigma_{\rm c}}),
$$
and thus the contribution of $\underline{\Set{M}}_{p_{\min}}$ to the entropy is at most $O(b_{\sigma_{\rm c}}|\log\sigma_{\rm c}|)=o(1)$.
\item Tail on the normal direction of the tube: Namely when $|r|>b_{\sigma_{\rm c}}$. In this case the contribution to the entropy can be bounded by
\begin{equation}
I \leq c_5\int_{|s|>\sqrt{L\log(1/\sigma_{\rm c})}}e^{-\frac{c_4 s^2}{2}}(1+|\log\sigma_{\rm c}|+s^2){\rm d}s,
\end{equation}
for some positive constant $c_4$, which is $o(1)$ for sufficiently large $L$.
\item Outside the tube: Namely when $\sigma_{\rm c}\|\RV{\omega}\|\geq \rho_0$. For sufficiently small $\sigma_{\rm c}$, we may control the Gaussian tail probability as follows
$$
\mathbb{P}\{\|\RV{\omega}\|>\rho_0/\sigma_{\rm c}\}\leq c_5e^{-\frac{\rho_0^2}{2\sigma_{\rm c}^2}},
$$
for some positive constant $c_4$. The polynomial ordered terms in the integrand would then be overwhelmed by the exponential decay. The corresponding entropy contribution is thus on the order of $o(1)$.
\end{enumerate}

Now, using \eqref{entropy_expansion}, we obtain
\begin{align}
h(\RV{y}_{\sigma_{\rm c}})&=-\int_{\Set{E}_{\sigma_{\rm c}}} q(\V{u},r)\log q(\V{u},r){\rm d}\V{u}{\rm d}r \nonumber \\
&\hspace{3mm}+\int_{\Set{E}_{\sigma_{\rm c}}} q(\V{u},r)\log v(\V{u},r){\rm d}\V{u}{\rm d}r + o(1). \label{entropy_exp2}
\end{align}
For the first term on the right hand side of \eqref{entropy_exp2}, we have
\begin{align}
-\int_{\Set{E}_{\sigma_{\rm c}}} q(\V{u},r)\log q(\V{u},r){\rm d}\V{u}{\rm d}r=h(\RV{u})+\frac{1}{2}\log(\pi e\sigma_{\rm c}^2)+o(1).
\end{align}
For the second term, from \eqref{log_vol_upper} we have
$$
\log v(\V{u},r) =\frac{1}{2}\log\det\M{G}(\V{u}) + O(|r|),
$$
and hence from \eqref{entropy_expansion} we have
\begin{align}
\int_{\Set{E}_{\sigma_{\rm c}}} q(\V{u},r)\log v(\V{u},r){\rm d}\V{u}{\rm d}r=\frac{1}{2}\mathbb{E}\{\log\det\M{G}(\RV{u})\}+o(1).
\end{align}
We thereby obtain
\begin{equation}
h(\RV{y}_{\sigma_{\rm c}}) = h(\RV{u})+\frac{1}{2}\mathbb{E}\{\log\det\M{G}(\RV{u})\}+\frac{1}{2}\log(\pi e\sigma_{\rm c}^2)+o(1).
\end{equation}

We are finally ready to expand the mutual information, as follows
\begin{align}
I(\RV{u};\V{y}_{\sigma_{\rm c}}) &= h(\RV{y}_{\sigma_{\rm c}})-h(\RV{y}_{\sigma_{\rm c}}|\RV{u})\nonumber \\
&=h(\RV{u})+\frac{1}{2}\mathbb{E}\{\log\det\M{G}(\RV{u})\} \nonumber \\
&\hspace{3mm}-\frac{2N-1}{2}\log(\pi e\sigma_{\rm c}^2)+o(1).
\end{align}
Since $\RV{x}=\Psi(\RV{u})$ is a one-to-one mapping, we have $I(\RV{u};\V{y}_{\sigma_{\rm c}})=I(\RV{x};\V{y}_{\rm c})$ yielding \eqref{lem1_highsnr}. Finally, we take the limit of $p_{\min}\rightarrow 0_+$, recovering the original feasible region, which completes the proof.
\end{IEEEproof}

\section{Proof of Proposition \ref{prop:opt_input}}\label{sec:proof_opt_input}
\begin{IEEEproof}
According to Lemma~\ref{lem:highsnr}, assuming that the optimal input distribution satisfies the conditions \eqref{conditions_input}, in the high-\ac{snr} regime the optimization-relevant term is
\begin{equation}
L[p_{\RV{u}}] = h(\RV{u}) + \frac{1}{2}\int_{\Set{M}} p_{\RV{u}}(\V{u}) \log\det\M{G}(\V{u}){\rm d}\V{u}.
\end{equation}
Expanding $h(\RV{u})$ we obtain
\begin{equation}
L[p_{\RV{u}}] = -\int_{\Set{M}} p_{\RV{u}}(\V{u}) \log \frac{p_{\RV{u}}(\V{u})}{\sqrt{\det \M{G}(\V{u})}}{\rm d}\V{u}.
\end{equation}
Note that $Z_G^{-1}\sqrt{\det \M{G}(\V{u})}$ is a valid probabilistic density function, with
$$
Z_G = \int_{\Set{M}}\sqrt{\det \M{G}(\V{u})}{\rm d}\V{u}.
$$
It then follows that
$$
L[p_{\RV{u}}] = \log Z_G - {\rm KL}[p_{\RV{u}}(\V{u})\|Z_G^{-1}\sqrt{\det\M{G}(\V{u}})],
$$
which implies that the optimal distribution is given by
\begin{equation}
p_{\RV{u}}(\V{u})= \frac{\sqrt{\det \M{G}(\V{u})}}{\int_{\Set{M}}\sqrt{\det \M{G}(\V{v})}{\rm d}\V{v}},~\V{u}\in\Set{M}.
\end{equation}
According to \eqref{G_structure} and \eqref{G_structure2}, we have
\begin{equation}
\det\M{G}(\V{u}) = \det \M{G}_{pp} \det \M{G}_{\theta\theta},
\end{equation}
and
\begin{subequations}
\begin{align}
\det \M{G}_{pp} &= \prod_{i=1}^{N-1} \frac{|[\V{h}_{\rm c}]_i|^2}{4p_i}\Bigl(1+\frac{|[\V{h}_{\rm c}]_N|^2}{p_N}\sum_{j=1}^{N-1}\frac{p_j}{|[\V{h}_{\rm c}]_j|^2}\Bigr),\\
\det \M{G}_{\theta\theta}&= \prod_{i=1}^N |[\V{h}_{\rm c}]_i|^2 p_i,
\end{align}
\end{subequations}
implying that
\begin{equation}
\sqrt{\det\M{G}(\V{u})}=\frac{1}{2^{N-1}} \Bigl(\prod_{j=1}^N|[\V{h}_{\rm c}]_j|^2\Bigr)\sqrt{\sum_{i=1}^N\frac{p_i}{|[\V{h}_{\rm c}]_i|^2}}.
\end{equation}
Note that $\det\M{G}(\V{u})$ is independent of $\RV{\theta}$, and hence the optimal distribution of $\RV{\theta}$ is uniform. Finally, we observe that the optimal input distribution indeed satisfies \eqref{conditions_input} and hence Lemma~\ref{lem:highsnr} is applicable. This completes the proof.
\end{IEEEproof}

\section{Proof of Proposition \ref{prop:violation_prob}}\label{sec:proof_violation_prob}
\begin{IEEEproof}
We first note that for all $k$, we have
$$
[\V{f}_k]_1[\RV{x}_{\rm{ACF}}]_1 = \frac{1}{\sqrt{N}}.
$$
Thus it now suffices to show that
\begin{equation}
\mathbb{P}\left\{\widetilde{\V{f}}_k^{\rm T}\widetilde{\RV{x}}_{\rm{ACF}} < -\frac{1}{\sqrt{N}}\right\} \leq \exp\left(-\alpha\zeta_{\rm ACF}\right),
\end{equation}
where $\widetilde{\V{f}}_k = [[\V{f}_k]_2,\dotsc, [\V{f}_k]_N]^{\rm T}$ and $\widetilde{\RV{x}}_{\rm{ACF}}=[[\RV{x}_{\rm{ACF}}]_2,\dotsc,[\RV{x}_{\rm{ACF}}]_N]^{\rm T}$. Next, we wish to bound
$$
\widetilde{\V{f}}_k^{\rm T}\widetilde{\RV{x}}_{\rm{ACF}} = \frac{1}{\sqrt{N}}\sum_{n=2}^N [\widetilde{\RV{x}}_{\rm{ACF}}]_n e^{-\frac{j2\pi(k-1)(n-1)}{N}}.
$$
Using the conjugate symmetry of $\RV{x}_{\rm{ACF}}$, we obtain
\begin{equation}
\widetilde{\V{f}}_k^{\rm T}\widetilde{\RV{x}}_{\rm{ACF}} =\frac{2}{\sqrt{N}}\sum_{n=2}^{\overline{N}} \rv{z}_n,
\end{equation}
where
$$
\rv{z}_n=\rv{r}_n\cos\left(\frac{2\pi (k-1) (n-1)}{N}\right)+\rv{i}_n\sin\left(\frac{2\pi (k-1) (n-1)}{N}\right),
$$
$\rv{r}_n = {\rm Re}\{[\RV{x}_{\rm{ACF}}]_n\}$, $\rv{i}_n = {\rm Im}\{[\RV{x}_{\rm{ACF}}]_n\}$, and $\overline{N}$ is given by
$$
\overline{N} = \left\{\begin{array}{ll}
N/2,&\hbox{$N$ even;}\\
(N-1)/2+1, &\hbox{$N$ odd.}
\end{array}\right.
$$
To facilitate further analysis, let us denote $a_n=\cos\left(\frac{2\pi (k-1) (n-1)}{N}\right)$ and $b_n=\sin\left(\frac{2\pi (k-1) (n-1)}{N}\right)$, and hence we have $\rv{z}_n=a_n\rv{r}_n+b_n\rv{i}_n$.

Now, note that each $\rv{z}_n$ is independent of one another. If we can show that each $\rv{z}_n$ is $K$-sub-Gaussian, we may then use the Chernoff inequality to bound the outage probability as follows

\begin{align}\label{outage_prob}
&\qquad\mathbb{P}\{\V{f}_k^{\rm T}\RV{x}_{\rm{ACF}}<0\} =\mathbb{P}\left\{\widetilde{\V{f}}_k^{\rm T}\widetilde{\RV{x}}_{\rm{ACF}} \!<\! -\frac{1}{\sqrt{N}}\right\} \nonumber\\
&=\mathbb{P}\!\left\{\sum_{n=2}^{\overline{N}} \rv{z}_n<-\frac{1}{2}\right\}\! \leq \! \exp\left(-\frac{1}{8(\overline{N}-1)K}\right)\nonumber \\
\!&\leq\! \exp\left(-\frac{1}{4NK}\right).
\end{align}

To elaborate, $K$-sub-Gaussianity means that
$$
\mathbb{E}\{e^{\lambda \rv{z}_n}\}\leq \exp\left(\frac{\lambda^2 K}{2}\right),~\forall \lambda\in \mathbb{R}.
$$
In general, by applying Hoeffding's inequality, we have a valid choice of $K=\frac{1}{N\zeta_{\rm ACF}}$, by the design of ACFK, which yields the generic bound corresponding to $\alpha=\frac{1}{4}$, the fourth line of \eqref{alpha_branches}. For specific constellations, smaller values of $K$ can be achievable. In particular, if $K$ satisfies
$$
K = \mathbb{E}\{[\rv{z}_n-\mathbb{E}(\rv{z}_n)]^2\},
$$
$\rv{z}_n$ is said to be \textit{strictly sub-Gaussian}. Next, we show that various constellations yield strictly sub-Gaussian $\rv{z}_n$'s.

\subsubsection{Uniform distribution}
In this case we see that $\rv{z}_n$'s follow a rescaled Wigner's semicircle law \cite{wigner1958distribution} as follows
\begin{equation}
p_{\rv{z}_n}(z_n)=\frac{2\sqrt{N\zeta_{\rm ACF}}}{\pi}\sqrt{1-N\zeta_{\rm ACF}z_n^2},~|z_n|\leq \frac{1}{\sqrt{N\zeta_{\rm ACF}}}.
\end{equation}
The corresponding variance reads
\begin{equation}\label{z2_uni}
\mathbb{E}\{\rv{z}_n^2\} = \frac{1}{4N\zeta_{\rm ACF}}.
\end{equation}
It has been shown that random variables following Wigner's semicircle law are indeed strictly sub-Gaussian \cite{10.1214}, and hence we obtain the first line of \eqref{alpha_branches} by substituting \eqref{z2_uni} into \eqref{outage_prob}.

\subsubsection{$4(M_q^2-L^2)$-\ac{qam}}
These constellations are quadrant symmetric. We may thus focus on their first quadrant, in which the constellation points form the set
\begin{subequations}
\begin{align}
\Set{S}_{M_q,L} &= (\Set{A}\times \Set{A}) ~\backslash ~\{(\zeta a,\zeta b)|(a,b)\in\Set{A}\times \Set{A}, \nonumber \\
&\hspace{8mm}a>2(M_q\!-\!L)\!-\!1,~b>2(M_q\!-\!L)\!-\!1\}\\
\Set{A}&=\{\zeta,3\zeta,\dotsc,\zeta(2M_q-1)\},
\end{align}
\end{subequations}
where $\zeta=\sqrt{\frac{1}{2[L^2+(2M_q-L-1)^2]N\zeta_{\rm ACF}}}$ is a normalization constant ensuring that the magnitudes of all constellation points do not exceed $\frac{1}{\sqrt{N\zeta_{\rm ACF}}}$, by the design of \ac{acfk}. Next, let us denote a random variable that is uniformly distributed over $\Set{S}_{M,L}$ as $\RV{\psi}=[\rv{u},\rv{v}]^{\rm T}$. Using this notation, we see that $\rv{r}_n$ and $\rv{i}_n$ may be expressed as $\rv{r}_n=\rv{u}\rv{\sigma}_u$ and $\rv{i}_n=\rv{v}\rv{\sigma}_v$, where $\rv{\sigma}_u$ and $\rv{\sigma}_v$ are two-valued random variables satisfying $\mathbb{P}\{\rv{\sigma}_u=\pm 1\}=\mathbb{P}\{\rv{\sigma}_v=\pm 1\}=1/2$. Thereby we have
\begin{align}
\mathbb{E}\{e^{\lambda\rv{z}_n}\}&=\mathbb{E}\{e^{\lambda(a_n\rv{r}_n+b_n\rv{i}_n)}\}=\mathbb{E}\{e^{\lambda(a_n\rv{u}\rv{\sigma}_u+b_n\rv{v}\rv{\sigma}_v)}\} \nonumber \\
&=\mathbb{E}\Big\{\mathbb{E}\{e^{\lambda(a_n\rv{u}\rv{\sigma}_u+b_n\rv{v}\rv{\sigma}_v)}|\rv{u},\rv{v}\}\Big\} \nonumber \\
&=\mathbb{E}\{\mathbb{E}\{e^{\lambda a_n\rv{u}\rv{\sigma}_u}|\rv{u}\}\mathbb{E}\{e^{\lambda b_n\rv{v}\rv{\sigma}_v}|\rv{v}\}\}\nonumber \\
&=\mathbb{E}\{\cosh(a_n\lambda\rv{u})\cosh(b_n\lambda\rv{v})\}.
\end{align}
To proceed, we need the following lemma.
\begin{lemma}\label{lem:negative_correlation}
Let $f$ and $g$ be non-decreasing functions over $\Set{A}$, and $\RV{\psi}=[\rv{u},\rv{v}]^{\rm T}$ be a uniformly distributed random vector over $\Set{S}_{M_q,L}$. We have
\begin{equation}
\mathbb{E}\{f(\rv{u})g(\rv{v})\}\leq \mathbb{E}\{f(\rv{u})\}\mathbb{E}\{g(\rv{v})\}.
\end{equation}
\begin{IEEEproof}
We observe that
$$
\mathbb{E}\{f(\rv{u})g(\rv{v})\}=\mathbb{E}\{f(\rv{u})h(\rv{u})\},
$$
where $h(\rv{u}):=\mathbb{E}\{g(\rv{v})|\rv{u}\}$. Next we define an event
$$
\RS{F}=\{\rv{u}\leq \zeta(2M_q-2L-1)\},
$$
implying that
$$
\begin{aligned}
{\rm cov}\{f(\rv{u})h(\rv{u})\} &= \mathbb{E}\{f(\rv{u})h(\rv{u})\}-\mathbb{E}\{f(\rv{u})\}\mathbb{E}\{h(\rv{u})\} \\
&=\mathbb{P}\{\RS{F}\}\mathbb{E}\{f(\rv{u})h(\rv{u})|\RS{F}\}+\mathbb{P}\{\overline{\RS{F}}\}\mathbb{E}\{f(\rv{u})h(\rv{u})|\overline{\RS{F}}\}\\
&\hspace{3mm}-[\mathbb{P}\{\RS{F}\}\mathbb{E}\{f(\rv{u})|\RS{F}\}+\mathbb{P}\{\overline{\RS{F}}\}\mathbb{E}\{f(\rv{u})|\overline{\RS{F}}\}]\\
&\hspace{5mm}\times[\mathbb{P}\{\RS{F}\}\mathbb{E}\{h(\rv{u})|\RS{F}\}+\mathbb{P}\{\overline{\RS{F}}\}\mathbb{E}\{h(\rv{u})|\overline{\RS{F}}\}].
\end{aligned}
$$
Now note that $h(\rv{u})|\RS{F}=\mathbb{E}\{g(\rv{v})|\RS{F}\} = h_{\rm full}$ and $h(\rv{u})|\overline{\RS{F}}=\mathbb{E}\{g(\rv{v})|\overline{\RS{F}}\} = h_{\rm corner}$ are constants. Thus we have
$$
\begin{aligned}
{\rm cov}\{f(\rv{u})h(\rv{u})\} &=\mathbb{P}\{\RS{F}\}\mathbb{P}\{\overline{\RS{F}}\}(\mathbb{E}\{f(\rv{u})|\RS{F}\}-\mathbb{E}\{f(\rv{u})|\overline{\RS{F}}\})\\
&\hspace{3mm}\times(h_{\rm full}-h_{\rm corner}).
\end{aligned}
$$
Since $g$ is non-decreasing, we have $h_{\rm full}-h_{\rm corner}\geq 0$. Since $f$ is non-decreasing, we have $\mathbb{E}\{f(\rv{u})|\RS{F}\}-\mathbb{E}\{f(\rv{u})|\overline{\RS{F}}\}\leq 0$. Thus we have ${\rm cov}\{f(\rv{u})h(\rv{u})\}\leq 0$, implying that
$$
\mathbb{E}\{f(\rv{u})h(\rv{u})\}\leq \mathbb{E}\{f(\rv{u})\}\mathbb{E}\{h(\rv{u})\} = \mathbb{E}\{f(\rv{u})\}\mathbb{E}\{g(\rv{v})\},
$$
completing the proof.
\end{IEEEproof}
\end{lemma}

With Lemma \ref{lem:negative_correlation}, we see that
\begin{subequations}
\begin{align}
\mathbb{E}\{e^{\lambda\rv{z}_n}\}&=\mathbb{E}\{\cosh(a_n\lambda\rv{u})\cosh(b_n\lambda\rv{v})\} \\
&=\mathbb{E}\{\cosh(|a_n|\lambda\rv{u})\cosh(|b_n|\lambda\rv{v})\} \label{even_cosh} \\
&\leq \mathbb{E}\{\cosh(|a_n|\lambda\rv{u})\}\mathbb{E}\{\cosh(|b_n|\lambda\rv{v})\} \label{increasing_factor}\\
&=\mathbb{E}\{e^{|a_n|\lambda\rv{r}_n}\}\mathbb{E}\{e^{|b_n|\lambda\rv{i}_n}\},
\end{align}
\end{subequations}
where \eqref{even_cosh} follows from the fact that $\cosh(x)$ is an even function of $x$, while \eqref{increasing_factor} follows from Lemma \ref{lem:negative_correlation} and the fact that $\cosh(x)$ is increasing on $x\geq 0$. Since $\rv{r}_n$ and $\rv{i}_n$ are identically distributed, it now suffices to show that $\rv{r}_n$ alone is strictly sub-Gaussian, namely,
\begin{equation}
\mathbb{E}\{e^{\lambda\rv{r}_n}\}\leq \exp\left(\frac{\lambda^2 \mathbb{E}\{\rv{r}_n^2\}}{2}\right),~\forall \lambda\in\mathbb{R},
\end{equation}
as long as $L\leq M_q/2$. Now, observe that
\begin{equation}
\mathbb{P}\{\rv{r}_n=\pm \zeta(2i-1)\} = \frac{w_i}{2(M_q^2-L^2)},
\end{equation}
where
$$
w_i = \left\{
\begin{array}{ll}
    M_q, & \hbox{$i\leq M_q-L$;} \\
    M_q-L, & \hbox{$M_q-L+1\leq i\leq M_q$.}
\end{array}
\right.
$$
Thus the characteristic function of $\rv{r}_n$ is given by
\begin{subequations}
\begin{align}
\mathbb{E}\{e^{j\lambda\rv{r}_n}\}&\!=\!\sum_{i=1}^{M_q} \frac{w_i}{M_q^2-L^2}\cos [\lambda\zeta(2i-1)]\\
&\!=\!\frac{M_q}{M_q^2-L^2}\sum_{i=1}^{M_q-L} \cos [\lambda\zeta(2i-1)] \nonumber\\
&\hspace{3mm} +\! \frac{M_q-L}{M_q^2-L^2}\sum_{i=M_q-L+1}^{M_q} \cos [\lambda\zeta(2i-1)] \\
&=\!\frac{\sin(2M_q \zeta \lambda)\!+\!\frac{L}{M_q-L}\sin[2(M_q\!-\!L)\zeta \lambda]}{2(M_q+L)\sin(\zeta \lambda)}, \label{geometric_series}
\end{align}
\end{subequations}
where \eqref{geometric_series} follows from $\sum_{i=1}^m \cos((2i-1)x)=\frac{\sin (2mx)}{2\sin x}$. Let us now denote the numerator of \eqref{geometric_series} as $G(\lambda)$, for which we have the following result.
\begin{lemma}
For $L\leq M_q/2$, all zeros of $G(\lambda)$ are real.
\begin{IEEEproof}
Assume by contradiction that $G(\lambda)=0$ for some complex $\lambda=a+jb$, which implies
$$
\begin{aligned}
&\sin(2M_q\zeta a)\cosh(2M_q\zeta b)\\
&\hspace{3mm}=-\frac{L}{M_q-L}\sin(2(M_q-L)\zeta a)\cosh(2(M_q-L)\zeta b),\\
&\cos(2M_q\zeta a)\sinh(2M_q\zeta b)\\
&\hspace{3mm}=-\frac{L}{M_q-L}\cos(2(M_q-L)\zeta a)\sinh(2(M_q-L)\zeta b),\\
\end{aligned}
$$
yielding
\begin{align}\label{normalize_contradiction}
&\frac{L^2}{(M_q\!-\!L)^2}\left[\sin^2(2(M_q\!-\!L)\zeta a)\left(\frac{\cosh(2(M_q\!-\!L)\zeta b)}{\cosh(2M_q\zeta b)}\right)^2  \right. \nonumber  \\
&\hspace{3mm}\left.+\cos^2(2(M_q\!-\!L)\zeta a)\left(\frac{\sinh(2(M_q\!-\!L)\zeta b)}{\sinh(2M_q\zeta b)}\right)^2\right] \!=\! 1,
\end{align}
after some manipulations. However, due to the non-decreasing property of $\cosh$ and $\sinh$, we see that $\frac{\cosh(2(M_q-L)\zeta b)}{\cosh(2M_q\zeta b)}\in (0,1)$ and $\frac{\sinh(2(M_q-L)\zeta b)}{\sinh(2M_q\zeta b)}\in (0,1)$, implying that the left hand side of \eqref{normalize_contradiction} is strictly less than $1$, as long as $L\leq M_q/2$ (so that $\frac{L^2}{(M_q-L)^2}\leq 1$). Thus the proof is completed.
\end{IEEEproof}
\end{lemma}

Now, according to \cite{real_zeros}, we may conclude that $\rv{r}_n$ is strictly sub-Gaussian, since the zeros of its characteristic function are all real. After some algebra, we obtain
\begin{equation}
K\!=\!\mathbb{E}\{\rv{r}_n^2\}\!=\! \frac{M_q(4M_q^2\!-\!1)+L(4(M_q\!-\!L)^2\!-\!1)}{3(M_q+L)}\cdot \zeta ^2,
\end{equation}
yielding the second line in \eqref{alpha_branches}.

\subsubsection{$M_p$-\ac{psk}}
For these constellations, we have
$$
\mathbb{E}\{e^{\lambda\rv{z}_n}\} = \frac{1}{M_p}\sum_{m=0}^{M_p-1} e^{\lambda\nu\cos \left(\frac{2\pi m}{M_p}+\theta\right)},
$$
where $\nu = \sqrt{\rv{r}_n^2+\rv{i}_n^2}=\frac{1}{\sqrt{N\zeta_{\rm ACF}}}$ and $\theta = 2\pi (k-1)(n-1)/N$. Using the Fourier expansion of $e^{a\cos x}$, we obtain
\begin{align}
\mathbb{E}\{e^{\lambda\rv{z}_n}\} &= \frac{2}{M_p}\sum_{m=0}^{M_p-1}\sum_{\ell=1}^{\infty}I_{\ell}(\lambda\nu)\cos\left(\ell \left(\frac{2\pi m}{M_p}+\theta\right)\right) \nonumber\\
&\hspace{3mm}+ I_0(\lambda\nu),
\end{align}
where $I_{\ell}(\cdot)$ denotes the $\ell$-th order modified Bessel function of the first kind. Changing the order of summation, we obtain
\begin{align}
&\sum_{m=0}^{M_p-1}\sum_{\ell=1}^{\infty}I_{\ell}(\lambda\nu)\cos\left(\ell \left(\frac{2\pi m}{M_p}+\theta\right)\right) \nonumber \\
&\hspace{3mm} = \sum_{\ell=1}^{\infty}I_{\ell}(\lambda\nu)\sum_{m=0}^{M_p-1}\cos\left(\ell \left(\frac{2\pi m}{M_p}+\theta\right)\right)\nonumber \\
&\hspace{3mm} = M_p\sum_{q=1}^{\infty} I_{q M_p}(\lambda\nu)\cos (qM_p\theta) \nonumber \\
&\hspace{3mm} \leq M_p\sum_{q=1}^{\infty} I_{q M_p}(\lambda\nu).
\end{align}
We thus have
\begin{equation}
\mathbb{E}\{e^{\lambda\rv{z}_n}\} \leq I_0(\lambda\nu)+2\sum_{q=1}^{\infty} I_{q M_p}(\lambda\nu) := \gamma_{M_p}.
\end{equation}
Next, since $I_{\alpha}(x)$ is a decreasing function of $\alpha$ \cite[Sec.~10.37]{handbook}, we see that for any $M_p\geq 4$, $I_{q M_p}(\lambda\nu)\leq I_{4q}(\lambda\nu)$, and hence
\begin{equation}
\mathbb{E}\{e^{\lambda\rv{z}_n}\} \leq \gamma_4,~\forall M_p\geq 4.
\end{equation}
But $\gamma_4$ is exactly the moment generating function in the case of \ac{qpsk}, which is also 4-\ac{qam}, and has already been shown to yield strictly sub-Gaussian $\rv{z}_n$ in our previous \ac{qam}-related discussions. Thus it now suffices to show that all $M_p$-\ac{psk} ($M_p\geq 4$) constellations have the same variance as that of \ac{qpsk}, which turns out to be true since
\begin{equation}\label{psk_var}
\mathbb{E}\{\rv{z}_n^2\} \!=\! \frac{1}{M_p}\sum_{m=0}^{M_p-1}\nu^2\cos^2\left(\frac{2\pi m}{M_p}\!+\!\theta\right) \!=\! \frac{1}{2}\nu^2=\frac{1}{2N\zeta_{\rm ACF}}.
\end{equation}
Substituting \eqref{psk_var} into \eqref{outage_prob} yields the third line of \eqref{alpha_branches}.

\subsubsection{8-\ac{qam}}
The 8-\ac{qam} constellation can be viewed as the Cartesian product of a 2-\ac{pam} constellation (corresponding to $\rv{i}_n$) and a 4-\ac{pam} constellation (corresponding to $\rv{r}_n$). According to \cite{pam_subgaussian}, all \ac{pam} constellations are strictly sub-Gaussian, and hence we have
\begin{align}
\mathbb{E}\{e^{\lambda\rv{z}_n}\} &= \mathbb{E}\{e^{\lambda a_n\rv{r}_n}\}\mathbb{E}\{e^{\lambda b_n\rv{i}_n}\} \nonumber \\
&\leq \exp\left(\frac{\lambda^2(a_n^2\mathbb{E}\{\rv{r}_n^2\}+b_n^2\mathbb{E}\{\rv{i}_n^2\})}{2}\right)
\end{align}
since $\rv{r}_n$ and $\rv{i}_n$ are mutually independent for 8-\ac{qam}. Because 8-\ac{qam} is not square, we apply an upper bound that does not depend on $a_n$ or $b_n$, as follows
\begin{align}
\mathbb{E}\{e^{\lambda\rv{z}_n}\} & \leq \exp\left(\frac{\lambda^2(a_n^2+b_n^2)\mathbb{E}\{\rv{r}_n^2\}}{2}\right) \nonumber \\
&=\exp\left(\frac{\lambda^2\mathbb{E}\{\rv{r}_n^2\}}{2}\right),
\end{align}
which yields the third line in \eqref{alpha_branches}, since
$$
\mathbb{E}\{\rv{r}_n^2\} = \frac{\frac{1}{2}(1^2+3^2)}{(1^2+3^2)N\zeta_{\rm ACF}} = \frac{1}{2N\zeta_{\rm ACF}}.
$$
This completes the proof.
\end{IEEEproof}

\section{Proof of Proposition \ref{prop:nominal_pacf}}\label{sec:proof_nominal_pacf}
\begin{IEEEproof}
Substituting \eqref{normalization} into \eqref{acf} and computing the squared modulus yields the expressions of the squared nominal \ac{pacf}, as given by \eqref{squared_pacf}. According to the definition of \ac{psl} and \ac{esl} in \eqref{define_psl}, we can obtain
\begin{equation}
\ac{psl} \!=\! \max_{i \in \Set{S}_{\rm ACF}}\left\{|[\RV{x}_{\rm{ACF}}]_i|^2\right\} \!=\! \frac{\max_{\rv{x}_{\rm{s}}\in {\Set{S}_{\rm{s}}}}|\rv{x}_{\rm{s}}|^2}{N\beta_{\rm{ACF}}\zeta_{\rm{ACF}}} \!=\! \frac{1}{N \zeta_{\rm{ACF}}},
\end{equation}
and
\begin{subequations}
\begin{align*}
[\ac{esl}]_i &= \mathbb{E}\left\{|[\RV{x}_{\rm{ACF}}]_i|^2\right\} = \frac{\mathbb{E}\left\{|\rv{x}_{\rm{s}}|^2\right\} }{N\beta_{\rm{ACF}}\zeta_{\rm{ACF}}} \\
&= \frac{1}{N \zeta_{\rm{ACF}}} \cdot \frac{\mathbb{E}\left\{|\rv{x}_{\rm{s}}|^2\right\}}{\max_{\rv{x}_{\rm{s}}\in {\Set{S}_{\rm{s}}}}|\rv{x}_{\rm{s}}|^2}. \tag{\theequation}
\end{align*}
\end{subequations}
This completes the proof.
\end{IEEEproof}

\section{Proof of Proposition \ref{prop:pacf_real}}\label{sec:proof_pacf_real}
\begin{IEEEproof}
We revisit the expression \eqref{pacf}
\begin{equation}
\RV{r}_{\RV{x}}=\frac{1}{\sqrt{N}}\M{F}_N^{\rm{H}}\left|\sqrt{N}\M{F}_N\RV{x}_{\rm{ACF}}\right|.
\end{equation}
Considering the power spectral non-negativity violation, the absolute value operation is given by
\begin{equation}
\left[\left|\sqrt{N}\M{F}_N\RV{x}_{\rm{ACF}}\right|\right]_{k}=
\begin{cases}
\left[\sqrt{N}\M{F}_N\RV{x}_{\rm{ACF}}\right]_{k}, \quad k \notin \RS{K},\\[5pt]
-\left[\sqrt{N}\M{F}_N\RV{x}_{\rm{ACF}}\right]_{k}, \quad k \in \RS{K}.
\end{cases}
\end{equation}
where $\RS{K}$ represents the set of all frequency bins $k$ that violate the power spectral non-negativity
constraint, namely $[\M{F}_N{\RV{x}}_{\rm{ACF}}]_{k}< 0$.
Therefore, the \ac{pacf} can be reformulated in the following generalized form
\begin{subequations}
\begin{align}
\RV{r}_{\RV{x}}=\frac{1}{\sqrt{N}}\M{F}_N^{\rm{H}}\left(\sqrt{N}\M{F}_N\RV{x}_{\rm{ACF}}+N\RV{e}\right)=\RV{x}_{\rm{ACF}}+\sqrt{N}\M{F}_N^{\rm{H}}\RV{e}.
\end{align}
\end{subequations}
where $\RV{e}$ denotes the error vector whose entries at frequency bins $k\in\RS{K}$ are given by $\frac{2\left|[\M{F}_N{\RV{x}}_{\rm{ACF}}]_{k}\right|}{\sqrt{N}}$ and others are zero. This completes the proof.
\end{IEEEproof}

\section{Proof of Corollary \ref{prop:violation_prob_corollary}} \label{sec:proof_violation_prob_corollary}
\begin{IEEEproof}
Note that we have
\begin{align*}
\left[\rv{e}_k>\delta\right]&\Leftrightarrow \left[[\M{F}_N\RV{x}_{\rm{ACF}}]_{k}<0\right]\land\left[\frac{2\left|[\M{F}_N\RV{x}_{\rm{ACF}}]_{k}\right|}{\sqrt{N}}>\delta\right] \\
&\Leftrightarrow \left[[\M{F}_N\RV{x}_{\rm{ACF}}]_{k}<-\frac{\delta\sqrt{N}}{2}\right].
\end{align*}
Analogous to the proof of Proposition \ref{prop:violation_prob}, we can bound the outage probability as follows
\begin{subequations}
\begin{align}
\mathbb{P}\left\{\rv{e}_k>\delta\right\}&=\mathbb{P}\left\{[\M{F}_N\RV{x}_{\rm{ACF}}]_{k}<-\frac{\delta\sqrt{N}}{2}\right\} \label{neg}\\
&=\mathbb{P}\left\{\widetilde{\V{f}}_k^{\rm T}\widetilde{\RV{x}}_{\rm{ACF}}<-\frac{1}{\sqrt{N}}-\frac{\delta\sqrt{N}}{2}\right\}\\
&=\mathbb{P}\left\{\sum_{n=2}^{\overline{N}} \rv{z}_n<-\frac{1}{2}-\frac{N\delta}{4}\right\}\\
&\leq \exp\left(-\frac{\left(\frac{1}{2}+\frac{N\delta}{4}\right)^2}{KN}\right)\label{sub}\\
& = \exp\left(-\alpha\zeta_{\rm ACF}{\left(1+\frac{N\delta}{2}\right)^2}\right)
\end{align}
\end{subequations}
where \eqref{neg} follows from the power spectral non-negativity constraint violation, namely $[\M{F}_N{\RV{x}}_{\rm{ACF}}]_{k}< 0$, while \eqref{sub}  is derived from the $K$-sub-Gaussian property of each $\rv{z}_n$ and the Chernoff inequality. This completes the proof.
\end{IEEEproof}

\section{Proof of Proposition \ref{prop:psl_prop}}\label{sec:proof_psl_prop}
\begin{IEEEproof}
Since $g(\rv{e}_k)$ is an increasing function of $\rv{e}_k$, we can obtain the following inequality
\begin{equation}
\ac{pslr}\leq g(\rv{e}_k)\leq g(\gamma),
\end{equation}
under the condition $\rv{e}_k\leq\gamma$. Under the assumption $|\RS{K}|\leq 1$, we first reformulate this line of reasoning in the logical symbolic representation \cite{rosen1999discrete}, given by
\begin{equation}
\forall k \left[k \in \RS{K} \rightarrow \rv{e}_k\leq\gamma \right]\rightarrow \left\{\ac{pslr}\leq g(\gamma) \right\}.
\end{equation}
From the set-theoretic perspective, the logical implication can be interpreted as the set inclusion relation, implying that
\begin{equation}
\left\{\forall k \left[k \in \RS{K} \rightarrow \rv{e}_k\leq\gamma \right]\right\} \subset \left\{\ac{pslr}\leq g(\gamma) \right\}.
\end{equation}
Based on the monotonicity of probability measures, we have
\begin{equation}\label{monotonicity}
\mathbb{P}\left\{\forall k \left[k \in \RS{K} \rightarrow \rv{e}_k\leq\gamma \right]\right\} \leq \mathbb{P}\left\{\ac{pslr}\leq g(\gamma) \right\}.
\end{equation}
Next, we aim to derive a lower bound for the left-hand side of the inequality. Note that
\begin{subequations}
\begin{align}
\forall k \left[k \in \RS{K} \rightarrow \rv{e}_k\leq\gamma \right] &\Leftrightarrow \forall k \left[\lnot [k \in \RS{K}] \lor \rv{e}_k\leq\gamma \right]  \label{implication} \\
&\Leftrightarrow \lnot \exists k \lnot \left[\lnot [k \in \RS{K}] \lor \rv{e}_k\leq\gamma \right] \label{quantifier} \\
&\Leftrightarrow \lnot \exists k \left[k \in \RS{K} \land \rv{e}_k>\gamma \right].
\label{De_Morgan}
\end{align}
\end{subequations}
where \eqref{implication} follows from the definition of implication in terms of negation and disjunction, \eqref{quantifier} follows from the quantifier negation rule and \eqref{De_Morgan} follows from De Morgan’s laws. Therefore, we obtain
\begin{equation}\label{lnot}
\mathbb{P}\left\{\forall k \left[k \in \RS{K} \rightarrow \rv{e}_k\leq\gamma \right]\right\} = 1-\mathbb{P}\left\{\exists k \left[k \in \RS{K} \land \rv{e}_k>\gamma \right]\right\}.
\end{equation}
Since the logical conjunction can be viewed as the set intersection operation, the following inequality is obtained by the monotonicity of probability measures
\begin{equation} \label{intersection}
\mathbb{P}\left\{\exists k \left[k \in \RS{K} \land \rv{e}_k>\gamma \right]\right\}\leq\mathbb{P}\left\{\exists k \left[\rv{e}_k>\gamma \right]\right\}.
\end{equation}
Note that the existential quantifier can be expressed as a logical disjunction form over a finite domain of discourse, and hence we have
\begin{equation}
\exists k \left[\rv{e}_k>\gamma \right] \Leftrightarrow [\RV{e}_1>\gamma]\lor[\RV{e}_2>\gamma]\lor\cdots\lor[\RV{e}_N>\gamma]
\end{equation}
Considering that the logical disjunction can be regarded as the set union operation, we further obtain
\begin{subequations}
\begin{align}
\mathbb{P}\left\{\exists k \left[\rv{e}_k\!>\!\gamma \right]\right\} &\!=\! \mathbb{P}\left\{\bigcup_{k=1}^{N}\{\rv{e}_k\!>\!\gamma\}\right\} \!\leq\! \sum_{k=1}^N\mathbb{P}\left\{\rv{e}_k\!>\!\gamma\right\} \label{union} \\
&\!\leq\! N e^{-\alpha\zeta_{\rm ACF}{\left(1+\frac{N\gamma}{2}\right)^2}}. \label{result}
\end{align}
\end{subequations}
where \eqref{union} is derived from the union bound, while \eqref{result} follows from Corollary \ref{prop:violation_prob_corollary}. Finally, combining \eqref{monotonicity}, \eqref{lnot}, \eqref{intersection} and \eqref{result} yields
\begin{subequations}
\begin{align*}
\mathbb{P}\left\{\ac{pslr}\leq g(\gamma) \right\} &\geq \mathbb{P}\left\{\forall k \left[k \in \RS{K} \rightarrow \rv{e}_k\leq\gamma \right]\right\} \\
&\geq 1-N e^{-\alpha\zeta_{\rm ACF}{\left(1+\frac{N\gamma}{2}\right)^2}}. \label{psl} \tag{\theequation}
\end{align*}
\end{subequations}
Now, relaxing the assumption that $|\RS{K}|\leq 1$, we obtain the looser bound
\begin{equation}
\mathbb{P}\left\{\ac{pslr}\leq g(\gamma) \right\} \geq 1-\mathbb{P}\{|\RS{K}|\geq 2\}-N e^{-\alpha\zeta_{\rm ACF}{\left(1+\frac{N\gamma}{2}\right)^2}}.
\end{equation}
This completes the proof.
\end{IEEEproof}

\section{Proof of Proposition \ref{prop:afkpsl}}\label{sec:proof_afkpsl}
\begin{IEEEproof}
According to the definition of \ac{psl} and \ac{esl}, we can obtain
\begin{equation}
\ac{psl} = \max_{(k,l) \in \Set{S}_{\rm AF}}\left\{|[\RM{X}_{\rm{AF}}]_{k,l}|^2\right\} = \frac{1}{MN \zeta_{\rm{AF}}},
\end{equation}
and
\begin{equation}
[\ac{esl}]_{k,l} \!=\! \mathbb{E}\left\{|[\RM{X}_{\rm{AF}}]_{k,l}|^2\right\} \!=\! \frac{1}{MN \zeta_{\rm{AF}}} \cdot \frac{\mathbb{E}\left\{|\rv{x}_{\rm{s}}|^2\right\}}{\max_{\rv{x}_{\rm{s}}\in {\Set{S}_{\rm{s}}}}|\rv{x}_{\rm{s}}|^2},
\end{equation}
where $\Set{S}_{\rm AF}$ denotes the set of all the sidelobe bins of \ac{fstaf}. Using the \ac{2dfft}, the \ac{fstaf} is represented as
\begin{subequations}
\begin{align}
\RM{R}_{\RM{X}} &= \frac{1}{\sqrt{MN}}\M{F}_N^{\rm{H}}|\RM{X}|^2\M{F}_M,\\
& = \frac{1}{\sqrt{MN}}\M{F}_N^{\rm{H}}\left|\RM{X}_{\rm{p}}\odot\sqrt{\sqrt{MN}\M{F}_N\RM{X}_{\rm{AF}}\M{F}_M^{\rm{H}}}\right|^2\M{F}_M\\
& = \frac{1}{\sqrt{MN}}\M{F}_N^{\rm{H}}\left(|\RM{X}_{\rm{p}}|^2\odot\left|\sqrt{MN}\M{F}_N\RM{X}_{\rm{AF}}\M{F}_M^{\rm{H}}\right|\right)\M{F}_M
\end{align}
\end{subequations}
Using the constant modulus of \ac{psk} constellations, we have
\begin{equation}
\RM{R}_{\RM{X}} = \frac{1}{\sqrt{MN}}\M{F}_N^{\rm{H}}\left|\sqrt{MN}\M{F}_N\RM{X}_{\rm{AF}}\M{F}_M^{\rm{H}}\right|\M{F}_M.
\end{equation}
Under the assumption that the spectral non-negativity constraint is satisfied, we further obtain
\begin{equation}
\RM{R}_{\RM{X}} = \frac{1}{\sqrt{MN}}\M{F}_N^{\rm{H}}\left(\sqrt{MN}\M{F}_N\RM{X}_{\rm{AF}}\M{F}_M^{\rm{H}}\right)\M{F}_M=\RM{X}_{\rm{AF}}.
\end{equation}
This completes the proof.
\end{IEEEproof}

\bibliographystyle{IEEEtran}
\bibliography{acfk}



\end{document}

%% file: SupportDocuments/YFPreamble.tex
\usepackage{amssymb}
\usepackage{overpic}
\usepackage{algorithm}
\usepackage{algorithmic}
\usepackage{array}
\usepackage{graphicx}
\usepackage{color}
\usepackage{epstopdf}
\usepackage{acronym}
\usepackage{enumerate}

\definecolor{BLUE}{rgb}{0,0,1}

\acrodef{aoa}[AOA]{angle-of-arrival}
\acrodef{acf}[ACF]{autocorrelation function}
\acrodef{bcrb}[BCRB]{Bayesian Cram\'{e}r-Rao bound}
\acrodef{bp}[BP]{belief propagation}
\acrodef{cdi}[CDI]{cooperative dilution intensity}
\acrodef{cl}[CL]{cooperative localization}
\acrodef{cdf}[CDF]{cumulative distribution function}
\acrodef{crb}[CRB]{Cram\'{e}r-Rao bound}
\acrodef{crlb}[CRLB]{Cram\'{e}r-Rao lower bound}
\acrodef{dof}[DoF]{degree of freedom}
\acrodef{dct}[DCT]{discrete cosine transform}
\acrodef{dpeb}[DPEB]{directional position error bound}
\acrodef{fim}[FIM]{Fisher information matrix}
\acrodef{efim}[EFIM]{equivalent Fisher information matrix}
\acrodef{icrb}[ICRB]{inverse CRB}
\acrodef{iid}[i.i.d.]{independently and identically distributed}
\acrodef{isac}[ISAC]{integrated sensing and communication}
\acrodef{mse}[MSE]{mean-squared error}
\acrodef{pdf}[PDF]{probability density function}
\acrodef{peb}[PEB]{position error bound}
\acrodef{speb}[SPEB]{squared position error bound}
\acrodef{pll}[PLL]{phase-locked loop}
\acrodef{rbs}[RBS]{reference broadcast synchronization}
\acrodef{rhs}[RHS]{right hand side}
\acrodef{rii}[RII]{ranging information intensity}
\acrodef{rss}[RSS]{received signal strength}
\acrodef{rc}[RC]{ranging coefficient}
\acrodef{toa}[TOA]{time-of-arrival}
\acrodef{tdoa}[TDOA]{time-difference-of-arrival}
\acrodef{tpsn}[TPSN]{time synchronization protocol for sensor network}
\acrodef{vmp}[VMP]{variational message passing}
\acrodef{wsn}[WSN]{wireless sensor network}
\acrodef{dio}[DIO]{distance-information-only}
\acrodef{aio}[AIO]{angle-information-only}
\acrodef{saaf}[SAAF]{squared array aperture function}
\acrodef{snc}[S\&C]{sensing and communications}
\acrodef{uoa}[UOA]{uniformly oriented array}
\acrodef{rgg}[RGG]{random geometric graph}
\acrodef{rms}[RMS]{root-mean-square}
\acrodef{snr}[SNR]{signal-to-noise ratio}
\acrodef{eoc}[EoC]{efficiency of cooperation}
\acrodef{npi}[NPI]{nominal position information}
\acrodef{gnss}[GNSS]{global navigation satellite system}
\acrodef{mimo}[MIMO]{multiple-input multiple-output}
\acrodef{mcs}[MCS]{minimally constrained system}
\acrodef{zzb}[ZZB]{Ziv-Zakai bound}
\acrodef{wwb}[WWB]{Weiss-Weinstein lower bound}
\acrodef{nlos}[NLOS]{non-light-of-sight}
\acrodef{mmse}[MMSE]{minimum mean squared error}
\acrodef{uav}[UAV]{unmanned aerial vehicle}
\acrodef{ppp}[PPP]{Poisson point process}
\acrodef{bpp}[BPP]{binomial point process}
\acrodef{cln}[CLN]{cooperative location-aware network}
\acrodef{pdr}[PDR]{pedestrian dead reckoning}
\acrodef{ml}[ML]{maximum likelihood}
\acrodef{map}[MAP]{maximum \textit{a posteriori}}
\acrodef{ofdm}[OFDM]{orthogonal frequency division multiplexing}
\acrodef{otfs}[OTFS]{orthogonal time frequency space}
\acrodef{sc}[SC]{single carrier}
\acrodef{cp}[CP]{cyclic prefix}
\acrodef{pacs}[P-ACS]{periodic autocorrelation sequence}
\acrodef{eisl}[EISL]{expected integrated sidelobe level}
\acrodef{dft}[DFT]{discrete Fourier transform}
\acrodef{6g}[6G]{sixth-generation}
\acrodef{drt}[DRT]{deterministic-random tradeoff}
\acrodef{psk}[PSK]{phase shift keying}
\acrodef{qpsk}[QPSK]{quadrature phase shift keying}
\acrodef{bpsk}[BPSK]{binary phase shift keying}
\acrodef{ask}[ASK]{amplitude shift keying}
\acrodef{apsk}[APSK]{amplitude phase shift keying}
\acrodef{qam}[QAM]{quadrature amplitude modulation}
\acrodef{pam}[PAM]{pulse amplitude modulation}
\acrodef{ba}[BA]{Blahut-Arimoto}

\acrodef{acfk}[ACFK]{auto-correlation function keying}
\acrodef{afk}[AFK]{ambiguity function keying}
\acrodef{af}[AF]{ambiguity function}
\acrodef{acf}[ACF]{auto-correlation function}
\acrodef{pacf}[P-ACF]{periodic auto-correlation function}
\acrodef{ber}[BER]{bit error rate}
\acrodef{ser}[SER]{symbol error rate}
\acrodef{dpaf}[DP-AF]{discrete periodic ambiguity function}
\acrodef{fstaf}[FST-AF]{fast-slow-time ambiguity function}
\acrodef{fst}[FST]{fast-slow-time}
\acrodef{ici}[ICI]{inter-carrier interference}
\acrodef{esl}[ESL]{expected sidelobe level}
\acrodef{psl}[PSL]{peak sidelobe level}
\acrodef{pslr}[PSLR]{peak sidelobe level ratio}
\acrodef{isi}[ISI]{inter-symbol interference}
\acrodef{zf}[ZF]{zero forcing}
\acrodef{llr}[LLR]{log-likelihood ratio}
\acrodef{awgn}[AWGN]{additive white Gaussian noise}
\acrodef{sir}[SIR]{signal-to-interference ratio}
\acrodef{cfar}[CFAR]{constant false-alarm rate}
\acrodef{cacfar}[CA-CFAR]{cell-averaging constant false-alarm rate}
\acrodef{tx}[Tx]{transmitter}
\acrodef{rx}[Rx]{receiver}
\acrodef{lti}[LTI]{linear time-invariant}
\acrodef{mf}[MF]{matched-filtering}
\acrodef{mgf}[MGF]{moment generating function}
\acrodef{pcs}[PCS]{probabilistic constellation shaping}
\acrodef{pas}[PAS]{probabilistic amplitude shaping}
\acrodef{ldpc}[LDPC]{low-density parity-check}
\acrodef{ccdm}[CCDM]{constant composition distribution matching}
\acrodef{roc}[ROC]{receiver operating characteristic}
\acrodef{cdf}[CDF]{cumulative distribution function}
\acrodef{csirs}[CSI-RS]{channel state information reference signals}
\acrodef{dmrs}[DMRS]{demodulation reference signals}
\acrodef{zcz}[ZCZ]{zero-correlation zone}
\acrodef{lcz}[LCZ]{low-correlation zone}
\acrodef{2dfft}[2D-FFT]{two-dimensional fast Fourier transform}

\acrodefplural{dof}[DoFs]{degrees of freedom}
